\documentclass[a4paper,11pt]{article}
\usepackage[utf8]{inputenc}
\usepackage[english]{babel}
\usepackage[utf8]{inputenc}

\usepackage{amssymb}
\usepackage{amsmath}
\usepackage{mathrsfs}

\usepackage{amsthm}
\usepackage{comment}

\usepackage[hmargin={2cm,2cm},vmargin={2cm,2cm}]{geometry}

\usepackage{tikz}
\usetikzlibrary{positioning}
\usepackage{tikz-cd}

\usepackage[colorlinks=true,linkcolor=black]{hyperref}
\newtheorem{theorem}{Theorem}[section]
\newtheorem{corollary}{Corollary}

\newtheorem{lemma}[theorem]{Lemma}
\newtheorem{proposition}{Proposition}

\theoremstyle{definition} 

\newtheorem{definition}[theorem]{Definition}
\newtheorem{remark}{Remark}

\numberwithin{equation}{section}
\numberwithin{proposition}{section}
\numberwithin{example}{section}
\usepackage{bm}
\usepackage{tikz}
\usepackage{tikz-cd}
\def\ho{\widehat{\omega}}

\numberwithin{equation}{section}

\newtheorem*{Theorem*}{Theorem}

 { \theoremstyle{definition}

 }
\def\om{\omega}

\def\>{\rangle}
\def\<{\langle}
\def\({\left (}
\def\){\right )}
\def\x{\times}

\def\B#1{\mathbf{#1}}

\def\tl#1{\tilde{#1}}

\newcommand{\co}{[\![}
\newcommand{\cc}{]\!]}

\newcommand{\Ad}{{\rm Ad}}

\newcommand{\lv}{\langle}
\newcommand{\rv}{\rangle}

\numberwithin{Theorem}{section}
\begin{document} 

\centerline{\Large \bf Cosymplectic geometry, reductions, and energy-}

\vskip 0.35cm
\centerline{\Large \bf momentum methods with applications}
\vskip 0.25cm

\centerline{J. de Lucas$^*$, A. Maskalaniec$^{**}$, and B.M. Zawora$^{**}$}
\vskip 0.35cm

\centerline{$^{*}$Centre de Recherches Math\'ematiques, Universit\'e de Montr\'eal}

\centerline{Pavillon André-Aisenstadt,
2920 Tour Rd, Montreal, Quebec H3T 1N8,
Canada}

\centerline{$^{**}$ UW Institute for Advanced Studies}

\centerline{Department of Mathematical Methods in Physics, University of Warsaw }   

\centerline{   ul. Pasteura 5, 02-093, Warsaw, Poland}

\begin{abstract}Classical energy-momentum methods study the existence and stability properties of solutions of $t$-dependent Hamilton equations on symplectic manifolds whose evolution is given by their Hamiltonian Lie symmetries. The points of such solutions are called \textit{relative equilibrium points}. This work devises a new cosymplectic energy-momentum method providing a new and more general framework to study $t$-dependent Hamilton equations. In fact, cosymplectic geometry allows for using more types of distinguished Lie symmetries (given by Hamiltonian, gradient, or evolution vector fields), relative equilibrium points, and reduction methods, than symplectic techniques. To make our work more self-contained and to fill some gaps in the literature, a review of the cosymplectic formalism and the cosymplectic Marsden--Weinstein reduction is included.  Known and new types of relative equilibrium points are characterised and studied. Our methods remove technical conditions used in previous energy-momentum methods, like the ${\rm Ad}^*$-equivariance of momentum maps. Eigenfunctions of $t$-dependent Schr\"odinger equations are interpreted in terms of relative equilibrium points in cosymplectic manifolds. A new cosymplectic-to-symplectic reduction is developed and a new associated type of relative equilibrium points, the so-called \textit{gradient relative equilibrium points}, are introduced and applied to study the Lagrange points and Hill radii of a restricted circular three-body system by means of a not Hamiltonian Lie symmetry of the system. 
\end{abstract}

{\bf Keywords:} cosymplectic geometry, cosymplectic reduction,  energy-momentum method, relative equilibrium point; restricted circular three-body problem, $t$-dependent Schr\"odinger equation.

{\bf MSC 2020:} {{34A26, 37J39} (primary) {34A05, 70H05, 70H14} (secondary)}

\section{Introduction}

Symplectic geometry has been successfully applied to the description of mechanical systems in physics \cite{Ar89, GS90}. The first works concerning this topic can be traced back to Lagrange and Poisson, who analysed the motion of rigid bodies and celestial mechanics \cite{La09, Ma09}. From the XXth century, the Marsden--Weinstein reduction theorem \cite{MW74} has played a crucial role in describing Hamiltonian systems on symplectic manifolds admitting a Lie group of symmetries of the Hamiltonian and the symplectic manifold. This has had a profound impact on the analysis of time-independent mechanical systems and led to numerous generalisations \cite{AM78, BS97, CZ13, MRSV15, MR86}.

The description of time-dependent mechanical systems cannot be directly approached using symplectic geometry (cf. \cite{AM78, Al89, CNY13}). Instead, it is possible to modify symplectic geometry to cope with them. For instance, one can use cosymplectic geometry \cite{Al89, CNY13, Li59, Li62} or to deal with time-parametrised Hamiltonians in  symplectic geometry  \cite{LZ21, Za21}. As illustrated by the applications in this paper, considering the time as a coordinate in cosymplectic geometry allows for new techniques that cannot be properly defined in time-parametrised symplectic frameworks \cite{AM78, LZ21, Za21}. 

In the cosymplectic scenario, time-dependent ($t$-dependent) Hamiltonian systems are described via a closed differential two-form $\omega$ and a closed non-vanishing one-form $\eta$, both on a manifold $M$, so that $\ker\omega\oplus\ker\eta=TM$. Hence, $M$ is odd-dimensional, while $(M,\omega,\eta)$ is called a {\it cosymplectic manifold}. Although some reviews on cosymplectic geometry can be found 
\cite{Al89, CNY13}, proofs of certain results may be difficult or even impossible to find. For instance, a Marsden--Weinstein cosymplectic reduction was developed in 
\cite{Al89}, but that work is not written in English and the only, as far as we know, available online manuscript of \cite{Al89} is a low-quality blurry scan of the original paper. Moreover, other pioneering works on cosymplectic geometry are not available in English \cite{Li62, Li75}. Posterior works, like \cite{CNY13, LS93}, frequently  use facts devised in  
\cite{Al89, Li62} without giving proofs. In this sense, this work fulfils a gap in the literature providing results that can be hardly available from other sources.

Given a cosymplectic manifold $(M,\omega,\eta)$, every function $f\in C^\infty(M)$ leads to three vector fields on $M$: a Hamiltonian vector field $X_f$; a gradient vector field $\nabla f$; and an evolution vector field $E_f$. The Hamilton equations induced by a function $h\in C^\infty(M)$, which determine in particular the dynamical behaviour of a $t$-dependent mechanical system, are here described, among other methods, as the integral curves of an evolution vector field $E_h$ on a manifold  $M=T\times P$, where $T$ is a one-dimensional manifold and $P$ is a symplectic manifold \cite{AM78, LR89, LM87}. Due to our interest in physical systems, we mainly focus on cosymplectic manifolds related to manifolds $M$ of the latter type. This is slightly more general than the approach given in the time-dependent energy-momentum method in \cite{LZ21} devoted to the particular case  $M=\mathbb{R}\times P$, as our new approach is more appropriate, for instance, for describing $t$-dependent Hamiltonian systems that are periodic with respect to the time.  The assumption $M=T\times P$ imposes a mild restriction in the cosymplectic manifolds under study. Indeed, we are here mainly focused on local aspects of dynamical systems on cosymplectic manifolds, and every cosymplectic manifold is locally diffeomorphic to $M=\mathbb{R}\times P$. Moreover, many interesting physical systems can be effectively analysed in our main framework.

Let us briefly summarise the  cosymplectic Marsden--Weinstein reduction to justify the significance of the findings of our paper. Recall that a cosymplectic manifold $(M,\omega,\eta)$ induces a unique vector field, $R$, on $M$, the so-called {\it Reeb vector field}, such that the vector field contractions $\iota_R\omega=0$ and $\iota_R\eta=1$ are satisfied.  Let $\Phi: G\times M\rightarrow  M$ be a Lie group action whose fundamental vector fields are Hamiltonian relative to $(M,\omega,\eta)$, take values in $\ker \eta$, and admit a common first integral $h\in C^\infty(M)$. Then, a cosymplectic momentum map $\mathbf{J}^\Phi: M\rightarrow \mathfrak{g}^*$, where $\mathfrak{g}^*$ is dual to the Lie algebra $\mathfrak{g}$ of $G$, can be defined under the condition that its coordinates are first integrals of the Reeb vector field of $(M,\omega,\eta)$ and Hamiltonian functions of a basis of fundamental vector fields of $\Phi$. Previous structures give rise to what can be, in short, called a {\it cosymplectic Hamiltonian system}. Under different technical conditions and approaches, first Albert \cite{Al89} and afterwards de Le\'on and Saralegi \cite{LS93}, reduced the space of orbits, $M_\mu^\Delta:=\mathbf{J}^{\Phi-1}(\mu)/G_\mu^\Delta$, of the restriction of the action  $\Phi$ to a certain $G_\mu^\Delta\subset G$ for a weak regular value $\mu\in\mathfrak{g}^*$ of $\mathbf{J}^\Phi$ and $\mathbf{J}^{\Phi-1}(\mu)$ being quotientable, which guarantees the existence of a canonical cosymplectic manifold  $(M^\Delta_\mu,\omega_\mu,\eta_\mu)$. In our work, this result is proved and described in more detail than in \cite{Al89, LS93}.  For cosymplectic manifolds $(T\times P,
 \omega,\eta)$ of special types related to classical mechanical systems, it is proved that the cosymplectic Marsden--Weinstein reduction permits us to define time-like coordinates on $M^
 \Delta_\mu$, i.e. $M^\Delta_\mu= T\times P^\Delta_\mu$ for a certain manifold $P^\Delta_\mu$. In this case, ${\bf J}^\Phi$ is shown to be time-independent provided $T$ is connected.

Our cosymplectic Marsden--Weinstein reduction shows that the function $h\in C^\infty(M)$ induces a new one, $k_\mu\in C^\infty(M^\Delta_\mu)$, defined uniquely by the condition $k_\mu\circ\pi_\mu:=h\circ \iota_\mu$, where $\iota_\mu:{\bf J}^{\Phi-1}(\mu)\hookrightarrow M$ is the natural immersion of ${\bf J}^{\Phi-1}(\mu)$ in $M$, while $\pi_\mu:\mathbf{J}^{\Phi-1}(\mu)\rightarrow  M^\Delta_\mu$ is the quotient map. The dynamics given by the Hamiltonian vector field, $E_h$, induced by $h$ on $M$ gives rise to a new evolution vector field $E_{k_\mu}$ on $M^\Delta_\mu$. The previous result, absent in \cite{Al89}, generalises the approach to the cosymplectic reduction of $h$ given in \cite{LS93}  by skipping unnecessary technical assumptions, like those concerning the ${\rm Ad}^*$-invariance of momentum maps, and extends the result for other types of vector fields, like $X_h$ and $\nabla h$. 

It is worth stressing that the cosymplectic Marsden--Weinstein reduction requires that the fundamental vector fields of $\Phi$ take values in $\ker \eta$. In physical systems, this is here proved to amount to the fact that the Lie symmetries to perform the reduction do not include partial derivatives relative to the time. Although this may seem restrictive, it is plenty of physical systems, depending explicitly on time, that admit such Lie symmetries. A relevant example of this type is given in Section \ref{Sec::QuantumExample}, where time-dependent Schr\"odinger equations are analysed. One may also study types of almost-rigid bodies with a $t$-dependent tensor inertia \cite{LZ21, Go80, Za21} satisfying, for instance, a symmetry around a certain axis at every time.  

If $M=T\times P$, one has that $h$ may have the so-called {\it  relative equilibrium points} with respect to the Lie group action $\Phi: G\times M\rightarrow  M$, i.e. points $z_e\in P$ giving rise to equilibrium points of the form $(t,\pi_{\mu}(z_e))$ for every $t\in T$ and $\mu={\bf J}^\Phi(t,z_e)$ of the Hamiltonian vector field associated with $k_\mu$ on $M^\Delta_\mu$, which are the projection of not necessarily equilibrium points of $X_h$ on $M$. Equivalently, relative equilibrium points are points related to integral curves of $X_h$ that can be described via the Lie group action. Relative equilibrium points are here devised as a cosymplectic analogue of the relative equilibrium point notion appearing in the $t$-dependent energy-momentum method \cite{LZ21,Za21} and as a generalisation of the classical relative equilibrium point concept occurring in the classical energy-momentum method \cite{AM78, MS88}. Next, the existence of  relative equilibrium points is characterised and analysed by using cosymplectic geometry. In particular, their relation to our cosymplectic approach to Hamilton equations and the equilibrium points of reduced Hamilton equations via the cosymplectic Marsden--Weinstein reduction is analysed.

It is relevant to study the stability of  relative equilibrium points, namely whether particular solutions get closer or away from relative equilibrium points as they evolve. Our cosymplectic energy-momentum method has been devised to analyse this problem and other related ones. Our techniques allow for studying the cosymplectic Hamiltonian system given by $k_\mu$ on the reduced space $M^\Delta_\mu$ via the properties of the initial function $h$, which avoids the necessity of constructing $k_\mu$ and $M^\Delta_\mu$ explicitly. Our techniques expand to a cosymplectic realm the classical energy-momentum method developed by Simo and Marsden in \cite{MS88} and represents a geometric, and more powerful, alternative to the $t$-dependent energy-momentum method given in \cite{LZ21}, where the time coordinate appeared as a parameter instead of a coordinate with some geometrical meaning. For instance, this is due to the fact that considering the time as a new variable will allow for the use of structures that cannot be described by using the time as a parameter. For instance, cosymplectic geometry has types of distinguished vector fields and other types of reductions that are not present in symplectic geometry \cite{Al89}. These results are here proved to provide new ways of studying relative equilibrium points on cosymplectic manifolds.

In this work, we provide conditions ensuring the stability of equilibrium points of reduced, $t$-dependent, cosymplectic Hamiltonian systems obtained via cosymplectic Marsden--Weinstein reductions. Then, we also show how to determine the stability in the reduced cosymplectic Hamiltonian system using the Hessian of the Hamiltonian function on the initial cosymplectic manifold, $(\mathbb{R}\times P,\omega,\eta)$, instead of using only the reduced cosymplectic manifold, $(M^\Delta_\mu\simeq \mathbb{R}\times P^\Delta_\mu,\omega_\mu,\eta_\mu)$, which is, in general, more difficult.

As a new application, the case of an $n$-level quantum system described by a $t$-dependent Schr\"odinger equation is studied. Relative equilibrium points are characterised as eigenvectors for every value of $t$ of the $t$-dependent Hamiltonian operator describing the system. Moreover, the stability of the reduced systems is analysed geometrically. An introductory two-level system is investigated in detail to fully understand the methods of our theory in a particularly simple case. It is worth noting that the previous examples  are described by a time-dependent Hamiltonian and  have relevant time-independent Lie symmetries that allow us to perform a Marsden--Weinstein reduction and apply our methods.

Next, a new type of cosymplectic-to-symplectic reduction is devised. This reduction allows for the use of Lie symmetries that do not take values in $\ker \eta$, which permits one the study of problems that cannot be analysed through the cosymplectic-to-symplectic reduction devised by Albert \cite[pg. 
 640]{Al89}. In particular, we prove that Albert's reduction is a particular case of ours. Moreover, the vector fields of our reduction, which appear in restricted circular three-body problems, cannot be described in $t$-dependent symplectic geometry and previous energy-momentum methods. Our new reduction cannot be fully described via a Poisson reduction either because the dynamics of the studied systems is not generated via a Hamiltonian vector field relative to a Poisson bivector. Additionally, our reduction has special features due to its cosymplectic nature that require a particular approach. Next, gradient vector fields in cosymplectic manifolds are employed to devise a new type of relative equilibrium points, the hereafter called {\it gradient relative equilibrium points}.

Our last example focuses on the study of a gradient relative equilibrium points for a $t$-dependent restricted circular three-body problem \cite{AM78}. In this case, it is shown that the existence of a natural Lie symmetry that contains a time derivative makes the approach given in the time-dependent energy-momentum method in \cite{LZ21} impossible to apply so as to reduce the initial cosymplectic problem. Our new cosymplectic-to-symplectic reduction is employed to reduce the restricted circular three-body problem to a new Hamiltonian system on the manifold of the orbits of the given Lie symmetry.  Our brand new gradient relative equilibrium point suggests that energy-momentum methods must be generalised to deal with new types of Poisson reductions. Note that the standard energy-Casimir method slightly depicted in \cite{MS88}, which is concerned with relative equilibrium points and their stability analysis for Hamiltonian systems relative to Poisson manifolds, does not apply  to our restricted circular three-body problem as the vector field determining its dynamics is not Hamiltonian neither. Some details and results concerning this new method are given. 

The structure of the paper goes as follows. Section \ref{Sec::BasicSympl} describes some basic notions on symplectic and cosymplectic geometry, and introduces the notation to be used hereafter. It also stresses the relevance and advantages of cosymplectic geometry relative to other related geometric approaches. Section \ref{Sec::AMomentumMap} describes the theory of ${
\rm Ad}^*$-equivariant momentum maps on cosymplectic manifolds while Section \ref{Sec::AGeneralMomentumMap} extends such results to the theory of momentum maps that are not ${
\rm Ad}^*$-equivariant. Section \ref{Sec::CosymplReduction} presents a generalisation of the classical Marsden--Weinstein reduction theorem to the cosymplectic realm. Section \ref{Sec::Stab} recalls some fundamental results on Lyapunov stability on manifolds. Section \ref{Sec::CharStrongEqP} generalises to cosymplectic Hamiltonian systems the relative equilibrium point definition for Hamiltonian systems in symplectic geometry. Section \ref{Sec::ReducedStability} studies a type of reduced cosymplectic Hamiltonian systems to give conditions to guarantee the stability of the equilibrium points of their Hamilton equations. Section \ref{Sec::RelationofHessians} analyses the connection between the relative equilibrium points in some $\mathbf{J}^{\Phi-1}(\mu)$ and the associated equilibrium points of reduced cosymplectic Hamiltonian systems. Section \ref{Sec::QuantumExample} applies our results to a two-level quantum system. Section \ref{Sec:SREP} generalises the previous example to an $n$-level quantum system. Section \ref{Sec::GradRelEqPoints} introduces a new type of cosymplectic-to-symplectic reduction, where Lie symmetries are not Hamiltonian vector fields, and the gradient relative equilibrium point concept. Section \ref{Sec::RDTP} analyses an application of our cosymplectic-to-symplectic reduction to a restricted circular three-body problem and new types of relative equilibrium points. Finally,  Section \ref{Sec::Conclusions} summarises our results and presents an outlook about further research.

\section{Basics on symplectic and cosymplectic geometry}
\label{Sec::BasicSympl}

Let us set the general assumptions and the notation for the whole paper. This section also presents general results on symplectic (see \cite{AM78,Ca06} for details) and cosymplectic geometry \cite{CNY13, Li59, Li62} to be used hereafter. Some of them are difficult to find in the previous literature. If not otherwise stated, all structures are assumed to be smooth and globally defined, while manifolds are assumed to be connected, paracompact, and Hausdorff. These simplifications stress our key ideas and allow us to avoid minor or unnecessary technical problems.  Moreover, $\Omega^k(M)$ and $\mathfrak{X}(M)$ stand for the spaces of differential $k$-forms and vector fields on a manifold $M$, respectively. We write $V$ for a finite-dimensional vector space.

A {\it symplectic manifold} is a pair $(P,\omega)$, where $P$ is a manifold and $\omega$ is a closed differential two-form on $P$ that is {\it non-degenerate}, namely the unique vector bundle morphism $\omega^\flat:TP\rightarrow  T^*P$ such that $\omega^\flat(v_z):=\omega_z(v_z,\cdot)\in T_z^*P$ for every $v_z\in T_zP$ and $z\in P$, is a vector bundle isomorphism. We  call $\omega$ a {\it symplectic form}. From now on, $(P,\omega)$ will always stand for a symplectic manifold.

The {\it symplectic orthogonal} of a subspace $V_z\subset T_zP$, with $z\in P$, relative to $(P,\omega)$ is defined by
\[
V^{\perp_\omega}_z:=\{\vartheta_z\in T_zP\,:\,\omega_z(\vartheta_z,v_z)=0,\,\forall v_z\in V_z\}.
\]

A vector field $X\in \mathfrak{X}(P)$ is {\it Hamiltonian} if $\iota_{X}\omega=df$ for some $f\in C^{\infty}(P)$.
Then, $f$ is called a {\it Hamiltonian function} of $X$. Since $\omega$ is non-degenerate, every $f\in C^\infty(P)$ has a unique Hamiltonian vector field $X_f$. We write ${\rm Ham}(P,\omega)$ for the vector space of Hamiltonian vector fields on $P$ relative to the symplectic form $\omega$. The Cartan's magic formula yields $\mathcal{L}_{X_f}\omega=0$ for every $f\in C^\infty(M)$,
where $\mathcal{L}_{X_f}\omega$ is the Lie derivative of $\omega$ with respect to $X_f$. 

Let us define the bracket
\begin{equation}
\label{PoissoneBracket}
\{\cdot,\cdot\}_\omega: C^\infty(P)\times C^\infty(P)\ni (f,g)\mapsto \omega(X_f,X_g)\in C^\infty(P).
\end{equation}
This bracket is bilinear, antisymmetric, and, since $d\omega=0$, it obeys the {\it Jacobi identity}, which makes $\{\cdot,\cdot\}_\omega$ into a {\it Lie bracket}. Moreover, $\{\cdot,\cdot\}_\omega$ obeys the {\it Leibniz rule}, i.e. 
\[
\{f,gh\}_\omega=\{f,g\}_\omega h+g\{f,h\}_\omega,\qquad \forall f,g,h\in C^\infty(P).
\]
Such properties turn  $\{\cdot,\cdot\}_\omega$ into a {\it Poisson bracket}. It can be proved that  \begin{equation*}
\label{Eq::AntiMorphismSym}
X_{\{f,g\}_\omega}=-[X_f,X_g],\qquad \forall f,g\in C^\infty(P),
\end{equation*}
and ${\rm Ham}(P,\omega)$ becomes a Lie algebra. Moreover, the mapping $f\in C^\infty(P)\mapsto -X_f\in {\rm Ham}(P,\omega)$ is a Lie algebra morphism relative to the Lie bracket $\{\cdot,\cdot\}_\omega$ in $C^\infty(P)$ and the commutator of vector fields in $\mathfrak{X}(P)$. In general, a pair $(P,\{\cdot,\cdot\})$ is called a {\it Poisson manifold}, where $P$ is a manifold and $\{\cdot,\cdot\}$ is a Poisson bracket on $C^\infty(P)$, which may not necessarily be associated with some symplectic form $\omega$. Then, any Poisson bracket gives rise to a bivector field $\Lambda$ on $P$ defined as
\[
\Lambda(z)(\alpha_z,\beta_z):=\{f,g\}(z),
\]
where $(df)_z=\alpha_z$ and $(dg)_z=\beta_z$ are elements of $T^*_zP$ for some $f,g\in C^\infty(P)$. The bivector field $\Lambda$ is called a {\it Poisson bivector}.

Let $\tau: T^*Q\rightarrow  Q$ be the canonical cotangent bundle projection onto $Q$. The {\it canonical one-form} on $T^*Q$ is the differential one-form $\theta_Q$ on $T^*Q$ defined by 
\begin{equation*}
(\theta_Q)_{\alpha_q}(v_{\alpha_q}):=\langle \alpha_q,T_{\alpha_q}\tau(v_{\alpha_q})\rangle,\quad\forall q\in Q ,\quad\forall \alpha_q\in T^*_qQ,\quad \forall v_{\alpha_q}\in T_{\alpha_q}T^*Q,
\end{equation*}
where $\langle\cdot,\cdot\rangle$ is the pairing between covectors and vectors on $Q$. The {\it canonical symplectic form} on $T^*Q$ is the differential two-form on $T^*Q$ given by $\omega_{Q}:=-d\theta_Q$.
On local adapted coordinates $\{q^i,p_i\}_{i=1,\ldots,n}$ to $T^*Q$, one has $\theta_Q=\sum_{i=1}^n p_idq^i$. Then, $\omega_Q=-d\theta_Q=\sum^n_{i=1}dq^i\land dp_i$ is a symplectic form.

We hereafter assume $G$ to be a connected Lie group with a Lie algebra $\mathfrak{g}$. Let us set some definitions and conventions concerning Lie group actions. If $\Phi: G\times P\rightarrow  P$ is a Lie group action and $\Phi$ is known from the context or its explicit form is irrelevant, we will write $gz$ instead of $\Phi(g,z)$ for every $g\in G$ and each $z\in P$. The {\it isotropy subgroup} of $\Phi$ at $z\!\in\!P$ is the Lie subgroup of $G$ given by $G_z:=\{g\in G:gz=z\}\subset G$.
{\it The orbit} of a point $z\in P$ relative to $\Phi$ is defined as $G z:=\{gz\,:\,g\in G\}$. The orbits of $\Phi$ are immersed submanifolds in $P$ \cite{Le13}. The {\it fundamental vector field} of a Lie group action $\Phi:G\times P\rightarrow  P$ related to $\xi\in\mathfrak{g}$ is the vector field on $P$ given by
\begin{equation*}
(\xi_P)_z:=\frac{d}{ds}\bigg|_{s=0}\Phi(\exp(s\xi),z),\quad \forall z\in P,
\end{equation*}
where $\exp:\mathfrak{g}\rightarrow  G$ is the exponential map related to the Lie group $G$. Then, $T_{\tilde{z}}(Gz)=\{(\xi_P)_{\tilde{z}}:\xi\in\mathfrak{g}\}$ for each $\tilde{z}\in Gz$.

Every Lie group acts on itself by inner automorphisms $I:(g,g')\in G\times G\mapsto gg'g^{-1}\in G$. This induces a Lie group action of $G$ on its Lie algebra $\mathfrak{g}$ given by ${\rm Ad}:(g,v)\in G\times \mathfrak{g}\mapsto T_eI_g(v)\in \mathfrak{g}$ and its dual ${\rm Ad}^*:(g,\vartheta)\in G\times \mathfrak{g}^*\mapsto \vartheta\circ {\rm Ad}_{g^{-1}}\in \mathfrak{g}^*$, which is called the {\it co-adjoint action} of $G$.

An {\it almost precosymplectic manifold of rank $2r$} is a triple $(M,\omega,\eta)$, where $M$ is a $(2n+1)$-dimensional manifold, $\omega\in \Omega^2(M)$ has constant rank $2r$ and $\eta\in\Omega^1(M)$ is such that $\omega^r\wedge\eta$ does not vanish at any point of $M$. If $\omega$ and $\eta$ are additionally closed, then $(M,\omega,\eta)$ is called a {\it precosymplectic manifold} of rank $2r$. If, moreover, $r=n$, then $(M,\omega,\eta)$ is said to be a {\it cosymplectic manifold}.  Note that the fact that $\eta\wedge \omega^n$ does not vanish at any point of $M$ amounts to $\eta\wedge \omega^n$ being a volume form. Hence, cosymplectic manifolds are always orientable and odd-dimensional. 

The Darboux theorem for cosymplectic manifolds   \cite{Al89} states that, given a cosymplectic manifold $(M,\omega,\eta)$, each point $x\in M$ admits a local coordinate system $\{t,q^1,\ldots,q^n,p_1,\ldots,p_n\}$ on an open neighbourhood $U$ of $x$ so that
\[
\omega=\sum_{i=1}^ndq^i\wedge dp_i,\qquad \eta=dt
\]
on $U$. Such local coordinates are called {\it cosymplectic Darboux coordinates}, although we simply call them {\it Darboux coordinates} if it is clear from the context that they are related to a cosymplectic manifold. Note that cosymplectic Darboux coordinates are not unique.

Each cosymplectic manifold $(M,\omega,\eta)$ admits a unique vector field $R$ on $M$ such that
\begin{equation}\label{ReebCon}
\iota_R\,\omega=0,\qquad \iota_R\,\eta=1.
\end{equation}
We call $R$ the {\it Reeb vector field} of $(M,\omega,\eta)$. In Darboux coordinates for $(M,\omega,\eta)$, let us say $\{t,q^1,\ldots,q^n,p_1,\ldots,p_n\}$, the Reeb vector field reads $R=\frac{\partial}{\partial t}$.

A {\it cosymplectomorphism} is a map $\varphi:M_1\rightarrow  M_2$ between cosymplectic manifolds $(M_1,\omega_1,\eta_1)$ and $(M_2,\omega_2,\eta_2)$ such that $\varphi^*\omega_2=\omega_1$ and $\varphi^*\eta_2=\eta_1$. A {\it cosymplectic Lie group action} relative to $(M,\omega,\eta)$ is a Lie group action $\Phi:G\times M\rightarrow  M$ such that, for every $g\in G$, the map $\Phi_g:M\rightarrow  M$ is a cosymplectomorphism. In other words,
\begin{equation}
\label{Eq::ContactomorphismG}
\Phi_g^*\,\omega=\omega,\qquad \Phi_g^*\,\eta=\eta,\qquad  \forall g\in G.
\end{equation}
Since manifolds are assumed to be connected in this work, $\Phi:G\times M\rightarrow  M$ is a cosymplectomorphism if and only if
\begin{equation}
\label{Eq::Contactomorphism}
\mathcal{L}_{\xi_M}\omega=0,\qquad\mathcal{L}_{\xi_M}\eta=0,\qquad\forall\xi\in\mathfrak{g} .
\end{equation}
Moreover, since $d\eta=0$, the condition $\mathcal{L}_{\xi_M}\eta=0$ implies that $\iota_{\xi_M}\eta$ is a constant function on $M$, which does not need to be one or zero. This will be of interest afterwards to study the restricted circular three-body problem \cite{AM78}.

Given a cosymplectic manifold $(M,\omega,\eta)$, the vector bundle morphism
\[
\flat :TM\rightarrow  T^*M, \qquad v_x\in T_xM\mapsto \flat (v_x):=\iota_{v_x}\omega_x+(\iota_{v_x}\eta_x)\eta_x\in T_x^*M,\qquad \forall x\in M,
\]
is a vector bundle isomorphism. 

The other way around, given a closed $\widehat \omega\in\Omega^2(M)$ and a closed $\widehat \eta\in\Omega^1(M)$, then $\widehat\omega$ and $\widehat\eta$ give rise to a cosymplectic manifold $(M,\widehat\omega,\widehat\eta)$, if the map
\[
\widehat\flat :TM\rightarrow  T^*M, \qquad v_x\in T_xM\mapsto \widehat\flat (v_x):=\iota_{v_x}\widehat\omega_x+(\iota_{v_x}\widehat\eta_x)\widehat\eta_x\in T_x^*M,\qquad \forall x\in M,
\]
is a vector bundle isomorphism (cf. \cite{Al89}). Hereafter, $(M,\omega,\eta)$ will stand for a cosymplectic
manifold.

Given $(M,\omega,\eta)$, every $f\in C^\infty(M)$ gives rise to three vector fields:

\begin{itemize}
    \item A {\it gradient vector field}, namely
\begin{equation}
\label{Eq::GradVecField}
\nabla f:=\flat^{-1}(df),
\end{equation}
which amounts to saying that $\iota_{\nabla f}\omega=df-(Rf)\eta$ and $\iota_{\nabla f}\eta=Rf$.

    \item A {\it Hamiltonian vector field}, $X_f$, given by
\begin{equation}
\label{Eq::HamVecField}
    X_f:=\flat^{-1}(df-(Rf)\eta),
\end{equation}
    which is equivalent to $\iota_{X_{f}}\omega=df-(Rf)\eta$ and $\iota_{X_f}\eta=0$.
    \item An {\it evolution vector field} 
    \begin{equation}
    \label{Eq::EvVecField}
        E_f:=R+X_f.
    \end{equation}
\end{itemize}

In Darboux coordinates for $(M,\omega,\eta)$ around a point $x\in M$, the vector fields \eqref{Eq::GradVecField}, \eqref{Eq::HamVecField}, and \eqref{Eq::EvVecField} read
\[
\nabla f=\frac{\partial f}{\partial t}\frac{\partial}{\partial t}+\sum_{i=1}^n\left(\frac{\partial f}{\partial p_i}\frac{\partial}{\partial q^i}-\frac{\partial f}{\partial q^i}\frac{\partial}{\partial p_i}\right),\qquad
X_f=\sum_{i=1}^n\left(\frac{\partial f}{\partial p_i}\frac{\partial}{\partial q^i}-\frac{\partial f}{\partial q^i}\frac{\partial}{\partial p_i}\right),
\]
and
\[
E_f=\frac{\partial}{\partial t}+\sum_{i=1}^n\left(\frac{\partial f}{\partial p_i}\frac{\partial}{\partial q^i}-\frac{\partial f}{\partial q^i}\frac{\partial}{\partial p_i}\right).
\]

The integral curves of $E_f$ are given, in Darboux coordinates, by the solutions of
\begin{subequations}
\begin{equation}
\label{Eq::CosymHamEq}
\frac{dt}{ds}=1,\qquad \frac{dq^i}{ds}=\frac{\partial f}{\partial p_i}(t,q,p),\qquad \frac{dp_i}{ds}=-\frac{\partial f}{\partial q^i}(t,q,p),\qquad i=1,\ldots,n,
\end{equation}
where $(t,q,p)$ stands for $(t,q^1,\ldots,q^n,p_1,\ldots,p_n)$. 

Let $T$ be a one-dimensional manifold and let $(P,\omega)$ be a symplectic manifold. Let us define a cosymplectic manifold on $M=T\times P$ and its related  Hamilton equations. Let $\pi_T:M\rightarrow  T$ and let $\pi_P:M\rightarrow  P$ be the projections onto the first and second factors of $M=T\times P$, respectively. A symplectic form $\omega$ on $P$ gives rise to a closed differential two-form  $\omega_P:=\pi_P^*\omega$ on $M$. Meanwhile, a non-vanishing differential one-form $\eta$ on $T$ gives rise to a closed differential one-form $\eta_T=\pi_T^*\,\eta$ on $M$. Then, $(T\times P,\omega_P,\eta_T)$ becomes a cosymplectic manifold. If not otherwise stated, Darboux coordinates on $(T\times P,\omega_P,\eta_T)$ will be assumed to be of the form $\{t,q^1,\ldots,q^n,p_1,\ldots,p_n\}$, where $t$ is the pull-back to $T\times P$ of a potential of $\eta$, while $q^1,\ldots,q^n,p_1,\ldots,p_n$ are the pull-backs to $M$ of Darboux coordinates for $\omega$ on $P$. It is worth noting that it is common to denote the pull-backs of functions on $T$ and $P$ to $M=T\times P$ in the same way as the initial variables in $T$ and $P$.

If $M=\mathbb{R}\times T^*Q$, $\eta=dt$, and $\omega=\sum_{i=1}^ndq^i\wedge dp_i$, then (\ref{Eq::CosymHamEq}) can be rewritten as
\begin{equation}
\label{Eq::HamCosEq}
    \frac{dq^i}{dt}=\frac{\partial f}{\partial p_i}(t,q,p),\qquad \frac{dp_i}{dt}=-\frac{\partial f}{\partial q^i}(t,q,p),\qquad i=1,\ldots,n.
\end{equation}
\end{subequations} Hence, \eqref{Eq::HamCosEq} retrieves the Hamilton equations for a $t$-dependent Hamiltonian system on $T^*Q$ (see \cite{AM78,LZ21}). 

More generally, given a cosymplectic manifold $(M:=T\!\times\! P,\omega_P,\eta_T)$, we call {\it Hamilton equations} induced by $h\in C^\infty(M)$  the system of differential equations that, locally on each coordinated open $U\subset M$ by Darboux coordinates $\{t,q^1,\ldots,q^n,p_1,\ldots,p_n\}$, takes the form  
\begin{equation}\label{Eq:CoHE}
\frac{dq^i}{dt}=\frac{\partial h}{\partial p_i}(t,q,p),\qquad \frac{dp_i}{dt}=-\frac{\partial h}{\partial q^i}(t,q,p),\qquad i=1,\ldots,n.
\end{equation}

Roughly speaking, \eqref{Eq:CoHE} is the system of differential equations for the integral curves of $E_h$ parametrised by points of $T$ described by the coordinate $t$ in Darboux coordinates obtained from coordinates on $T$ and $P$ as indicated previously. Although the variable $t$ is defined up to an additive constant, equations \eqref{Eq:CoHE} are equivalent for each possible variable $t$ of our Darboux coordinates. Hence, our Hamilton equations have a geometrical meaning. The above remains valid even for a case like $T=\mathbb{S}^1$ as far as particular solutions are allowed to match every point of $T$ with several points of $P$ (see Figure \ref{Fig:Exa}).

\begin{figure}[ht]\label{Fig:Exa}
\includegraphics{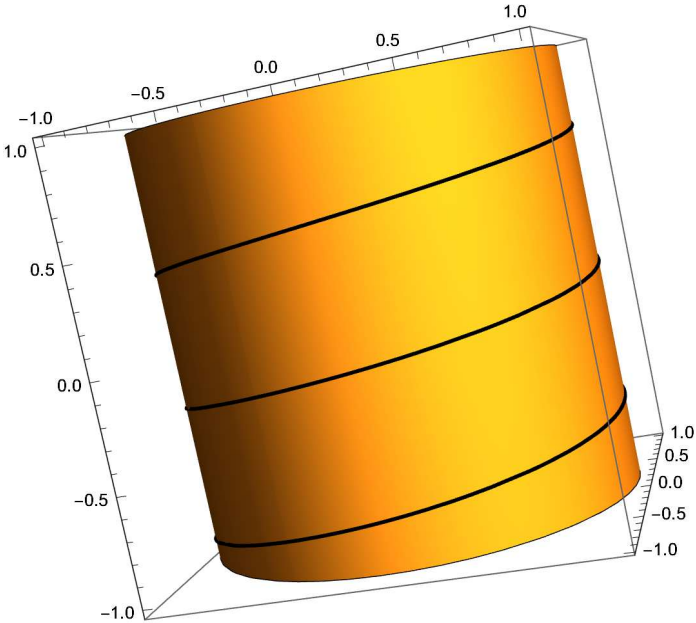}
{\caption{Example of solutions of Hamilton equations on a cosymplectic manifold $(\mathbb{S}^1\times T^*\mathbb{R},\omega_{T^*\mathbb{R}},\eta_{\mathbb{S}^1})$ for $\mathbb{S}^1$ being  the circle of radius one and centred at zero. Just the coordinates of solutions in $\mathbb{S}^1\times \mathbb{R}$ are represented}}
\end{figure}

Meanwhile, the integral curves of $X_f$ on $M$ are given by the solutions of
\[
\frac{dt}{ds}=0,\qquad \frac{dq^i}{ds}=\frac{\partial f}{\partial p_i}(t,q,p),\qquad \frac{dp_i}{ds}=-\frac{\partial f}{\partial q^i}(t,q,p),\qquad i=1,\ldots,n.
\]
It is worth noting that $X_h$ on $M=\mathbb{R}\times P$ can also be considered as a so-called {\it $t$-dependent vector field} \cite{AM78,LZ21}. Then, its integral curves take the form $\mathbb{R}\ni t\mapsto(t,z(t))\in \mathbb{R} \times P$, where $z(t)$ is a solution to \eqref{Eq:CoHE}. This also shows that its solutions have geometric meaning. Note that an analogue could be done for any $(M:=T\times P,\omega_P,\eta_T)$, but this could lead to having solutions of Hamilton equations mapping  every $t\in T$ into several different points of $P$ while being possible locally around every $t_0\in T$ to consider a solution as a union of local sections of $\pi_T:M\rightarrow  T$ whose images do not intersect each other (see again Figure \ref{Fig:Exa}).

Let us finally provide a pair of useful results on cosymplectic geometry.

\begin{proposition}
\label{Prop::gradXR}
The gradient vector field  of $f\in C^\infty(M)$ relative to $(M,\omega,\eta)$ reads $\nabla \,f=X_f+(Rf)R$. If $Rf=0$, then $[R,X_f]=0$.
\end{proposition}
\begin{proof}
Using the definition of Hamiltonian and gradient vector field, we have
\[
\iota_{\nabla f}\,\omega=df -(Rf)\eta=\iota_{X_f}\omega\quad\Rightarrow  \quad \iota_{\nabla f-X_f}\omega=0.
\]
Therefore, $\nabla f=X_f+Y$ for some vector field $Y$ on $M$ such that $\iota_Y\omega=0$. Hence,
\[
\iota_{\nabla f}\eta=\iota_{X_f}\eta+\iota_Y\eta=Rf.
\]
Since $\iota_{X_f}\eta=0$ and $\ker\eta\oplus\ker\omega=TM$, then $Y=(Rf)R$ and $\nabla f=X_f+(Rf)R$.

Moreover, if $Rf=0$, then
\[
\iota_{[X_f,R]}\omega=\mathcal{L}_{X_f}\iota_R\omega-\iota_R\mathcal{L}_{X_f}\omega=-\iota_Rd\iota_{X_f}\omega=\iota_Rd(df-(Rf)\eta)=0,
\]
and
\[
\iota_{[X_f,R]}\eta=\mathcal{L}_{X_f}\iota_R\eta -\iota_R\mathcal{L}_{X_f}\eta=-\iota_Rd\iota_{X_f}\eta=0.
\]
Since $TM=\ker\eta\oplus\ker\omega$, then $[X_f,R]=0$.
\end{proof}

Every cosymplectic manifold $(M,\omega,\eta)$ gives rise to a Poisson bracket $\{\cdot,\cdot\}_{\omega,\eta}\!:\!C^\infty(M)\times C^\infty(M)\rightarrow  C^\infty(M)$ of the form
\begin{equation}
\label{Eq::PoissonStructure}
\{f,g\}_{\omega,\eta}:=\omega(\nabla \,f,\nabla \,g)=\omega(X_f,X_g),\qquad \forall f,g\in C^\infty(M),
\end{equation}
where the last equality is a consequence of Proposition \ref{Prop::gradXR} and $\iota_R\,\omega=0$. As in the symplectic case \begin{equation}
\label{Eq::AntiMorphism}
X_{\{f,g\}_{\omega,\eta}}=-[X_f,X_g],\qquad \forall f,g\in C^\infty(M).
\end{equation} 

It is worth noting that the Poisson bivector associated with the Poisson bracket $\{\cdot,\cdot\}_{\omega,\eta}$ is given by
\[
\Lambda_{\omega,\eta}(x)(\alpha_x,
\beta_x)=\{f,g\}(x)=\omega_{x}(X_f,X_g),\qquad \forall x\in M,
\]
where $df(x)=\alpha_x$ and $dg(x)=\beta_x$ are elements of $T^*_xM$ for certain $f,g\in C^\infty(M)$.
In Darboux coordinates, the Poisson bivector $\Lambda_{\omega,\eta}$ reads\begin{equation}
\label{eq:PoiCosym}
\Lambda_{\omega,\eta}=\sum_{i=1}^n\frac{\partial}{\partial x^i}\wedge \frac{\partial}{\partial p_i}.
\end{equation}
It is remarkable that $\mathcal{L}_R\Lambda_{\omega,\eta}=0$. 

The space of Hamiltonian vector fields relative to a cosymplectic manifold $(M,\omega,\eta)$, let us say ${\rm Ham}(M,\omega,\eta)$, is a Lie algebra. Moreover, there exists a homomorphism of Lie algebras $f\in C^\infty(M)\mapsto -X_f\in {\rm Ham}(M,\omega,\eta)$. Note that $X_f$ is a Hamiltonian vector field relative to the Poisson bracket $\{\cdot,\cdot\}_{\omega,\eta}$, namely $X_f=\{\cdot,f\}_{\omega,\eta}$, but $E_f$ is never so, while  $\nabla f$ is not Hamiltonian in general (cf. \cite{Va94}).

\subsection{Cosymplectic, symplectic, and Poisson manifolds}
\label{SubSec::CosStructures}

In physical applications, one is mostly interested in cosymplectic manifolds $(T\times P,\omega_P,\eta_T)$, where $T$ is a one-dimensional manifold describing a certain time interval and $P$ is a symplectic manifold. Let us detail some interesting cases of $T$ and $P$. If $Q$ is the configuration manifold of a physical system, then $T^*Q$ can be endowed with its canonical symplectic structure and
its common to set $P=T^*Q$ (see \cite{AM78}). On the other hand, we can also assume $P=TQ$ to be endowed with the symplectic structure induced by a {\it regular Lagrangian} function (see \cite{RM99} for details). Meanwhile, $T$ can be chosen to be $\mathbb{R}$ with its natural variable $t$ and to define $\eta=dt$ on $\mathbb{R}$, which is used to describe a physical system at every time $t\in \mathbb{R}$. Another option is to set $T=\mathbb{S}^1$ with the closed non-degenerate one-form $d\theta$, where $\mathbb{S}^1$ is the unit circle in $\mathbb{R}^2$ with a coordinate given by the angle $\theta$ of each point of $\mathbb{S}^1$ relative to a reference point in it. The model $T=\mathbb{S}^1$ may be employed to study $t$-dependent Hamilton equations with a $t$-dependent periodic Hamiltonian. 

We recall that cosymplectic Darboux coordinates for $(T\times P,\omega_P,\eta_T)$ are assumed to be of the form $\{t,q^1,\ldots,q^n,p_1,\ldots,p_n\}$ so that $\{q^1,
\ldots,q^n,p_1,\ldots,p_n\}$ are the pull-back to $M$ of Darboux coordinates for a symplectic form on $P$, while $t$ is the pull-back to $M$ of a potential of a closed one-form on $T$. 

Although a cosymplectic manifold $(M,
\omega,\eta)$ induces a symplectic manifold $(\mathbb{R}\times M, \widehat{\omega})$ and a Poisson manifold $(M, \{\cdot,\cdot\}_{\omega,\eta})$,  we will show that these approaches are not appropriate to generalise the energy-momentum method on symplectic manifolds to cosymplectic manifolds. 

Let us detail the following lemma \cite{LS93}, which states how symplectic and cosymplectic manifolds are naturally related. A proof is given, as it uses to be absent in many works in the literature.
\begin{lemma}
\label{Lemm::CosymSym}
Let $\omega\in\Omega^2(M), \eta\in\Omega^1(M)$ and let ${\rm pr} :\mathbb{R}\times M\rightarrow  M$ be the canonical projection onto $M$. Let $s$ be the natural coordinate in $\mathbb{R}$ understood as a variable in $\mathbb{R}\times M$ in the natural manner. Then, $(M,\omega,\eta)$ is a cosymplectic manifold if and only if $(\mathbb{R}\times M,{\rm pr}^*\omega+ds\wedge {\rm pr}^*\eta=:\widehat{\omega})$ is a symplectic manifold. Moreover,  ${\rm pr}$ is a Poisson morphism, i.e.
\[
\{f\circ {\rm pr},k\circ {\rm pr}\}_{\widehat{\omega}}=\{f,k\}_{\omega,\eta}\circ {\rm pr},\qquad \forall f,k\in C^\infty(M).
\]
\end{lemma}
\begin{proof}
Recall that if $(M,\omega,\eta)$ is a cosymplectic manifold and $\dim M=2n+1$, then $\omega^n\wedge\eta$ is a volume form. Since $\om\in\Omega^2(M)$ and $\eta\in\Omega^1(M)$ are closed,
\begin{equation}
\label{Eq::OmHatClosed}
d\widehat\omega=d({\rm pr}^*\omega+ds\wedge {\rm pr}^*\eta)={\rm pr}^*d\om-ds\wedge{\rm pr}^*d\eta=0,
\end{equation}
and $\widehat\omega\in\Omega^2(\mathbb{R}\times M)$ is also closed. Since $\vartheta^{n+1}=0$ for every differential two-form $\vartheta$ on $M$, one has that
\begin{multline}
    \label{Eq::OmHatVol}
\widehat\omega^{n+1}=({\rm pr}^*\omega+ds\wedge{\rm pr}^*\eta)^{n+1}=(n+1)({\rm pr}^*\omega)^n\wedge ds\wedge {\rm pr}^*\eta=(n+1)ds\wedge{\rm pr}^*(\omega^n\wedge\eta),
\end{multline}
is clearly a volume form on $\mathbb{R}\times M$ and it is non-zero. Thus, $\widehat\omega$ is non-degenerate.

Conversely, if $(\mathbb{R}\times M,\widehat\omega)$ is a symplectic manifold, relation \eqref{Eq::OmHatVol} shows that $\omega^n\wedge\eta\neq 0$. Moreover, \eqref{Eq::OmHatClosed} gives that ${\rm pr}^*\omega$ and ${\rm pr}^*\eta$ are closed forms. Since ${\rm pr}$ is a surjective submersion, $d\omega=0$ and $d\eta=0$. Therefore, $(M,\omega,\eta)$ is a cosymplectic manifold. 

Moreover, if $\{\cdot,\cdot\}_{\widehat{\omega}}$ is the Poisson bracket induced by the symplectic form $\widehat{\omega}$, then
\begin{multline*}
{\rm pr}^*\{f,k\}_{\omega,\eta}=-{\rm pr}^*(\iota_{X_f}\iota_{X_k}\omega)=-{\rm pr}^*(\iota_{X_f}dk-(Rk)\iota_{X_f}\eta)\\=-\iota_{X_{{\rm pr}^*f}}{\rm pr}^*dk=-\iota_{X_{{\rm pr}^*f}}d{\rm pr}^*k=\{{\rm pr}^*f,{\rm pr}^*k\}_{\widehat{\omega}},
\end{multline*}
for every $f,k\in C^\infty(M)$,
and ${\rm pr}$ is a Poisson morphism. Note that $X_{{\rm pr}^*f}$ stands for the Hamiltonian vector field on $(\mathbb{R}\times M, \widehat{\omega})$ of the function ${\rm pr}^*f\in C^\infty(\mathbb{R}\times M)$. The fact that ${\rm pr}_* X_{{\rm pr}^*f}=X_f$ follows, for instance, from writing $X_{{\rm pr}^*f}$ in Darboux coordinates for $\widehat{\omega}$ obtained by adding $s$ to the pull-back to $\mathbb{R}\times M$ of some Darboux coordinates for $(M,\omega,\eta)$.
\end{proof}

Let us show that the vector fields $\nabla f$, $X_f$, and $E_f$ on $M$ cannot, straightforwardly, be considered as Hamiltonian vector fields relative to the symplectic form $\widehat\om$ on $\mathbb{R}\times M$ induced by the cosymplectic manifold $(M,\om,\eta)$. Let us comment on this fact by relating $f$, $\nabla f$, $X_f$, and $E_f$ to natural mathematical structures on $\mathbb{R}\times M$, e.g. let $\tilde{f}:={\rm pr}^*f$ and let $\tl F_g$, $\tl X_{f}$, and $\tl E_{f}$ be the vector fields on $\mathbb{R}\times M$ obtained by considering the isomorphism $T_{(s,x)}(\mathbb{R}\times M)\simeq T_s\mathbb{R} \oplus T_x M$ for every $s\in \mathbb{R}$ and $x\in M$. Equivalently, one can define $\tl F_g$, $\tl X_{f}$, and $\tl E_{f}$ to be the only vector fields on $\mathbb{R}\times M$ projecting onto $\nabla f$, $X_f$, and $E_f$ via ${\rm pr}_*$, respectively, and satisfying $\iota_{\tilde{F}_g}ds=\iota_{\tilde{X}_{f}}ds=\iota_{\tilde{E}_{f}}ds=0$. Then,
\begin{multline*}
d\iota_{\tilde{F}_g}\,\,\ho=d(\iota_{\tilde{F}_g}\,{\rm pr}^*\omega - (\iota_{\tilde{F}_g}\,{\rm pr}^*\eta)ds)=d({\rm pr}^*(\iota_{\nabla f}\,\omega )-{\rm pr}^*( Rf) ds)\\={\rm pr}^*(d(df-(Rf)\eta))-{\rm pr}^*(d(Rf))\wedge ds=-{\rm pr}^*(d(Rf))\wedge( ds+ {\rm pr}^*\eta).
\end{multline*}
Therefore, $\tilde{F}_g$ is not, in general, a Hamiltonian vector field on $\mathbb{R}\times M$ relative to $\widehat\om$. Similarly,
\begin{multline*}
d\iota_{\tilde{X}_{f}}\,\widehat{\omega}=d(\iota_{\tilde{X}_{f}}\,{\rm pr}^*\omega-(\iota_{\tilde{X}_{f}}{\rm pr}^*\eta)ds)=d({\rm pr}^*(\iota_{X_f}\omega)-{\rm pr}^*(\iota_{X_f}\eta)ds)\\={\rm pr}^*(d\iota_{X_f}\omega)-{\rm pr}^*(d\iota_{X_f}\eta)\wedge ds=-{\rm pr}^*(d(Rf)\wedge\eta),
\end{multline*}
and $\tilde{X}_{f}$ is not, neither, a Hamiltonian vector field on $\mathbb{R}\times M$ in general. Finally,

\begin{equation*}
d\iota_{\tilde{E}_{f}}\,\widehat{\omega}=d(\iota_{\tilde{R}}\,\ho+\iota_{\tilde{X}_{f}}\,\widehat{\omega})=d\iota_{\tilde{X}_{f}}\,\ho=-{\rm pr}^*(d(Rf)\wedge\eta),
\end{equation*}
where $\tilde{R}$ is the only vector field on $\mathbb{R}\times M$ pojectable onto $M$ via ${\rm pr}_*$ and satisfying $\iota_{\tilde{R}} ds=0$.
Then, $\tilde{E}_{f}$ is not, in general,  a Hamiltonian vector field on $\mathbb{R}\times M$.

Notwithstanding, if $d(Rf)=0$, then $\nabla f$, $X_f$, $E_f$ give rise naturally to Hamiltonian vector fields $\tilde{F}_g$, $\tilde{X}_{f}$, and $\tilde{E}_{f}$, respectively, relative to the symplectic structure on $\mathbb{R}\times M$. However, in general, the latter does not need to be true. The  condition $Rf=0$, which appears in a different disguise in cosymplectic theory \cite{Al89} may be used to define, on cosymplectic manifolds, an analogue of the geometric structures and techniques appearing in symplectic manifolds \cite{LS93}. 

It is remarkable that there exist other methods to consider some of the above stressed vector fields in $M$ as Hamiltonian vector fields on $\mathbb{R}\times M$  (see for instance \cite{LS93} or the proof of Lemma \ref{Lemm::CosymSym}). Nevertheless, these methods use to change the properties of  vector fields on $M$ in such a manner that they may make them harder to study, e.g. some methods can turn a vector field on $M$ with zeros into one without them in $\mathbb{R}\times M$, which can give rise to problems to study stability. For instance, $(Rf)\frac{\partial}{\partial s}+X_f$ is a Hamiltonian vector field on $\mathbb{R}\times M$ with respect to $\widehat{\omega}$ and projecting onto $X_f$ relative to ${\rm pr}_*$. The vector field $(Rf)\frac{\partial}{\partial s}+X_f$ has different stability properties than $X_f$, e.g. it may have not equilibrium points at all while $X_f$ has, and thus introduces new difficulties in the stability analysis of the latter.

Finally, we would like to recall that studying a problem defined on a manifold $M$ through a related problem in a manifold of larger dimension may turn the problem on $M$ harder to analyse (see \cite{LX22} for instance) and it goes, at the every end, against the original idea in the foundations of differential geometry: structures must be studied in terms of the properties of the manifold where structures are naturally defined (unless one has a very good reason to do otherwise like in \cite{GG22}). To see the inconvenience of such approaches, let us pay attention to \cite{LS15, LS20}, which contain a classification of finite-dimensional Lie algebras of Hamiltonian vector fields on $\mathbb{R}^2$ relative to Jacobi structures around the so-called {\it generic points} of the Lie algebras. The classification was accomplished thanks to the fact that all finite-dimensional Lie algebras of vector fields around generic points are classified in $\mathbb{R}^2$ \cite{GO92}. One could relate the Jacobi structure on $\mathbb{R}^2$ to a homogeneous Poisson structure on $\mathbb{R}^3$ and look for Lie algebras of Hamiltonian vector fields relative to a homogeneous Poisson bracket on $\mathbb{R}^3$ that could be projected onto $\mathbb{R}^2$, which would recover the Lie algebras of Hamiltonian vector fields we are looking for \cite{BG17}. This approach makes things so difficult that solving the problem in this way is nowadays impossible: no full classification of finite-dimensional Lie algebras of vector fields on $\mathbb{R}^3$ is known\footnote{Only partial results due to Sophus Lie are known (see \cite{BBHLS15} and references therein). } and, even if it were known, it would lead to a more complicated approach than the approach one in \cite{LS15, LS20}. In fact, it is worth noting that there exist 28 non-trivial classes of locally diffeomorphic finite-dimensional Lie algebras of vector fields on $\mathbb{R}^2$, and an analogous 	classification on $\mathbb{R}^3$ would be much more complicated and harder to obtain (we refer to \cite{LS20,GO92} for further details).

Note that each cosymplectic Hamiltonian vector field is
Hamiltonian relative to the Poisson bracket associated with its cosymplectic manifold. Nevertheless, gradient vector fields are not in general Hamiltonian, while evolution vector fields are never   Hamiltonian relative to the Poisson bracket of cosymplectic manifolds. These facts along with others to be described in the following sections will make the Poisson bracket associated with cosymplectic structures inappropriate, by itself, to analyse the problems to be studied hereafter.  Moreover, we will not be able to use in the following sections either the classical energy-momentum method \cite{MS88} nor the Casimir-energy method, designed for studying the stability of relative equilibrium points of Hamiltonian systems on Poisson manifolds \cite{MS88}. In particular, the restricted circular three-body problem analysed in Section \ref{Sec::RDTP} will show the necessity of the new techniques of our work to study certain problems and how even the time-dependent energy-momentum method in \cite{LZ21} is not enough to analyse some types of problems studied with the new methods of our present article.

\section{Momentum maps}\label{Sec::AMomentumMap}

Let us recall the results needed to understand the cosymplectic Marsden--Weinstein reduction  \cite{Al89}. This reduction plays a crucial role in our cosymplectic energy-momentum method since its related relative equilibrium points, which are defined in Section \ref{Sec::CharStrongEqP} and are a key in our study, are points projecting, in a certain sense, to equilibrium points of a reduced Hamiltonian appearing via the cosymplectic  Marsden--Weinstein reduction. Although some of the following results can be found in the literature, we prove them here to keep our work self-contained and because some proofs are, in general, not available. For instance, the classical works \cite{Al89,Li62} are in French and the manuscripts of \cite{Al89} available online are illegible at relevant places. Other results given next are generalisations to the cosymplectic realm of standard results in symplectic geometry. Let us start by defining momentum maps for cosymplectic manifolds.

\begin{definition}
\label{Def:MomentumMap}
Let $\Phi:G\times M\rightarrow M$ be a cosymplectic action such that $\iota_{\xi_M}\eta=0$. A {\it cosymplectic momentum map} for a Lie group action $\Phi:G\times M\rightarrow  M$ relative to a cosymplectic manifold $(M,\omega,\eta)$ is a map $\B J^\Phi:M\rightarrow  \mathfrak{g}^*$ such that
\begin{equation}
\label{Eq::CosMomMap}
\iota_{\xi_M}\omega=d\langle \B J^\Phi,\xi\rangle:=dJ_\xi,\qquad RJ_\xi=0,\qquad \forall\xi\in\mathfrak{g} .
\end{equation}
\end{definition}
In the literature, it is frequently assumed that the cosymplectic momentum map is $\Ad^*$-equivariant.

\begin{definition}
    A momentum map $\mathbf{J}^\Phi:M\rightarrow  \mathfrak{g} ^*$ is {\it ${\rm Ad}^*$-equivariant} if  
    \[
    \mathbf{J}^\Phi\circ\Phi_g={\rm Ad}^*_{g^{-1}}\circ \mathbf{J}^\Phi,\qquad  \forall g\in G.
    \]
    In other words, the following diagram commutes
\\
\begin{center}
    \begin{tikzcd}
    M
    \arrow[r,"\mathbf{J}^\Phi"]
    \arrow[d,"\Phi_g"]& \mathfrak{g} ^*
    \arrow[d,"{\rm Ad}^*_{g^{-1}}"]\\
    M
    \arrow[r,"\mathbf{J}^\Phi"]&
    \mathfrak{g} ^*
    \end{tikzcd},
    \end{center}
    for every $g\in G$ and ${\rm Ad}_{g^{-1}}^*$ being the transpose of ${\rm Ad}_{g^{-1}}$.
\end{definition}
Instead of saying ``cosymplectic momentum map”, we will use the term momentum map when it is clear from the context that we mean a momentum map relative to a cosymplectic manifold.
Note that $RJ_\xi=0$ is a condition required to apply the cosymplectic reduction theorem to be introduced in Section \ref{Sec::CosymplReduction}. This condition may be considered very restrictive but one can find interesting examples of $t$-dependent or dissipative Hamiltonian systems satisfying it that can be studied via our cosymplectic energy-momentum method to be introduced from Section \ref{Sec::Stab}.

Note that if a Lie group action $\Phi: G\times M\rightarrow  M$ admits a momentum map relative to $\Phi$ and $(M,\omega,\eta)$, then $\Phi$ is a cosymplectic Lie group action (assuming $G$ to be connected). Nevertheless, not every cosymplectic Lie group action on $M$ admits a momentum map, e.g. the flow of a Reeb vector field $R$ relative to $(M,\omega,\eta)$ is a cosymplectic Lie group action, but the flow does not admit a momentum map relative to its associated cosymplectic manifold because $\iota_R\eta=1\neq 0$.

In view of \eqref{Eq::CosMomMap}, the Reeb vector field, $R$, corresponding to $(M,\omega,\eta)$ is always tangent to the level sets of any momentum map $\mathbf{J}^\Phi$ related to $(M,\omega,\eta)$. Nevertheless, $R$ is never tangent to the orbits of $\Phi$ since, in that case, one should have $\iota_R\,\eta=0$, which is in contradiction with the definition of $R$.

Given the canonical one- and two-forms on $T^*Q$, given by $\theta_{T^*Q}$ and $\omega_{T^*Q}$ respectively, one can define canonical one- and two-forms on $T\times T^*Q$ given by 
\[
\theta_{T\times T^*Q}:=\pi_{T^*Q}^*\theta_{T^*Q},\qquad \omega_{T\times T^*Q}:=\pi^*_{T^*Q}\omega_{T^*Q},
\]
where $\pi_{T^*Q}:T\times T^*Q\rightarrow  T^*Q$ is the canonical projection onto $T^*Q$. Then, $\omega_{T\times T^*Q}=-d\theta_{T\times T^*Q}$, which turns $\omega_{T\times T^*Q}$ into a closed differential form, while $\ker\omega_{T\times T^*Q}$ is a distribution of rank $1$. In consequence, one has a canonical cosymplectic manifold
\begin{equation}
\label{Eq::CanonicalCosStructure}
(T\times T^*Q,\omega_{T\times T^*Q},\eta_{T\times T^*Q}).
\end{equation}
where $\eta_{T\times T^*Q}$ is the pull-back to $T\times T^*Q$ of a closed non-vanishing one-form $\eta$ on $T$. When $T=\mathbb{R}$, we will assume $\eta=dt$, where $t$ is the natural variable in $\mathbb{R}$.

Let us recall that every Lie group action $\Phi:G\times Q\rightarrow  Q$ gives rise to a canonical lift $\widehat\Phi:G\times  T^*Q\rightarrow  T^*Q$ given by 
\[
\<\widehat\Phi_g(\alpha_q),v_{\Phi_g(q)}\>:=\<\alpha_q,T_{\Phi_{g}(q)}\Phi_{g^{-1}}(v_{\Phi_g(q)})\>,\,\forall g\in G,\, \forall q\in Q,\, \forall \alpha_q\in T_q^*Q,\,\forall v_{\Phi_g(q)}\in T_{\Phi_g(q)}Q.
\]
Therefore, we can define Lie group actions $\Psi:G\times T \times Q\rightarrow  T\times Q$ and $\widehat\Psi:G\times T\times T^*Q\rightarrow  T\times T^*Q$ in the following manner
\begin{equation}
    \label{Eq::LiftRAction}
    \Psi:G\times T \times Q\ni(g,t,q)\mapsto(t,\Phi_g(q))\in T\times Q
\end{equation}
and
\begin{equation}
\label{Eq::LiftCosAction}
\widehat\Psi:G\times T\times T^*Q\ni(g,t,\alpha_q)\mapsto (t,\widehat \Phi_g(\alpha_q))\in T\times T^*Q.
\end{equation}
If the Lie group action $\widehat\Phi$ is a symplectic Lie group action relative to the canonical symplectic structure on $T^*Q$, namely it leaves the symplectic form invariant, and it does not change the coordinate $t$, the action $\widehat\Psi$ is a cosymplectomorphism relative to the cosymplectic manifold \eqref{Eq::CanonicalCosStructure}.

\begin{proposition}
\label{Prop::LiftSymAc}
Every Lie group action $\Phi:G\times Q\rightarrow  Q$ is such that its lift, $\widehat\Psi$, given by \eqref{Eq::LiftCosAction}, admits an  ${\rm Ad}^*$-equivariant cosymplectic momentum map ${\bf J}^{\widehat\Psi}:T \times T^*Q\rightarrow  \mathfrak{g}^{*}$ such that
\begin{equation}
\label{Eq::CanMomMap}
\langle {\bf J}^{\widehat\Psi},\xi\rangle=\iota_{\xi_{T\times T^*Q}}\theta_{T \times T^*Q},\qquad \quad\forall \xi\in\mathfrak{g},
\end{equation}
relative to the canonical cosymplectic structure $(T \times T^*Q,\omega_{T\times T^*Q},\eta_{T \times T^*Q})$.
\end{proposition}
\begin{proof}
Since $\mathbf{J}^{\widehat\Psi}$, $\xi_{T\times T^*Q}$, and $\theta_{T\times T^*Q}$ are invariant relative to the Lie derivative with respect to the Reeb vector field $R$ of the considered cosymplectic manifold, one has that \eqref{Eq::CanMomMap} amounts to the pull-back 
 via $\pi_{T^*Q}:T\times T^*Q\rightarrow  T^*Q$ of
\[
\<\mathbf{J}^{\widehat\Phi},\xi\>=\iota_{\xi_{T^*Q}}\theta_{T^*Q},\qquad \forall \xi\in \mathfrak{g},
\]
which is a well-defined ${\rm Ad}^*$-equivariant momentum map on $T^*Q$ (see \cite{AM78}).
\end{proof}

To simplify the notation, the cosymplectic manifold $(M,\omega,\eta)$ will be frequently denoted by $M_\eta^\omega$. 

\begin{definition}
A triple $(M^\omega_\eta,h,{\bf J}^\Phi)$ is called a {\it $G$-invariant cosymplectic Hamiltonian system} if it consists of a cosymplectic manifold $(M,\om,\eta)$, an associated cosymplectic Lie group action $\Phi:G\times M\rightarrow  M$ such that $\Phi_g^*h=h$ for every $g\in G$, and a momentum map $\mathbf{J}^\Phi$ related to $\Phi$. An {\it ${\rm Ad}^*$-equivariant $G$-invariant cosymplectic Hamiltonian system} $(M^\omega_\eta,h,{\bf J}^\Phi)$ is a $G$-invariant cosymplectic Hamiltonian system whose momentum map is ${\rm Ad}^*$-equivariant.
\end{definition}

\section{General momentum maps}
\label{Sec::AGeneralMomentumMap}

This section devises the theory of non-${\rm Ad}^*$-equivariant momentum maps on cosymplectic manifolds. Essentially, our results provide a rather immediate extension of the well-established theory for general momentum maps in symplectic manifolds (see \cite[p. 278]{AM78}). It also provides a slight adaptation of the results by Albert in \cite{Al89}. Recall that all manifolds are assumed to be connected unless otherwise stated.

\begin{proposition}
Let $(M^\omega_\eta, h,{\bf J}^\Phi)$ be a $G$-invariant cosymplectic Hamiltonian system and let us define the functions on $M$ of the form
\[
\psi_{g,\xi}:M\ni x\mapsto { J}_{\xi}(\Phi_g(x))-{ J}_{{\rm Ad}_{g^{-1}}\xi}(x)\in\mathbb{R},\qquad \forall g\in G,\qquad \forall \xi\in \mathfrak{g}.
\]
Then, $\psi_{g,\xi}$ is constant on $M$ for every $g\in G$ and $\xi\in \mathfrak{g}$. Moreover, the function $\sigma:G\ni g\mapsto \sigma(g)\in\mathfrak{g} ^*$ such that $\<\sigma(g),\xi\>:=\psi_{g,\xi}$  for all $\xi\in \mathfrak{g} $ satisfies
\begin{equation}
\label{Eq::SigmaProperty}
\sigma(gg')=\sigma(g)+{\rm Ad}_{g^{-1}}^*\sigma(g'),\qquad \forall g,g'\in G.
\end{equation}
\end{proposition}

\begin{proof}
To show that each $\psi_{g,\xi}$ is constant on $M$, note that
\begin{multline*}
d\psi_{g,\xi}=d[J_\xi\circ \Phi_g]                          -d{ J}_{{\rm Ad}_{g^{-1}}\xi}=\Phi_g^*(\iota_{\xi_M}\omega)-\iota_{({\rm Ad}_{g^{-1}}\xi)_M}\omega
\\=\Phi_g^*(\iota_{\xi_M}\omega)-\iota_{\Phi_{g^{-1}*}\xi_M}\omega=\iota_{\Phi_{g^{-1}*}\xi_M}\Phi_g^*\omega-\iota_{\Phi_{g^{-1}*}\xi_M}\omega=\iota_{\Phi_{g^{-1}*}\xi_M}\omega-\iota_{\Phi_{g^{-1}*}\xi_M}\omega=0,
\end{multline*}
where we have used that $\Phi$ is a cosymplectic Lie group action and the fact that $({\rm Ad}_{g}\xi)_M=\Phi_{g*}\xi_M$ for every $g\in G$ and every $\xi\in\mathfrak{g}$ (see \cite{AM78}). Hence, each $\psi_{g,\xi}$ is constant on every connected component of $M$. Since $M$ is assumed to be connected, the functions $\psi_{g,\xi}$ are constant.

To study simultaneously the mappings $\{\psi_{g,\xi}\}_{g\in G,\xi\in\mathfrak{g} }$, let us rewrite $\psi_{g,\xi}$ as follows
\begin{align*}
\psi_{g,\xi}(x)&={J}_\xi(\Phi_g(x))-{ J}_{{\rm Ad}_{g^{-1}}\xi}(x)=\lv {\bf J}^\Phi(\Phi_g(x)),\xi\rv - \lv {\bf J}^\Phi(x),{\rm Ad}_{g^{-1}}\xi\rv\\&=\lv {\bf J}^\Phi(\Phi_g(x)),\xi\rv- \lv {\rm Ad}_{g^{-1}}^*{\bf J}^\Phi(x),\xi\rv=\lv {\bf J}^\Phi(\Phi_g(x))-{\rm Ad}_{g^{-1}}^*{\bf J}^\Phi(x),\xi\rv,
\end{align*}
for all $x\in M$. Since $\lv \sigma(g),\xi\rv =\psi_{g,\xi}$ is constant on $M$ for every $g\in G$ and $\xi\in \mathfrak{g}$, one has  that $\sigma$ is given by
\begin{equation}\label{Eq:SpeSigma}
\sigma:G\ni g\mapsto {\bf J}^\Phi\circ \Phi_g -{\rm Ad}_{g^{-1}}^*{\bf J}^\Phi\in\mathfrak{g}^*,
\end{equation}
and every $\sigma(g)$, with $g\in G$, is constant on $M$. 

A simple calculation shows that, as all the $\psi_{g,\xi}$ are constant, one has
\begin{align*}
\sigma(gg')=&{\bf J}^\Phi\circ \Phi_{gg'}-{\rm Ad}_{{(gg')}^{-1}}^*{\bf J}^\Phi={\bf J}^\Phi\circ\Phi_g\circ \Phi_{g'}-{\rm Ad}_{g^{-1}}^*{\rm Ad}_{g'^{-1}}^*{\bf J}^\Phi\\=&{\bf J}^\Phi\circ \Phi_g\circ\Phi_{g'}-{\rm Ad}^*_{g^{-1}}{\bf J}^\Phi\circ\Phi_{g'}+{\rm Ad}^*_{g^{-1}}{\bf J}^\Phi\circ \Phi_{g'}-{\rm Ad}_{g^{-1}}^*{\rm Ad}_{g^{'-1}}^*{\bf J}^\Phi\\=&{\bf J}^\Phi\circ \Phi_g\!-\!{\rm Ad}^*_{g^{-1}}{\bf J}^\Phi \!+\!{\rm Ad}^*_{g^{-1}}({\bf J}^\Phi\circ\Phi_{g'}\!-\!{\rm Ad}_{g^{'-1}}^*{\bf J}^\Phi\!)=\!\sigma(g)\!+\!{\rm Ad}_{g^{-1}}^*\sigma(g')
\end{align*}
for every $g,g'\in G$, which proves \eqref{Eq::SigmaProperty}.
\end{proof}

As in the symplectic case, the map 
 (\ref{Eq:SpeSigma}) is called the {\it co-adjoint cocycle} associated with the cosymplectic momentum map ${\bf J}^\Phi$ on $M$. As in the symplectic case, ${\bf J}^\Phi$ is ${\rm Ad}^*$-equivariant momentum map if and only if $\sigma=0$. Roughly speaking, $\sigma$ measures the lack of ${\rm Ad}^*$-equivariance of a cosymplectic momentum map.

A map $\sigma:G\rightarrow \mathfrak{g}^*$ is a {\it coboundary} if there exists $\mu\in\mathfrak{g}^*$ such that
\begin{equation}
\label{CEq:on}
\sigma(g)=\mu-{\rm Ad}_{g^{-1}}^*\mu,\qquad \forall g\in G.
\end{equation}
 
Every coboundary satisfies \eqref{Eq::SigmaProperty} and it is therefore called a \textit{co-adjoint cocycle}. The space of co-adjoint cocycles admits an equivalence relation, whose equivalence classes are called {\it cohomology classes}, given by setting that two co-adjoint cocycles belong to the same cohomology class if their difference is a coboundary. The following proposition shows that every cosymplectic action admitting a cosymplectic momentum map induces a well-defined cohomology class $[\sigma]$. Note that these results could be straightforwardly adapted to the symplectic case.
\begin{proposition}\label{Prop:CoAdjCocycles}
Let $\Phi:G\times M \rightarrow  M$ be a cosymplectic Lie group action relative to $(M,\omega,\eta)$. If ${\bf J}_1^\Phi$ and ${\bf J}_2^\Phi$ are two  momentum maps related to $\Phi$ with co-adjoint cocycles $\sigma_1$ and $\sigma_2$, respectively, then $[\sigma_1]=[\sigma_2]$.
\end{proposition}
\begin{proof}
From the definition of a co-adjoint cocycle for a cosymplectic momentum map,
\[
\lv\sigma_1(g)-\sigma_2(g),\xi\rv=\lv {\bf J}_1^\Phi\circ\Phi_g-{\bf J}_2^\Phi\circ\Phi_g,\xi\rv-\lv {\rm Ad}^*_{g^{-1}}({\bf J}_1^\Phi-{\bf J}_2^\Phi),\xi\rv,
\]
for all $g\in G$ and $\xi\in\mathfrak{g}$.
However, ${\bf J}_1^\Phi-{\bf J}_2^\Phi$ takes a constant value $\mu\in\mathfrak{g}^*$, since ${\bf J}_1^\Phi$ and ${\bf J}_2^\Phi$  are cosymplectic momentum maps for the same Lie group action. Indeed,
\[
d\lv {\bf J}^\Phi_1 -{\bf J}^\Phi_2,\xi\rv=dJ_{1,\xi}-dJ_{2,\xi}=\iota_{\xi_M}\om-\iota_{\xi_M}\omega=0,\qquad \forall\xi\in\mathfrak{g}.
\]
Thus, $({\bf J}^\Phi_1-{\bf J}^\Phi_2)\circ \Phi_g={\bf J}^\Phi_1-{\bf J}^\Phi_2$ for every $g\in G$ and then 
\[
\sigma_1(g)-\sigma_2(g)=\mu-{\rm Ad}^*_{g^{-1}}\mu,\qquad \forall g\in G.
\]
\end{proof}
Proposition \ref{Prop:CoAdjCocycles} yields that a cosymplectic Lie group action has an ${\rm Ad}^*$-equivariant momentum map if and only if  it has an associated coboundary. Indeed, if a cosymplectic Lie group action has an ${\rm Ad}^*$-equivariant momentum map ${\bf J}^\Phi_2$ relative to $(M,\om,\eta)$, then its associated co-adjoint cocycle satisfies $\sigma_2=0$, and any other momentum map ${\bf J}^\Phi_1$ for the same action is such that its co-adjoint cocycle, let us say $\sigma_1$, satisfies $[\sigma_1]=[\sigma_2]=0$, and $\sigma_1$ becomes a coboundary. Moreover, if $\sigma_1$ is a coboundary induced by $\mu\in \mathfrak{g}^*$, then the momentum map
\[
\B J^\Phi :=\B J^\Phi_1 -\mu,
\]
is an ${\rm Ad}^*$-equivariant momentum map for the same cosymplectic Lie group action of ${\bf J}_1^\Phi$, where $\mu\in\mathfrak{g}^*$ satisfies that $\sigma_1(g)=\mu-{\rm Ad}^*_{g^{-1}}\mu$ for every $g\in G$. In fact,
\[
\langle \B J^\Phi,\xi\rangle=\langle \B J^\Phi_1,\xi\rangle-\langle \mu,\xi\rangle=J_{1,\xi}-\langle \mu,\xi\rangle,\qquad \forall \xi\in\mathfrak{g},
\]
and
\[
\sigma(g)=\B J^\Phi\circ\Phi_g-{\rm Ad}^*_{g^{-1}}\B J^\Phi=\sigma_1(g)+{\rm Ad}_{g^{-1}}^*\mu-\mu=0,
\]
for every $g\in G$. Therefore, the result follows.

To summarise, if a co-adjoint cocycle of a given momentum map is a coboundary, then we can construct an ${\rm Ad}^*$-equivariant momentum map. However, the following proposition shows that for every momentum map there exists a Lie group action $\Delta:G\times \mathfrak{g} ^*\rightarrow  \mathfrak{g} ^*$ such that the momentum map becomes {\it $\Delta$-equivariant}, namely, one has that for every $g\in G$ the following diagram is commutative\begin{center}
 \begin{tikzcd}
   M 
   \arrow{r}{\Phi_g}
   \arrow{d}{\B J^\Phi}
   &
   M
   \arrow{d}{\B J^\Phi}
   \\
   \mathfrak{g}^*
   \arrow{r}{\Delta_g}
   &
   \mathfrak{g} ^*.
 \end{tikzcd}
\end{center}

\begin{proposition}\label{Prop::GenEqJ}
Let ${\bf J}^\Phi:M\rightarrow  \mathfrak{g}^*$ be a momentum map for a cosymplectic Lie group action $\Phi:G\times M\rightarrow  M$ with co-adjoint cocycle $\sigma$. Then,
\begin{enumerate}
    \item the map $\Delta:G\times \mathfrak{g}^*\ni(g,\mu)\mapsto \Delta_g(\mu):={\rm Ad}_{g^{-1}}^*\mu +\sigma(g)\in\mathfrak{g}^*$ is a Lie group action,
    \item the momentum map ${\bf J}^\Phi$ is $\Delta$-equivariant\,.

\end{enumerate}
\end{proposition}
\begin{proof}
First, since $\sigma(e)=0$, one has $
\Delta(e,\mu)={\rm Ad}^*_{e^{-1}}\mu+\sigma(e)=\mu$ and,
\begin{align*}
\Delta(g,\Delta(g',\mu))&={\rm Ad}_{g^{-1}}^*({\rm Ad}_{g'^{-1}}^*\mu +\sigma(g'))+\sigma(g)={\rm Ad}_{g^{-1}}^*{\rm Ad}_{g^{'-1}}^*\mu +{\rm Ad}_{g^{-1}}^*\sigma(g')+\sigma(g)\\&={\rm Ad}^*_{(gg')^{-1}}\mu +{\rm Ad}_{g^{-1}}^*\sigma(g')+\sigma(g)={\rm Ad}^*_{(gg')^{-1}}\mu +\sigma(gg')=\Delta(gg',\mu).
\end{align*}
Thus, $\Delta$ is a Lie group action on $\mathfrak{g}^*$, which proves $1$.
Second, from the definition of $\Delta$ and $\sigma$, we get
\[
\Delta_g\circ \B J^\Phi={\rm Ad}_{g^{-1}}^*\B J^\Phi+\sigma(g)=\B J^\Phi\circ \Phi_g,\qquad \forall g\in G,
\]
which shows that $\mathbf{J}^\Phi$ is $\Delta$-equivariant.
\end{proof}
Proposition \ref{Prop::GenEqJ} ensures that a general cosymplectic momentum map $\B J^\Phi$ gives rise to an equivariant momentum map relative to a new action $\Delta: G\x\mathfrak{g}^*\rightarrow \mathfrak{g}^*$, called an {\it affine action}. If ${\bf J}^\Phi$ is ${\rm Ad}^*$-invariant, then $\sigma=0$ and $\Delta={\rm Ad}^*$ becomes the co-adjoint action of $G$. In the following theorem, we shall analyse the commutation relations between the functions $\{J_\xi\}_{\xi\in \mathfrak{g}}$ associated with a cosymplectic momentum map $\B J^\Phi$.
\begin{theorem}
\label{Th::SigmabracketCo}
Let $\Phi:G\times M\rightarrow  M$ be a cosymplectic Lie group action relative to $(M,\omega,\eta)$ with a cosymplectic momentum map $\B J^\Phi:M\rightarrow \mathfrak{g}^*$ and let $\sigma:G\rightarrow  \mathfrak{g}^*$ be the co-adjoint cocycle of $\B J^\Phi$. Let us define
\begin{equation*}
\sigma_\tau:G\ni g\mapsto \langle\sigma(g),\tau\rangle\in\mathbb{R},\quad \Sigma:\mathfrak{g}\times\mathfrak{g}\ni(\xi_1,\xi_2)\mapsto T_e\sigma_{\xi_2}(\xi_1) \in\mathbb{R}, \quad\forall\tau\in\mathfrak{g}.
\end{equation*}
Then,
\begin{enumerate}
    \item the map $\Sigma$ is a skew-symmetric bilinear form on $\mathfrak{g}$ satisfying the following identity
    \[
    \Sigma(\xi,[\zeta,\nu])+\Sigma(\nu,[\xi,\zeta])+\Sigma(\zeta,[\nu,\xi])=0,\qquad \forall \xi,\zeta,\nu\in \mathfrak{g},
    \]
    \item $\Sigma(\xi,\nu)=\{J_\nu,J_\xi\}_{\omega,\eta}-J_{[\nu,\xi]}$ for all $\xi,\nu\in\mathfrak{g}$.
\end{enumerate}
\end{theorem}
\begin{proof}
Let us prove $2$. Taking the tangent map of $\sigma_\tau$ at $e$, we get
\begin{multline*}
     \Sigma(\xi,\tau)=T_e\sigma_\tau(\xi)=\frac{d}{ds}\bigg|_{s=0}\left(\lv \mathbf{J}^\Phi(\Phi_{\exp(s\xi)}(x)),\tau\rv-\lv{\rm Ad}^*_{\exp(-s\xi)}\mathbf{J}^\Phi(x),\tau\rv \right)\\=dJ_\tau(\xi_M)_x-\frac{d}{ds}\bigg|_{s=0}\lv \mathbf{J}^\Phi(x),{\rm Ad}_{\exp(-s\xi)}\tau\rv\\=-(\iota_{\tau_M}\iota_{\xi_M}\omega)_x-\lv \mathbf{J}^\Phi(x),[\tau,\xi]\rv=\{J_\tau,J_\xi\}_{\omega,\eta}(x)-J_{[\tau,\xi]}(x),
\end{multline*}
where the last equality stems from \eqref{Eq::PoissonStructure}. Since $X_{\{J_\tau,J_\xi\}_{\omega,\eta}}=-[X_{J_\tau},X_{J_\xi}]=-[\tau_M,\xi_M]=[\tau,\xi]_M$, the vector fields $X_{\{J_\tau,J_\xi\}_{\omega,\eta}}$ and $X_{J_{[\tau,\xi]}}$ have the same Hamiltonian function up to a constant. Therefore, $\Sigma$ does not depend on $x\in M$, which proves $2.$ Recall, that each $\sigma(g)$ does not depend on the point $x\in M$, which is employed to give a practical expression (\ref{Eq:SpeSigma}) for $\sigma$. 

Meanwhile, $1.$ stems from 
\begin{align*}
-\Sigma(\xi,[\zeta,\nu])=&\{J_\xi,J_{[\zeta,\nu]}\}_{\omega,\eta}-J_{[\xi,[\zeta,\nu]]}=\{J_\xi,\{J_\zeta,J_\nu\}_{\omega,\eta}-\Sigma(\nu,\zeta)\}_{\omega,\eta}-J_{[\xi,[\zeta,\nu]]}
\end{align*}
and the fact that $\Sigma(\nu,\eta)$ is a constant function, while $\{\cdot,\cdot\}_{\omega,\eta}$ and $[\cdot,\cdot]$ are anti-symmetric, bilinear, and satisfy the Jacobi identity.
\end{proof}

Recall that, for an ${\rm Ad}^*$-equivariant momentum map, $\sigma(g)=0$ for every $g\in G$. Thus, $\Sigma(\xi,\tau)=0$ for all $\xi,\tau\in\mathfrak{g}$ if  $\mathbf{J}^\Phi$ is ${\rm Ad}^*$-equivariant.  

The part 2. in Theorem \ref{Th::SigmabracketCo} retrieves, as a particular case, that if  $\mathbf{J}^\Phi$ is an ${\rm Ad}^*$-equivariant cosymplectic momentum map, then there exists a Lie algebra morphism $\mathfrak{g}\ni\xi\mapsto J_\xi\in C^\infty(M)$ .

\section{Cosymplectic Marsden--Weinstein reduction}
\label{Sec::CosymplReduction}

The following propositions are natural extensions to the cosymplectic realm of its analogues in the symplectic case. Proofs of some results can be found in \cite{Al89}, but since they are not widely available, e.g. \cite{Al89} is in not in English and illegible at times, we provide their proofs here.

The following proposition shows that the momentum map ${\bf J}^\Phi: M\rightarrow \mathfrak{g}^*$ related to $(M^\omega_\eta,h,{\bf J}^\Phi)$ is conserved for the dynamics of the vector fields $\nabla h$, $X_h$, and $E_h$. In other words, the flows of $\nabla h$, $X_h$, and $E_h$ leave ${\bf J}^\Phi$ invariant for every $h\in C^\infty(M)$. It is worth noting that Proposition \ref{Prop:InvCos} is, as far as we know, new. Nevertheless, it complements a partial result already given in \cite{LS15}.

\begin{proposition}\label{Prop:InvCos}
Let $(M^\omega_\eta,h,{\bf J}^\Phi)$ be a $G$-invariant cosymplectic Hamiltonian system and let $F:\mathbb{R}\times M\ni(s,m)\mapsto F_s(m)=F(s,m)\in M$ be the flow of $\nabla h$. Then, ${\bf J}^\Phi\circ F_s={\bf J}^\Phi$ for every $s\in \mathbb{R}$. Analogous results apply to the flows of $E_h$ and $X_h$.
\end{proposition}
\begin{proof}
Since $h$ is $G$-invariant, $\xi_Mh=0$ for every $\xi\in\mathfrak{g} $. Therefore,
\begin{equation*}
\frac{d}{ds}\bigg|_{s=0}J_\xi\circ F_s=\iota_{\nabla h}dJ_\xi=\iota_{{X_h}+(Rh)R}dJ_\xi=\iota_{X_h}dJ_\xi=\iota_{X_h}\iota_{\xi_M}\omega=\iota_{\xi_M}((Rh)\eta-dh)=0,
\end{equation*}
for every $\xi\in\mathfrak{g} $. Thus, $J_\xi\circ F_s=J_\xi$, for every $\xi\in\mathfrak{g} $ and every $s\in\mathbb{R}$. This implies that ${\bf J}^\Phi\circ F_s={\bf J}^\Phi$ for all $s\in\mathbb{R}$. 
Meanwhile, if $L$ is the flow of $E_h$, then
\begin{equation*}
\frac{d}{ds}\bigg|_{s=0}J_\xi\circ L_s=\iota_{E_h}dJ_\xi =\iota_{{X_h}+R}\,dJ_\xi=\iota_{X_h}dJ_\xi=\iota_{X_h}\iota_{\xi_M}\omega=\iota_{\xi_M}((Rh)\eta-dh)=0,
\end{equation*}
for every $\xi\in\mathfrak{g} $. Hence, ${\bf J}^\Phi\circ L_s={\bf J}^\Phi$ for every $s\in \mathbb{R}$.  Since $\iota_{X_h+R}dJ_\xi=\iota_{X_h}dJ_\xi$, it follows that ${\bf J}^\Phi\circ K_s={\bf J}^\Phi$ for the diffeomorphisms $K_s$ of the one-parametric group of diffeomorphisms of $X_h$.  
\end{proof}

It's important to note that the theorem above guarantees the conservation of momentum of a $G$-invariant cosymplectic Hamiltonian system.

\begin{remark}Recall that, for $(M:=\mathbb{R}\times P,\omega_P,\eta_T)$, the vector field $X_h$ on $M$ can be considered as a  $t$-dependent vector field on $P$. The integral curves of $X_h$ as a $t$-dependent vector field are the integral curves of $E_h$ of the form $t\mapsto (t,x(t))$. Therefore, Proposition \ref{Prop:InvCos} also applies to the integral curves of $X_h$ as a $t$-dependent vector field, which was proved in \cite{LZ21} .
\end{remark}

The following lemma is a generalisation of a well-known result in symplectic geometry, which is crucial to obtain the cosymplectic Marsden--Weinstein reduction theorem. Recall that a {\it weakly regular value} of ${\bf J}^\Phi:M\rightarrow  \mathfrak{g} ^*$ is a point $\mu\in\mathfrak{g} ^*$ such that ${\bf J}^{\Phi-1}(\mu)$ is a submanifold in $M$ and $T_x{\bf J}^{\Phi-1}(\mu)=\ker T_x{\bf J}^\Phi$ for every $x\in {\bf J}^{\Phi-1}(\mu)$ (see \cite{Al89}).

\begin{lemma} 
\label{Lemm::NonAdPerp}
Let $\mu\in\mathfrak{g} ^*$ be a weak regular value of a cosymplectic momentum map ${\bf J}^\Phi:M\rightarrow  \mathfrak{g}^*$ and let $G^\Delta_\mu$ be the isotropy group at $\mu\in \mathfrak{g}^*$ of the affine action $\Delta :G\times \mathfrak{g}^*\rightarrow  \mathfrak{g}^*$ relative to the co-adjoint cocycle $\sigma:G\rightarrow  \mathfrak{g}^*$ of ${\bf J}^\Phi$, then, for every $x\in {\bf J}^{\Phi-1}(\mu)$, one has 
\begin{enumerate}
\item $T_{x}(G^\Delta_{\mu} x)=T_{x}(G x)\cap T_{x}({\bf J}^{\Phi-1}(\mu))$, 
\item $T_{x}({\bf J}^{\Phi-1}(\mu))=T_{x}(Gx)^{\perp_\omega}$,
\item $\left(T_x\left({\bf J}^{\Phi-1}(\mu)\right)\right)^{\perp_\om}=T_x(Gx)\oplus\<R_x\>$.
\end{enumerate}
\end{lemma}

\begin{proof}
Let us assume that $(\xi_M)_x\in T_x{\bf J}^{\Phi-1}(\mu)$. Since $T_x({\bf J}^{\Phi-1}(\mu))=\ker T_x{\bf J}^\Phi$, then,
\begin{multline*}
(\iota_{\xi_M}dJ_\tau)_x\!=\!\frac{d}{du}\bigg|_{u=0}\!\!\!\!\!J_{\tau}(\Phi(\exp(u\xi),x))\! =\!\left\langle\!\frac{d}{du}\bigg|_{u=0}\!\!\!\!{\bf J}^\Phi(\Phi(\exp(u\xi),x)),\tau\!\right\rangle\!=\!\left\langle\!\frac{d}{du}\bigg|_{u=0}\!\!\!\!\Delta_{\exp(u\xi)}{\bf J}^\Phi(x),\tau\!\right\rangle\!=\!0,
\end{multline*}
for every $\tau\in \mathfrak{g} $ if and only if $\xi\in\mathfrak{g} ^\Delta_\mu$, where $\mathfrak{g}^\Delta_\mu$ is the Lie algebra of $G^\Delta_\mu$. This proves $1$. 

Let us  prove $2$. Since ${\bf J}^\Phi$ is a cosymplectic momentum map, we have
\[
\om_x((\xi_M)_x,v_x)=(dJ_\xi)_x(v_x)=\left\langle T_x{\bf J}^\Phi(v_x),\xi\right\rangle,\qquad \forall x\in M,\quad \forall v_x\in T_xM,\quad \forall \xi\in\mathfrak{g} .
\]
Thus, $v_x\in \ker T_x{\bf J}^\Phi=T_x\left({\bf J}^{\Phi-1}(\mu)\right)$ if and only if $\langle T_x{\bf J}^\Phi(v_x),\xi\rangle=0$ for all $\xi\in\mathfrak{g}$, and $T_x({\bf J}^{\Phi-1}(\mu))=(T_x(Gx))^{\perp_\omega}$  for all  $x\in {\bf J}^{\Phi-1}(\mu)$. 

To prove $3.$, let us take $X=\xi_M+\lambda R$ for any  $\lambda\in \mathbb{R}$. Then, for any $v_x\in \ker T_x{\bf J}^\Phi$, we have
\[
\om_x(X_x,v_x)=(dJ_\xi)_x(v_x)=0,\qquad \forall\xi\in\mathfrak{g} ,
\]
and $ T_x(G x)\oplus\<R_x\>\subset \left(T_x({\bf J}^{\Phi-1}(\mu))\right)^{\perp_\om}$. On the other hand, for every $x\in\mathbf{J}^{\Phi-1}(\mu)$, since $R_x$ does take values in $T_x(\mathbf{J}^{\Phi-1}(\mu))$ but is not tangent to $Gx$, then
\begin{equation*}
\left(T_x({\bf J}^{\Phi-1}(\mu))\right)^{\perp_\omega}=\left(\dim T_x(Gx)\right)^{\perp_\omega\perp_\omega}=T_x(Gx)\oplus\<R_x\>,\qquad \forall x\in {\bf J}^{\Phi-1}(\mu).
\end{equation*}
Thus, statement $3.$ holds.
\end{proof}

The following theorem is a generalisation of the classical Marsden--Weinstein reduction theorem to the cosymplectic realm. It follows the ideas of the proof given in \cite{Al89}.

\begin{theorem}\label{Th:CoSymRed}
Let $\Phi: G\times M\rightarrow  M$ be a cosymplectic Lie group action on the cosymplectic manifold $(M,\omega,\eta)$ associated with a cosymplectic momentum map ${\bf J}^\Phi: M\rightarrow \mathfrak{g}^*$. Assume that $\mu\in\mathfrak{g}^*$ is a weakly regular value of ${\bf J}^\Phi$ and let ${\bf J}^{\Phi-1}(\mu)$ be quotientable, i.e. $M^\Delta_\mu:={\bf J}^{\Phi-1}(\mu)/G^\Delta_{\mu}$ is a manifold and $\pi_\mu:{\bf J}^{\Phi-1}(\mu)\rightarrow  M^\Delta_\mu$ is a submersion. Let $\iota_\mu:{\bf J}^{\Phi-1}(\mu)\hookrightarrow M$ be the natural immersion and let $\pi_\mu:{\bf J}^{\Phi-1}(\mu)\rightarrow  M^\Delta_\mu$ be the canonical projection. Then, there exists a unique cosymplectic manifold $(M_\mu^\Delta,\omega_\mu,\eta_\mu)$ such that
\begin{equation}
\label{Eq::CoSymRed}
\iota_\mu^*\omega=\pi_\mu^*\omega_\mu,\qquad \iota_\mu^*\eta=\pi_\mu^*\eta_\mu.
\end{equation}
\end{theorem}

\begin{proof}
The quotient space $M^\Delta_\mu={\bf J}^{\Phi-1}(\mu)/G^\Delta_\mu$ is a manifold because $\mathbf{J}^{\Phi-1}(\mu)$ is quotientable. Meanwhile, $\pi_\mu:{\bf J}^{\Phi-1}(\mu)\rightarrow  M^\Delta_\mu$ is a surjective submersion by assumption. Then, $\ker T\pi_\mu$ is a subbundle of $T({\bf J}^{\Phi-1}(\mu))$. From our assumptions, $\Phi_g$ is a cosymplectomorphism for every $g\in G$ and then $\mathcal{L}_{\xi_M}\omega=0$
and $\mathcal{L}_{\xi_M}\eta=0$ for every $\xi\in\mathfrak{g} $. This ensures that $\mathcal{L}_{\xi_{\mathbf{J}^{\Phi-1}(\mu)}}\iota_\mu^*\omega=0$ and $\mathcal{L}_{\xi_{\mathbf{J}^{\Phi-1}(\mu)}}\iota_\mu^*\eta=0$ for every $\xi\in\mathfrak{g} ^\Delta_\mu$, where $\mathfrak{g}^\Delta_\mu$ is the Lie algebra of $G_\mu^\Delta$ and $\xi_{{\bf J}^{\Phi-1}(\mu)}$ is the fundamental vector field of the restriction of the action of $G_\mu^\Delta$ to ${\bf J}^{\Phi-1}(\mu)$ via $\Phi$. Similarly, for every vector field $Y_{{\bf J}^{\Phi-1}(\mu)}$ on ${\bf J}^{\Phi-1}(\mu)$ and  tangent to ${\bf J}^{\Phi-1}(\mu)$, one can consider that there exists some vector field $Y$ on $M$ coinciding with $Y_{{\bf J}^{\Phi-1}(\mu)}$ on ${\bf J}^{\Phi-1}(\mu)$. Then,
\[
\iota_{Y_{{\bf J}^{\Phi-1}(\mu)}}\iota_{\xi_{{\bf J}^{\Phi-1}(\mu)}}\iota_\mu^*\om=\iota_\mu^*(\iota_{Y}\iota_{\xi_M}\om)=\iota_{\mu}^*(\iota_{Y}dJ_\xi)=0,
\]
and
\[
\iota_{Y_{{\bf J}^{\Phi-1}(\mu)}}\iota_{\xi_{{\bf J}^{\Phi-1}(\mu)}}\iota_\mu^*\eta=\iota_\mu^*(\iota_{Y}\iota_{\xi_M}\eta)=0.
\]
These conditions guarantee the existence of $\om_\mu\in\Omega^2(M^\Delta_\mu)$ and $\eta_\mu\in\Omega^1(M^\Delta_\mu)$ satisfying \eqref{Eq::CoSymRed}. Therefore, $\om_\mu$ and $\eta_\mu$ are unique, due to the fact that $\pi^*_\mu$ is injective, and well defined. Note that $\om_\mu$ and $\eta_\mu$ are closed, since $\om$ and $\eta$ are closed and \eqref{Eq::CoSymRed} are satisfied. Recall that $\iota_R dJ_\xi=0$ for every $\xi\in\mathfrak{g} $ due to the definition of ${\bf J}^{\Phi}$ . Hence, $R$ is tangent to ${\bf J}^{\Phi-1}(\mu)$. Thus, there is a vector field $\widetilde{R}$ on ${\bf J}^{\Phi-1}(\mu)$ such that $\widetilde{R}=R|_{{{\bf J}^{\Phi-1}(\mu)}}$. Since $\Phi_{g*}\widetilde{R}=\widetilde{R}$ and $ \mathcal{L}_{\xi_{\mathbf{J}^{\Phi-1}(\mu)}}\widetilde{R}=0$ for every $g\in G^\Delta_\mu$ and $\xi\in\mathfrak{g} ^\Delta_\mu$, there exists a well defined vector field $R_\mu$ on $M^\Delta_\mu$ such that $R_\mu=\pi_{\mu*}\widetilde{R}$. Moreover,
\[
\pi_\mu^*(\iota_{R_\mu}\eta_\mu)=\iota_{\widetilde{R}}\,\pi^*_\mu\eta_\mu=\iota_\mu^*(\iota_R\eta)=1,
\]
and
\[
\pi_\mu^*(\iota_{R_\mu}\om_\mu)=\iota_{\widetilde{R}}\pi^*_\mu\om_\mu=\iota^*_\mu(\iota_R\om)=0.
\]
Thus, $\iota_{R_\mu}\eta_\mu=1$ and $\iota_{R_\mu}\om_\mu=0$. To prove that $\ker\omega_\mu\oplus\ker\eta_\mu=TM^\Delta_\mu$, we will show that 
\[
\flat_\mu:X_\mu\in TM^\Delta_\mu\mapsto \iota_{X_\mu}\om_\mu+(\iota_{X_\mu}\eta_\mu)\eta_\mu\in T^*M^\Delta_\mu
\]
is an isomorphism. To show that $\flat_\mu$ is injective, let us assume that there is a vector field $X_\mu$ taking values in $\ker\flat_\mu$. Since $\iota_{R_\mu}\flat_\mu(X_\mu)=0$, then $\iota_{X_\mu}\eta_\mu=0$ and $\iota_{X_\mu}\omega_\mu=0$. Hence, there exist $\widetilde{X}\in \mathfrak{X}({\bf J}^{\Phi-1}(\mu))$ and $X\in\mathfrak{X}(M)$, such that $\pi_{\mu*}\widetilde{X}=X_\mu$ and $\widetilde{X}=X|_{{\bf J}^{\Phi-1}(\mu)}$. However, $\pi_\mu^*(\iota_{X_\mu}\om_\mu)=\iota_{X}\omega|_{\mathbf{J}^{\Phi-1}(\mu)}=0$ implies that $X$ takes values in $\left(T_x({\bf J}^{\Phi-1}(\mu))\right)^{\perp_\om}=T_x(Gx)\oplus\<R_x\>$ for all $x\in\mathbf{J}^{\Phi-1}(\mu)$. Therefore, $X_x=(\xi_M)_x+\lambda R_x$ for some $\xi\in \mathfrak{g}^\Delta_\mu$ and $\lambda\in\mathbb{R}$ depending on  $x\in\mathbf{J}^{\Phi-1}(\mu)$. Since $\eta_\mu(T_x\pi_\mu X_x)=0$, one gets $\lambda=0$. Then, $(X_{\mu})_{\pi_\mu(x)}=T_x\pi_\mu X_x=T_x\pi_\mu(\xi_M)_x=0$.

Thus, $\ker\flat_\mu=0$ and $\flat_\mu$ is injective. The map $\flat_\mu$ is also surjective as a result of a dimension analysis. Hence, $(M_\mu^\Delta,\om_\mu,\eta_\mu)$ is a cosymplectic manifold. 
\end{proof}

The following result will be interesting for physical applications of our cosymplectic energy-momentum method to be developed in the following sections. Note that the existence of a cosymplectic momentum map for $\Phi:G\times M\rightarrow  M$ relative to $(M:=T\times P,\omega_P,\eta_T)$ implies that $\Phi$ can be restricted to a Lie group action of $G$ on $P$. In fact, this is due to the fact that the fundamental vector fields of $\Phi$ are required to take values in $\ker \eta$.

\begin{corollary}
\label{Cor::De}
Assume the assumptions of Theorem \ref{Th:CoSymRed} to remain valid. Additionally, consider the cosymplectic manifold to be $(T\times P,\omega_P,\eta_T)$. Then, 
\[
{\bf J}^{\Phi-1}(\mu)\simeq T\times \pi_P({\bf J}^{\Phi -1}(\mu)),\qquad M^\Delta_\mu\simeq T\times P^\Delta_\mu,
\]
where $P^\Delta_\mu:=\pi_P(\mathbf{J}^{\Phi-1}(\mu))/G^\Delta_\mu$.
\end{corollary}
\begin{proof}
By Definition \ref{Def:MomentumMap}, one has that $\iota_R dJ_\xi=0$ and $\mathcal{L}_{R}dJ_\xi=0$ for every $\xi\in\mathfrak{g}$. Hence, $dJ_\xi$ is a basic one-form with respect to $\pi_P$. Thus, for each $\xi\in\mathfrak{g} $, there exists $\widetilde{J}_\xi\in C^\infty(P)$ such that $\pi_P^*\widetilde{J}_\xi=J_\xi$. Therefore, there exists $\widetilde{\mathbf{J}}^\Phi:P\rightarrow \mathfrak{g}^*$ such that $\mathbf{J}^\Phi=\widetilde{\mathbf{J}}^\Phi\circ \pi_P$
 and  $\mathbf{J}^{\Phi-1}(\mu)=T\x\widetilde{\mathbf{J}}^{\Phi-1}(\mu)=T\x\pi_P(\mathbf{J}^{\Phi-1}(\mu))$. 
 
 Let $\widetilde{R}$ be the restriction of $R$ to ${\bf J}^{\Phi-1}(\mu)$. Note that $\Phi^\mu:G_\mu^\Delta\times T\x\mathbf{J}^{\Phi-1}(\mu)\rightarrow  T\times \mathbf{J}^{\Phi-1}(\mu)$ is a well-defined Lie group action obtained by restricting the action $\Phi$ of $G^\Delta_\mu$ on $T\times P$ to $T\times \mathbf{J}^{\Phi-1}(\mu)$. Since $\iota_{\xi_{\mathbf{J}^{\Phi-1}(\mu)}}\iota_\mu^*\eta_T=0$ 
  for every $\xi\in \mathfrak{g}^\Delta_\mu$, there exists a Lie group action $\widetilde{\Phi}^\mu:G_\mu^\Delta\x\widetilde{\mathbf{J}}^{\Phi-1}(\mu)\rightarrow  \widetilde{\mathbf{J}}^{\Phi-1}(\mu)$ satisfying $\widetilde{\Phi}^\mu_g\circ\pi_P=\pi_P\circ\Phi^\mu_g$ for every $g\in G_\mu^\Delta$. Hence, $\Phi^\mu_g(t,p)=(t,\widetilde{\Phi}^\mu_g(p))$ for every $t\in T$ and $p\in \widetilde{\bf J}^{\Phi-1}(\mu)$. Thus, $\mathbf{J}^{\Phi-1}(\mu)/G^\Delta_\mu=(T\times \widetilde{\mathbf{J}}^{\Phi-1}(\mu))/G^\Delta_\mu=T\times (\widetilde{\mathbf{J}}^{\Phi-1}(\mu)/G^\Delta_\mu)$.
\end{proof}
 
\begin{proposition} \label{Pr::RedHam}
    Consider the assumptions of Theorem \ref{Th:CoSymRed} to hold for $(M_\eta^\omega,h, {\bf J}^\Phi)$. Then, the restriction of $E_h$ to ${\bf J}^{\Phi-1}(\mu)$ is projectable onto $M^\Delta_\mu={\bf J}^{\Phi-1}(\mu)/G^\Delta_\mu$ and $\pi_{\mu*} (E_h|_{\mathbf{J}^{\Phi-1}(\mu)})=E_{k_\mu}$, where $k_\mu$ is the only function on $M^\Delta_\mu$ such that $\pi^*_\mu k_\mu=\iota_\mu^* h$.  
\end{proposition}
\begin{proof}
By Proposition \ref{Prop:InvCos}, the vector field $E_h$ is tangent to ${\bf J}^{\Phi-1}(\mu)$. Moreover, for every $\xi\in \mathfrak{g}_\mu^\Delta$, one has that $\xi_M=X_{J_\xi}$. Since $RJ_\xi=0$, Proposition \ref{Prop::gradXR} entails that $[\xi_M,R]=0$. Therefore, 
\[
[\xi_M,E_h]=[\xi_M,R+X_h]=[\xi_M,X_h],\qquad \forall \xi\in \mathfrak{g}.
\]
Then,  (\ref{Eq::AntiMorphism}) yields
\[
[\xi_M,X_h]=X_{\{h,J_\xi\}}=X_{\xi_Mh}=0.
\]
Hence, $E_h|_{\mathbf{J}^{\Phi-1}(\mu)}$ is projectable onto $M^\Delta_\mu$. By Theorem \ref{Th:CoSymRed}, the differential forms $\iota^*_\mu\omega$ and $\iota^*_\mu\eta$ are also projectable and, it follows  from its proof that
\[
\iota_{\pi_{\mu^*}(E_h|_{{\bf J}^{\Phi-1}(\mu)})}\omega_\mu=dk_\mu-(R_\mu k_\mu)\eta_\mu,\qquad \iota_{\pi_{\mu^*}(E_h|_{{\bf J}^{\Phi-1}}(\mu))}\eta_\mu=1,
\]
with Hamiltonian function $k_\mu\in C^\infty(M^\Delta_\mu)$ given by the condition $\pi_\mu^*k_\mu=\iota_\mu^*h$, whose existence is ensured by the fact that $h$ is invariant relative to $G^\Delta_\mu$. This yields that $\pi_{\mu^*}E_h|_{\mathbf{J}^{\Phi-1}(\mu)}$ is an evolutionary vector field relative to $(M^\Delta_\mu,\omega_\mu,\eta_\mu)$. 
\end{proof}

\section{Lyapunov stability}
\label{Sec::Stab}

Let us introduce basic notions and theorems on the stability of dynamical systems that will be used in our cosymplectic formulation of the time-dependent energy-momentum method \cite{LZ21, Za21}.

As we assume that all manifolds considered in this work are paracompact and Hausdorff, they, therefore, admit a Riemannian metric $\bm{g}$ \cite{Le09}. The topology induced by $\bm{g}$  is the same as the topology of the manifold $P$ \cite{Za21}. The metric $\bm{g}$ induces a distance in $P$ so that the distance between two points $x_1,x_2\in P$ is given by 
\begin{equation}
    d(x_1,x_2):=\!\!\!\!\!\!\inf_{\substack{\tiny \gamma:[0,1]\Rightarrow  P\\\gamma(0)=x_1,\gamma(1)=x_2}}\!\!\!\!\!\! \!\!\!l_{\bm{g}}(\gamma),
\end{equation}
where $l_{\bm{g}}(\gamma)$ is the length of a  smooth curve $\gamma:[0,1]\rightarrow  P$. The fact that the topology of the manifold is the one induced by any Riemannian metric $\bm{g}$ implies that our results are independent of $\bm{g}$. In particular, this allows one to use open coordinated subsets and standard norms, which much simplifies the practical application of developed methods (see \cite{Le13,LZ21,Za21} for details).  

Let $X:T\times P\ni(t,x)\mapsto X(t,x)\in TP$ be a $t$-dependent vector field on $P$, namely a $t$-parametric family of vector fields $X_t:P\ni x\mapsto X_t(x)\in TP$ with $t\in T$ (see \cite{LS20} for details on them and vector fields along projections). As $X$ is assumed to be smooth, the theorem of existence and uniqueness of solutions  can be applied to 
\begin{equation}\label{Eq::NonAutDyn}
    \frac{dx}{dt}=X(t,x),\quad \forall x\in P,\quad\forall t\in T.
\end{equation}
Let us denote  $\overline{\mathbb{R}}_+:=\mathbb{R}_+\cup\{0\}$ and let us set $I_t:=[t,\infty[$ for every $t\in\mathbb{R}$ and $I_{-\infty}=\mathbb{R}$. A point $x_e\in P$ is an {\it equilibrium point} of \eqref{Eq::NonAutDyn} if $X(t,x_e)=0$ for every $t\in T$. In that case, it is also said that $x_e\in P$ is an \textit{equilibrium point of the $t$-dependent vector field} $X$. Let us hereafter assume, when talking about stability, that $T=\mathbb{R}$. An equilibrium point $x_e$ is {\it stable} from $t^0\in \mathbb{R}$ if, for every $t_0\in I_{t^0}$ and for every ball $B_{x_e,\varepsilon}:=\{x\in P:d(x,x_e)<\epsilon\}$, there exists a radius $\delta(t_0,\varepsilon)$ such that every solution $x(t)$ of \eqref{Eq::NonAutDyn} with initial condition $x(t_0)\in B_{x_e,\delta(t_0,\varepsilon)}$ is contained in $B_{x_e,\varepsilon}$ for $t>t_0$. In further applications, we assume $t^0=-\infty$, if not otherwise stated. An equilibrium point $x_e$ is {\it uniformly stable} if one can set the radius $\delta(t_0,\varepsilon)$ to be independent of $t_0$. An equilibrium point $x_e$ is {\it unstable} if it is not stable.

An equilibrium point $x_e$ is said to be {\it asymptotically stable} if it is stable and, for every $t_0\in I_{t^0}$, there exists a radius $r(t_0)$ such that every solution $x(t)$ of \eqref{Eq::NonAutDyn} with initial condition $x(t_0)\in  B_{x_e,r(t_0)}$ converges to $x_e$ when $t$ tends to infinity. Moreover, an equilibrium point $x_e$ is {\it uniformly asymptotically stable} if it is asymptotically stable, the radius $\varepsilon=r(t_0)$ can be chosen so that it is independent of $t_0$ and the convergence of $x(t)$ to $x_e$ is uniform relative to $x\in B_{x_e,\varepsilon}$ and $t\in I_{t^0}$.

Let us introduce notions that are necessary to formulate the Basic Lyapunov Theorem on manifolds (see \cite{LZ21} and references therein).

\begin{definition}
\label{lpdf}
A function $\mathcal{M}:\mathbb{R}\times P\rightarrow \mathbb{R}$ is a {\it locally positive definite function} ({\it lpdf}) {\it at an equilibrium point} $x_e$ from $t^0\in \mathbb{R}$ if, there exists $r>0$ and a continuous, strictly increasing function $\alpha:\overline{\mathbb{R}}_+\rightarrow \mathbb{R}$ with $\alpha(0)=0$, such that
\begin{equation*}
    \mathcal{M}(t,x_e)=0,\quad \mathcal{M}(t,x)\geq \alpha(d(x,x_e)),\quad\forall t\in I_{t^0},\quad \forall x\in B_{x_e,r}.
\end{equation*}

Meanwhile, $\mathcal{M}:\mathbb{R}\times P\rightarrow \mathbb{R}$ is called {\it  decrescent} {\it at an equilibrium point} $x_e$ from $t^0\in \mathbb{R}$ if there exists $s>0$ and a continuous, strictly increasing function $\beta:\overline{\mathbb{R}}_+\rightarrow \mathbb{R}$ with $\beta(0)=0$ such that
\begin{equation*}
    \mathcal{M}(t,x)\leq \beta(d(x,x_e)),\quad\forall t\in I_{t^0},\quad \forall x\in B_{x_e,s}.
\end{equation*}
\end{definition}

If $t^0$ is avoided when describing lpdf functions, we will assume that $t^0=0$.

Let $\dot{\mathcal{M}}:\mathbb{R}\times P\rightarrow  \mathbb{R}$ be defined by 
\begin{equation}\label{Eq:Mdot}
\dot{\mathcal{M}}( t, x):=\frac{\partial \mathcal{M}}{\partial t}( t, x)+\sum_{i=1}^{\dim P}\frac{\partial \mathcal{M}}{\partial x^i}( t, x) X^i( t, x),\qquad \forall (t,x)\in \mathbb{R}\times P,
\end{equation}
where $\{x^1,\ldots,x^{\dim P}\}$ is a local coordinate system defined on a neighbourhood of the point $ x\in P$ and $X=\sum_{i=1}^{\dim P}X^i\frac{\partial}{\partial x^i}$. 

Let us present the {\it basic Lyapunov theorem on manifolds} \cite{LZ21, Za21}, which is a generalisation of the Lyapunov theorem for vector spaces. This theorem allows for the analysis of types of equilibrium points using associated functions. Note that the types of equilibrium points do not depend on the Riemannian metric defined on $M$ as every Riemannian metric generates the same topology.

\begin{theorem}
\label{Th:BasicTheoremOfLyapunov} {\bf (Lyapunov Theorem on manifolds \cite{LZ21,Kh87,Vi02,Za21})}
Let $\mathcal{M}:\mathbb{R}\times P\rightarrow  \mathbb{R}$ be a non-negative function, let $x_e\in P$ be an equilibrium point of (\ref{Eq::NonAutDyn}), and let $\dot{\mathcal{M}}$ stand for the function (\ref{Eq:Mdot}). Then, 
\begin{enumerate}
    \item If $\mathcal{M}$ is   lpdf from $t^0$ and $\dot{\mathcal{M}}(t,x)\leq 0$ for $x$ locally around $x_e$ and for all $t\in I_{t^0}$, then $x_e$ is stable.
    
    \item If $\mathcal{M}$ is  lpdf and decrescent from $t^0$ and   $\dot{\mathcal{M}}(t,x)\leq 0$ locally around $x_e$ and for all $t\in I_{t^0}$, then $x_e$ is uniformly stable.
    
    \item If $\mathcal{M}$ is  lpdf and decrescent from $t^0$, and $-\dot{\mathcal{M}}(t,x)$ is lpdf around $x_e$ and for all $t\in I_{t^0}$, then $x_e$ is uniformly asymptotically stable.
\end{enumerate}
\end{theorem}

\section{Characterisation of  relative equilibrium points}
\label{Sec::CharStrongEqP}

It is hereafter assumed that $\mu\in \mathfrak{g}^*$ is a weakly regular value of ${\bf J}^\Phi$. Additionally, it is assumed that the isotropy subgroup $G_\mu^\Delta$ of the element $\mu\in\mathfrak{g} ^*$ relative to the affine action, defined in Proposition \ref{Prop::GenEqJ}, acts via $\Phi$ on ${\bf J}^{\Phi-1}(\mu)$ in a {\it quotientable manner}, namely ${\bf J}^{\Phi-1}(\mu)/G_\mu^\Delta$ is a manifold and the projection $\pi:{\bf J}^{\Phi-1}(\mu)\rightarrow  {\bf J}^{\Phi-1}(\mu)/G_\mu^\Delta$ is a submersion. To guarantee that ${\bf J}^{\Phi-1}(\mu)/G^\Delta_\mu$ is a manifold, one may assume that $G^\Delta_\mu$ acts freely and properly on ${\bf J}^{\Phi-1}(\mu)$ \cite{AM78}. However, these assumptions can be weakened (cf. \cite{AM78}). For the reasons already mentioned in Section \ref{SubSec::CosStructures}, cosymplectic manifolds to be studied hereafter are of form $(T\times P,\omega_P ,\eta_T)$. To simplify the notation, the sub-indexes of the differential forms $\omega_P$ and $\eta_T$ will be hereafter skipped.

Let us extend Poincar\'{e}'s  {\it relative equilibrium point} notion (see \cite[p. 306]{AM78}) for a $t$-independent Hamiltonian function to the realm of cosymplectic Hamiltonian systems. There are several manners of giving such a  generalisation. In fact, we will introduce a new second one in Section \ref{Sec::RDTP}.

\begin{definition}
\label{Def::CosStrongRelEqPointn}
A point $z_e\in P$ is a {\it relative equilibrium point} of  $((T\times P)^\omega_\eta,h,{\bf J}^\Phi)$ if there exists a curve $\xi(t)\in\mathfrak{g} $ so that
\begin{equation}
(X_h)_{(t,z_e)}=(\xi(t)_M)_{(t,z_e)},\qquad \forall t\in T.
\end{equation}
\end{definition}

If $T=\mathbb{R}$, Definition \ref{Def::CosStrongRelEqPointn} can be reformulated using the integral curve of an evolution vector field. In fact, a point $z_e\in P$ is a {\it relative equilibrium point} of $((\mathbb{R}\times P)_\eta^\omega,h,\mathbf{J}^\Phi)$ if, for each $t_0\in \mathbb{R}$, there exists some curve $\xi_{t_0}(s)$ in $\mathfrak{g} $ such that 
\begin{equation} \label{Eq::StrRelEqForR}
s\in \mathbb{R}\mapsto \Phi(\exp(\xi_{t_0}(s)),(t_0+s,z_e))\in \mathbb{R}\times P,
\end{equation}
is the integral curve of $E_h$ with initial condition $(t_0,z_e)$. Equivalently, $\Phi(\exp(\xi_{t_0}(t),z_e))$ is a solution of the Hamiltonian equations for $h$ with initial condition $z_e$ at $t=t_0$. In short, this tells us that the evolution for the Hamilton equations of $h$ is given by the symmetries of the problem encoded in $\Phi$. A similar result could be proved for a general $T$, but an analogue local version of \eqref{Eq::StrRelEqForR} should be written as $T$ does not need to admit a global coordinate, e.g. this happens when $T=\mathbb{S}^1$.

Hereafter, we assume $t_0=0$ unless otherwise stated, and to simplify the notation, we will write $\xi_{t_0=0}(t)$ as $\xi(t)$.

The points in $z_e\in P$ that give rise to equilibrium points $(t,z_e)$ of the Hamilton equations of $((T\times P)^\om_\eta,h,{\bf J}^\Phi)$ for every $t\in T$, are particular cases of relative equilibrium points. In fact, in such a case, $(X_h)_{(t,z_e)}=0$ for every $t\in T$, and the point $(t,z_e)$ is invariant relative to the dynamic induced by $X_h$. Therefore, $\xi(t)\in \mathfrak{g}$ can be chosen in \eqref{Def::CosStrongRelEqPointn} to be equal to zero.

\begin{proposition}
\label{Prop:SingPPCos}
Every integral curve, $m(t)=(t,z(t))$ of $E_h$ with respect to $((T\times P)^\omega_\eta,h,{\bf J}^\Phi)$ such that $z(t_0)=z_e$  for a relative equilibrium point $z_e\in P$ with $\mu_e={\bf J}^\Phi(t_0,z_e)$ and $t_0\in T$, projects onto the single point $(\pi_{P^\Delta_{\mu_e}}\circ \pi_{\mu_e})(t_0,z_e)$, i.e. $(\pi_{P^\Delta_{\mu_e}}\circ \pi_{\mu_e})(m(t))=(\pi_{P^\Delta_{\mu_e}}\circ \pi_{\mu_e})(t_0,z_e)$ for every $t\in T$.
\end{proposition}

\begin{proof}
Proposition \ref{Prop:InvCos} yields that every integral curve $m(t)$ to $E_h$ is fully contained within the submanifold ${\bf J}^{\Phi-1}(\mu_e)$. Proposition \ref{Pr::RedHam}  shows that $m(t)$ projects, via $\pi_{\mu_e}$, onto a curve in $M^\Delta_{\mu_e}:={\bf J}^{\Phi-1}(\mu_e)/G^\Delta_{\mu_e}\simeq T\times P_{\mu_e}^\Delta$, where $G^\Delta_{\mu_e}$ is the isotropy subgroup of $\mu_e\in\mathfrak{g} ^*$ relative to the affine action $\Delta$. Since $z_e\in P$ is a  relative equilibrium point and ${\bf J}^\Phi$ is $\Delta$-equivariant, it turns out that
\[
0=T_{(t,z_e)}{\bf J}^{\Phi}(E_{h})_{(t,z_e)}=T_{(t,z_e)}{\bf J}^{\Phi}(R+\xi(t)_{M})_{(t,z_e)}=(\xi(t)^\Delta_{\mathfrak{g}^*})_{\mu_e},\qquad \forall t\in T,
\]
for some curve $\xi(t)$ in $\mathfrak{g}$. Hence, the curve $\xi(t)$ is contained in $\mathfrak{g}^\Delta_{\mu_e}$.

Note that $\pi_{\mu_e}(m(t))$ is the integral curve to the vector field on $T\times P_{\mu_e}^\Delta$ given by  $R_{\mu_e}+Y_{\mu_e}:=\pi_{\mu_e*}(E_{h})$, where $R_{\mu_e}$ is the Reeb vector field of $(M^\Delta_{\mu_e},\omega_{\mu_e},\eta_{\mu_e})$. Since $(X_{h})_{(t,z_e)}=(\xi(t)_M)_{(t,z_e)}$, for a certain curve $\xi(t)$ in $\mathfrak{g}^\Delta_{\mu_e}$ and $\pi_{\mu_e*}R|_{\mathbf{J}^{\Phi-1}(\mu_e)}=R_{\mu_e}$, then $(Y_{\mu_e})_{\pi_{\mu_e}(m(t))}=(T_{(t,z_e)}\pi_{\mu_e})(\xi(t)_M)_{(t,z_e)}=0$.

As a consequence, the $ \pi_{\mu_e}(m(t))$ are equilibrium points of $Y_{\mu_e}$ and the integral curve of the vector field $Y_{\mu_e}$ passing through $(\pi_{P^\Delta_{\mu_e}}\circ\pi_{\mu_e})(t_0,z_e)$ is just that point. Hence, $(\pi_{P^\Delta_{\mu_e}}\circ\pi_{\mu_e})(m(t))=(\pi_{P^\Delta_{\mu_e}}\circ\pi_{\mu_e})(t_0,z_e)$ for every $t\in T$. Then, the projection of every solution passing through $z_e$ is just the equilibrium point $(\pi_{P^\Delta_{\mu_e}}\circ\pi_{\mu_e})(t_0,z_e)$ of the  Hamilton vector field  $Y_{\mu_e}$ on $P_{\mu_e}^\Delta$. In other words, it is an equilibrium point of the Hamilton equations induced in $M^\Delta_{\mu_e}$ .
\end{proof}

It follows from Proposition \ref{Prop:SingPPCos}  that $z_e\in \pi_P({\bf J}^{\Phi-1}(\mu_e))$ is a  relative equilibrium point of $((T\times P)^\omega_\eta,h,{\bf J}^\Phi)$ if and only if every solution to the Hamilton equations to $h$ passing through $z_e$ projects onto a point in $P^\Delta_{\mu_e}$.

Let us characterise relative equilibrium points of $((T\times P)^\omega_\eta,h,\mathbf{J}^\Phi)$ via Lagrange multipliers \cite[p. 307]{AM78} as critical points of $h$ restricted to ${\bf J}^{\Phi-1}(\mu_e)$.

\begin{theorem}
A point $z_e\in P$ is a  relative equilibrium point of $((T\times P)^\omega_\eta,h,{\bf J}^\Phi)$ if and only if there exists a curve $\xi(t)$ in $\mathfrak{g}$ such that, for every $t\in T$, the point $z_e$ is a critical point of the restriction to $\{t\}\times P$ of the function $h_{\xi(t)}:T\times P\rightarrow  \mathbb{R}$ of the form
\[
h_{\xi(t)}(t',z):=h(t',z)-\langle {\bf J}^\Phi(t',z)-{\bf J}^\Phi(t',z_e),\xi(t)\rangle.
\]
\end{theorem}
\begin{proof}
Let $z_e\in P$ be a  relative equilibrium point. Then, $(\xi(t)_M)_{(t,z_e)}=(X_h)_{(t,z_e)}$ for every $t\in T$ and a curve $\xi(t)$ in $\mathfrak{g}$. Due to the definition of the cosymplectic momentum map ${\bf J}^\Phi$, it turns out that $(\xi(t)_M\,)_{(t,z_e)}=(X_{{ J}_{\xi(t)}})\,_{(t,z_e)}$ and $(X_{h-J_{\xi(t)}})_{(t,z_e)}=0$ for every $t\in T$. Recalling that $J_{\xi(t)}(t',z_e)$ does not really depend on $t'$, one has 
\[
0=[\flat(X_{h-J_{\xi(t)}})]_{(t,z_e)}=(d h_{\xi(t)})_{(t,z_e)}-(Rh_{\xi(t)})_{(t,z_e)}\eta_{(t,z_e)},\qquad\forall t\in T.
\]
Therefore, $(dh_{\xi(t)})_{(t,z_e)}\upharpoonright_{\ker\eta_{(t,z_e)}}=0$ and $(t,z_e)$ is a critical point of $h_{\xi(t)}\upharpoonright_{ \{t\}\times P}$ for every $t\in T$.
 
 Conversely, let $(t,z_e)\in T\times P$ be a critical point of $h_{\xi(t)}\upharpoonright_{\{t\}\times P}$ for every $t\in T$. Then,
\[
(dh_{\xi(t)})_{(t,z_e)}\upharpoonright_{\ker\eta_{(t,z_e)}}=d(h-J_{\xi(t)})_{(t,z_e)}\upharpoonright_{\ker\eta_{(t,z_e)}}=(\iota_{X_{h-J_{\xi(t)}}}\omega)_{(t,z_e)}\upharpoonright_{\ker\eta_{(t,z_e)}}=0,\quad\forall t\in T.
\]
Since $X_{h-J_{\xi(t)}}(t,z_e)$ takes values in $\ker\eta$, one has $(X_{h-J_{\xi(t)}})_{(t,z_e)}=0$ for every $t\in T$. Therefore, $(X_h)_{(t,z_e)}=(X_{J_{\xi(t)}})_{(t,z_e)}=(\xi(t)_M)_{(t,z_e)}$ for every $t\in T$ and hence $z_e$ is a  relative equilibrium point.
\end{proof}

Note that the above theorem can be rewritten as follows.

\begin{corollary}\label{Cor::CosymStrongRelEqTh} A point $z_e\in P$ is a  relative equilibrium point of $((T\times P)^\omega_\eta,h,{\bf J}^\Phi)$ if and only if there exists a curve $\xi(t)$ in $\mathfrak{g}$ such that $(z_e,\xi(t))\in P\times \mathfrak{g}$, for every $t\in T$, are critical points of the functions $\widehat{h}_t:P\times \mathfrak{g}\rightarrow   \mathbb{R}$ of the form
\[
\widehat{h}_t(z,\nu):=h(t,z)-\< \B J^\Phi(t,z)-  \B J^\Phi(t,z_e),\nu\>.
\]
\end{corollary}
Note that $\xi(t)$ plays the role of a $t$-dependent Lagrange multiplier in Corollary \ref{Cor::CosymStrongRelEqTh}.

Let $z_e$ be a  relative equilibrium point of $((T\times P)^\omega_\eta,h,{\bf J}^\Phi)$. Let us define the second variation of $h_{\xi(t_e)}$ at $(t_e,z_e)$, for any $t_e\in T$, as the mapping $(\delta^2h_{\xi(t_e)})_{(t_e,z_e)}:\ker\eta_{(t_e,z_e)}\times \ker\eta_{(t_e,z_e)}\rightarrow  \mathbb{R}$ of the form
\begin{equation}
\label{Eq::SecVarh}
(\delta^2h_{\xi(t)})_{(t_e,z_e)}(v_1,v_2):=\iota_Y(d(\iota_Xdh_{\xi(t_e)}))_{(t_e,z_e)},
\end{equation}
for some vector fields $X,Y$ on $M$ defined on a neighbourhood of $(t_e,z_e)$ taking values in $\ker\eta$ and such that $v_1=X_{(t_e,z_e)}$, $v_2=Y_{(t_e,z_e)}$. Note that, for each pair $v_1,v_2$, it is always possible to find some $X, Y$ satisfying given conditions. In Darboux coordinates $\{t,x_1,\ldots,x_{2n}\}$  on an open neighbourhood $U$ of $(t_e,z_e)$, one has that $X=\sum^{2n}_{i=1}f_i\frac{\partial}{\partial x_i}$ and $Y=\sum^{2n}_{i=1}g_i\frac{\partial}{\partial x_i}$, where $\iota_{\frac{\partial}{\partial x_i}}\eta=0$ for $i=1,\ldots, 2n$.  It is worth noting that the functions $f_1,\ldots,f_{2n}, g_1,\ldots,g_{2n}\in C^\infty(U)$ may depend on $t$.

\begin{proposition}
Let $z_e\in P$ be a  relative equilibrium point of $((T\times P)^\omega_\eta,h,{\bf J}^\Phi)$. If $\{t,x_1,\ldots,x_{2n}\}$ are Darboux coordinates on a neighbourhood of $(t_e,z_e)\in T\times P$, for a $t_e\in T$, then
\begin{equation}
\label{Eq::SecVarFormh}
(\delta^2h_{\xi(t_e)})_{(t_e,z_e)}(w,v)=\sum^{2n}_{i,j=1}\frac{\partial^2h_{\xi(t_e)}}{\partial x_i\partial x_j}(t_e,z_e)w_iv_j,\quad \forall v,w\in\ker\eta_{(t_e,z_e)},\quad
\end{equation}
where $w=\sum_{i=1}^{2n}w_i\partial/\partial x_i$ and $v=\sum_{i=1}^{2n}v_i\partial/\partial x_i$.
\end{proposition}
\begin{proof}
From \eqref{Eq::SecVarh} and the fact that the vector fields $X,Y$ associated with the tangent vectors $w,v$ take values in $\ker \eta$, we have
\begin{multline*}
(\delta^2h_{\xi(t_e)})_{(t_e,z_e)}(w,v)=\iota_Y(d\iota_Xdh_{\xi(t_e)})_{(t_e,z_e)}=\sum^{2n}_{i,j=1}\frac{\partial^2h_{\xi(t_e)}}{\partial x_i\partial x_j}(t_e,z_e)w_iv_j\\+\sum^{2n}_{i,j=1}\frac{\partial h_{\xi(t_e)}}{\partial x_i}(t_e,z_e)\frac{\partial X_i}{\partial x_j}(t_e,z_e)v_j=\sum^{2n}_{i,j=1}\frac{\partial^2h_{\xi(t_e)}}{\partial x_i\partial x_j}(t_e,z_e)w_iv_j,
\end{multline*}
where $X=\sum_{i=1}^{2n}X^i\partial/\partial x^i$, $X(t_e,z_e)=w$, and we have used that $z_e$ is a  relative equilibrium point.
\end{proof}
It follows from \eqref{Eq::SecVarFormh} that the mappings $(\delta^2h_{\xi(t_e)})_{(t,z_e)}$, for each $t\in T$, are symmetric. Let us study \eqref{Eq::SecVarh} in more detail.

\begin{proposition}\label{Prop::CosGaugeDir}
Let $z_{e}\in P$ be a  relative equilibrium point for $((T\times P)^\om_\eta,h,{\bf J}^\Phi)$. Then, for every $t\in T$, one has
\begin{equation}
\label{Eq:Cos2var}
(\delta^{2}h_{\xi(t)})_{(t,z_e)}((\zeta_{M})_{(t,z_e)}, v_{(t,z_e)})=0, \quad\forall\zeta\in \mathfrak{g},\quad\forall v_{(t,z_e)}\in T_{(t,z_e)}({\bf J}^{\Phi-1}(\mu_e))\cap \ker\eta_{(t,z_e)}.
\end{equation}
\end{proposition}

\begin{proof}The $G$-invariance of $h:T\times P\rightarrow  \mathbb{R}$ and the equivariance condition for ${\bf J}^\Phi$ relative to the affine Lie group action, $\Delta$, yields, for every $g\in G$ and all $(t',z)\in T\times P$, that, for $\mu_e={\bf J}^{\Phi}(t',z_e)$, one has
\begin{multline*}
h_{\xi(t)}(\Phi_{g}(t',z))=h(\Phi_{g}(t',z))-\left\langle\Delta_g{\bf J}^\Phi(t',z),\xi(t)\right\rangle+\left\langle \mu_{e}, \xi(t)\right\rangle\\=h(t',z)-\left\langle {\bf J}^\Phi(t',z),\Delta^T_g\xi(t)\right\rangle+\left\langle \mu_{e}, \xi(t)\right\rangle,
\end{multline*}
where $\Delta^T_g:\mathfrak{g}\rightarrow \mathfrak{g}$ is the transpose of $\Delta_g$ for every $g\in G$. Since, for any fixed $t\in T$, we can substitute $g=\exp(s\zeta)$, with $\zeta\in \mathfrak{g}$, and differentiating with respect to $s$, one gets
\[
(\iota_{\zeta_M}d h_{\xi(t)})(t',z)=-\left\langle {\bf J}^\Phi(t',z),\frac{d}{ds}\bigg|_{s=0}\Delta^T_{\exp(s\zeta)}\xi(t)\right\rangle=\left\langle {\bf J}^\Phi(t',z),(\zeta_{\mathfrak{g}}^\Delta)_{\xi(t)}\right\rangle,
\]
where $(\zeta_{\mathfrak{g}}^\Delta)_{\xi(t)}$ is the fundamental vector field of $\Delta^T:G\times\mathfrak{g}\rightarrow \mathfrak{g}$ associated with $\zeta\in \mathfrak{g}$ at $\xi(t)\in \mathfrak{g}$, for a fixed $t\in T$. Note that the induced action on $\mathfrak{g}$ has fundamental vector fields
\[
\zeta^\Delta_\mathfrak{g}(v):=\frac{d}{ds}\bigg|_{s=0}\Delta^T_{\exp(-s\zeta )}v,\qquad \forall v\in \mathfrak{g}.
\]
Recall, that $(\zeta_M)_{(t,z_e)}$ and $v_{(t,z_e)}$ take values in $\ker\eta_{(t,z_e)}$. Taking variations relative to $z\in P$ above and evaluating at $(t,z_{e})$, one gets 
\[
(\delta^{2}h_{\xi(t)})_{(t,z_e)}((\zeta_{M})_{(t,z_e)},v_{(t,z_e)})=\left\langle T_{(t,z_e)}\B J^\Phi(v_{(t,z_e)}), (\zeta_{\mathfrak{g}}^\Delta)_{\xi(t)}\right\rangle,
\]
which vanishes if $T_{(t,z_e)}\B J^\Phi(v_{(t,z_e)})=0$, i.e. if $v_{(t,z_e)}\in \ker T_{(t,z_e)}\B J^\Phi=T_{(t,z_e)}({\bf J}^{\Phi-1}(\mu_{e}))$.
\end{proof}

The following corollary is a natural consequence of Proposition \ref{Prop::CosGaugeDir}.

\begin{corollary}
If $z_e$ is a  relative equilibrium point of $((T\times P)^\om_\eta,h,{\bf J}^\Phi)$, then the subspace $T_{(t,z_e)}(G^\Delta_{\mu_e} (t,z_e))$ belongs to the kernel of the restriction of $(\delta^2 h_{\xi(t)} )_{(t,z_e)}$ to $T_{(t,z_e)}({\bf J}^{\Phi-1}(\mu_e))\cap \ker \eta_{(t,z_e)}$.
\end{corollary}

\section{Stability on the reduced manifold}
\label{Sec::ReducedStability}

Section \ref{Sec::CharStrongEqP} introduced the basic results of a cosymplectic energy-momentum method, which allows for finding  relative equilibrium points of $((T\times P)^\omega_\eta,h,{\bf J}^\Phi)$.  This section analyses the stability on the reduced space by applying and interpreting the results of \cite{LZ21} in our cosymplectic framework without some unnecessary technical conditions on the momentum map assumed there. We hereafter assume $T=\mathbb{R}$ so as to use Definition \ref{lpdf}, which is the basis to establish conditions ensuring different types of stability on manifolds.

Recall that a Marsden--Weinstein reduction for $((\mathbb{R}\times P)^\omega_\eta,h,{\bf J}^\Phi)$ is a reduction from $(\mathbb{R}\times P,\omega,\eta)$ to a cosymplectic manifold $(\mathbb{R}\times P^\Delta_\mu,\omega_\mu,\eta_\mu)$, where $\omega_\mu$ and $\eta_\mu$ are given by $\iota_\mu^*\omega=\pi_\mu^*\omega_\mu$ and $\iota_\mu^*\eta=\pi_\mu^*\eta_\mu$, for the immersion $\iota_\mu:{\bf J}^{\Phi-1}(\mu)\hookrightarrow \mathbb{R}\times P$ and the projection $\pi_\mu:{\bf J}^{\Phi-1}(\mu)\rightarrow  M^\Delta_\mu={\bf J}^{\Phi-1}(\mu)/G_\mu^\Delta$. Recall that $M_\mu^\Delta\simeq \mathbb{R}\times P_\mu^\Delta$ for a certain manifold $P_\mu^\Delta$ introduced in Corollary \ref{Cor::De}.

Let us analyse the function $h_{z_e}:\mathbb{R}\times P\rightarrow  \mathbb{R}$ given by
\[
h_{z_e}(t,z):=h(t,z)-h(t,z_e).
\]
Then, $h_{z_e}(t,z_e)=0$ for every $t\in \mathbb{R}$. This is done to study $h_{z_e}(t,z)$ with lpdf functions and other functions of the sort. If $(t,z(t))$ is the particular solution to our $G$-invariant cosymplectic Hamiltonian system $((\mathbb{R}\times P)^\omega_\eta,h,{\bf J}^\Phi)$ with the initial condition $(0,z_e)$, then
\[
\frac{d}{dt}\bigg|_{t=0}h_{z_e}(t,z(t))=\frac{d}{dt}\bigg|_{t=0}h(t,z(t))-\frac{d}{dt}\bigg|_{t=0}h(t,z_e).
\]
For simplicity, we have chosen $t=0$, but we could have chosen any other initial time. Since the integral curves of $E_f$ for  $(\mathbb{R}\times P,\omega,\eta)$ are given by \eqref{Eq::CosymHamEq}, the derivative with respect to $t$ of a function $h_{z_e}$ along the solutions of the Hamilton equations for $h$ reads
\[
\frac{dh_{z_e}}{dt}=E_{h}h_{z_e}=Rh_{z_e}+ \{h_{z_e},h\}_{\omega,\eta}=Rh_{z_e}=\frac{\partial h_{z_e}}{\partial t}.
\]
Note that choosing another variable $t$ in a Darboux coordinate system on $\mathbb{R}\times P$ does not change the above relations. 
Note that $h_{z_e}\circ \Phi_g=h_{z_e}$ for every $g\in G$. Since $\pi_T\circ\Phi_g=\pi_T$, there exists a function $H_{z_e}:\mathbb{R}\times P_{\mu_e}^\Delta\rightarrow  \mathbb{R}$ of the form
\[
H_{z_e}(t,[z]):=h_{z_e}(t,z),\qquad \forall (t,z)\in {\bf J}^{\Phi-1}(\mu_e),
\]
where $(t,[z])$ stands for the equivalence class of $(t,z)\in {\bf J}^{\Phi-1}(\mu_e)$ in ${\bf J}^{\Phi-1}(\mu_e)/G^\Delta_{\mu_e}$. The function $H_{z_e}(t,[z])-k_{\mu_e}(t,[z])$ depends only on $t$, because $k_{\mu_e}(t,[z])$ is given by $\pi_{\mu_e}^*k_{\mu_e}=\iota_{\mu_e}^*h$. Recall that $\mathbb{R}\times P_{\mu_e}^\Delta$ is also a cosymplectic manifold. Therefore, similarly $\pi_{\mu_e*}(R+X_h)=R_{\mu_e}+X_{k_{\mu_e}}$, and $[z_e]$ is an equilibrium point of $X_{k_{\mu_e}}$. Hence, $H_{z_e}|_{\{t\}\times P_{\mu_e}^\Delta}$ has an equilibrium point in $[z_e]$.
Moreover,
\[
\frac{d}{dt}\bigg|_{t=0}H_{z_e}(t,[z(t)])=(Rh_{z_e})(0,z(0)),\qquad \forall (0,[z(0)])\in {\bf J}^{\Phi-1}(\mu_e)/ G^\Delta_{\mu_e}\simeq \mathbb{R}\times P_{\mu_e}^\Delta,
\]
where $z(t)$ is any solution to the initial Hamiltonian equations of $h$ within ${\bf J}^{\Phi-1}(\mu_e)$ with initial condition $z(0)$.

Let us use $H_{z_e}$ to study the stability of $[z_e]$ in $P_{\mu_e}^\Delta$. In particular, we will study the conditions on $h$ to ensure that $H_{z_e}$ gives rise to different types of stable equilibrium points at $[z_e]$. With this aim, consider a coordinate system $\{x_1,\ldots,x_n\}$ on an open neighbourhood $\mathcal{U}$ of $[z_e]\in P_{\mu_e}^\Delta$ such that $x_i([z_e])=0$ for $i=1,\ldots, n$. Let $\alpha=(\alpha_1,\ldots,\alpha_n)$, with $\alpha_1,\dots,\alpha_n\in \mathbb{N}\cup \{0\}$, be a multi-index with $n=\dim {\bf J}^{\Phi-1}(\mu_e)/G_{\mu_e}^\Delta-1$. To understand this, recall that $M^\Delta_{\mu_e}$ contains the $t$-dependence that will be a parameter, not a variable, in the stability analysis. Let $|\alpha|=\sum_{i=1}^n\alpha_i$ and $D^\alpha=\partial^{\alpha_1}_{x_1}\cdots \partial^{\alpha_n}_{x_n}$ for every $\alpha$. The proof of the following lemma is very technical and will be omitted (see \cite{LZ21, MS88} for detailed proof).

\begin{lemma}\label{Lemm:MegaLemma} Let us define the $t$-dependent parametric family of $n\times n$  matrices $M(t)$ with entries
\[
[M(t)]_i^j=\frac12\frac{\partial^2H_{z_e}}{\partial x_i\partial x_j}(t,[z_e]),\qquad \forall t\in \mathbb{R},\qquad i,j=1,\ldots,n,
\]
and let ${\rm spec}(M(t))$ stands for the spectrum of the matrix $M(t)$ at $t\in \mathbb{R}$. Assume that there exists a $\lambda\in \mathbb{R}$ such that $0<\lambda<\inf_{t\in I_{t^0}}\min {\rm spec}(M(t))$ for some $t^0\in \mathbb{R}$. Suppose also that there exists a real constant $c$ such that
\[
c\geq \frac 16\sup_{t\in I_{t^0}}\max_{|\alpha|=3}\max_{[y]\in \mathcal{B}}|D^\alpha H_{z_e}(t,[y])|
\]
for a certain compact neighbourhood $\mathcal{B}$ of $[z_e]$. Then, there exists an open neighbourhood $\mathcal{U} $  of $[z_e]$, where the function $H_{z_e}:\mathbb{R}\times \mathcal{U}\rightarrow  \mathbb{R}$ is lpdf from $t^0$. If there exists additionally a constant $\Lambda$ such that 
\[
\sup_{t\in I_{t^0}}\max {\rm spec}(M(t))< \Lambda,
\]
then, $H_{z_e}:\mathbb{R}\times \mathcal{U}\rightarrow  \mathbb{R}$ is a decrescent function from $t^0$. 
\end{lemma}

The eigenvalues of $M(t)$ depend on the chosen coordinate system around $[z_e]$. 

\begin{lemma}\label{Lem:Coorl} If the $t$-dependent matrix $M(t)$, which is defined in a local coordinate system $\{x_1,\ldots,x_n\}$ on an open neighbourhood of an equilibrium point $[z_e]\in P_{\mu_e}$, satisfies that $0< \lambda<{\rm inf}_{t\in I_{t^0}}\min{\rm spec}\,M(t) $  for some $\lambda$ (resp. ${\rm sup}_{t\in I_{t^0}}\max {\rm spec}\,M(t)< \Lambda$ for some $\Lambda$), then $M_{\mathcal{B}'}(t)$, defined as $M(t)$ but in  another coordinate system $\mathcal{B}':=\{\tilde{x}_1,\ldots,\tilde{x}_n\}$ on another neighbourhood in $P_{\mu_e}$ of $[z_e]$, satisfies that $0< \lambda'<{\rm inf}_{t\in I_{t^0}}\min {\rm spec}\,M_{\mathcal{B}'}(t)$  for some $\lambda'$  (resp. ${\rm sup}_{t\in I_{t^0}}\max{\rm spec}\,M_{\mathcal{B}'}(t)<\Lambda'$ for some $\Lambda'$).
\end{lemma}

An appropriate coordinate system may simplify $M(t)$ at certain values of $t$, e.g. by writing $M(t)$ in a canonical form. Nevertheless, the simplification of $M(t)$ at every time $t\in I_{t^0}$ for a certain coordinate system in $P_{\mu_e}^\Delta$ around $[z_e]$ will be, in general, impossible. We therefore restrict ourselves to determining a condition on a particular coordinate system. 

Lemma \ref{Lemm:MegaLemma} leads, immediately, to  the following theorem. 

\begin{theorem}\label{Th:StabilityCon} If there exist $\lambda,c>0$ and an open neighbourhood $U$ of $[z_e]$ so that
\[
\lambda< {\rm min}({\rm spec}(M(t))),\qquad c\geq \frac{1}{3!}\max_{ |\alpha|=3}\sup_{[x]\in U}|D^\alpha H_{z_e}(t,[x])|,\qquad
\frac{\partial H_{z_e}}{\partial t}\bigg|_{U}\leq 0,
\]
for every $t\in I_{t^0}$, then $[z_{\!e}]$ is a stable point of the Hamiltonian vector field related to $k_{\mu_e}$ on ${\bf J}^{\Phi\!-\!1\!}(\mu_e)\!/G_{\!\mu_e}^{\!\Delta}$ from $t^0$. If there exists $\Lambda$ such that $\max ({\rm spec}(M(t)))<\Lambda$ for every $t\in I_{t^0}$, then $[z_e]$ is uniformly  stable from $t^0$.
\end{theorem}

Assuming stronger conditions on the derivatives of $H_{z_e}$ than in Theorem \ref{Th:StabilityCon}, one gets Corollary \ref{Cor::StablePoint}, whose  conditions that can indeed be proved to hold independently of the chosen coordinate system (cf. \cite{LZ21}), which makes them geometrical as proved next.

\begin{corollary}\label{Cor::StablePoint}
If there exist $\lambda,c>0$ and an open neighbourhood $U$ of $[z_e]$ such that
\begin{equation}
\label{Eq::StabConditions}
\lambda <\min\left({\rm spec}\left(M(t)\right)\right),\qquad c\geq \frac{1}{3!}\max_{1\leq |\alpha|\leq 3}\sup_{[x]\in U}\left|D^\alpha H_{z_e}(t,[x])\right|,\qquad \frac{\partial H_{z_e}}{\partial t}\bigg|_{U}\leq 0,
\end{equation}
for every $t\in I_{t^0}$, then $[z_e]$ is a uniformly  stable point of the Hamiltonian system $k_{\mu_e}$ on ${\bf J}^{\Phi-1}(\mu_e)/G_{\mu_e}^\Delta$ from $t^0$.
\end{corollary}

The existence of $c$ in Corollary \ref{Cor::StablePoint} yields, from the first and second expression in \eqref{Eq::StabConditions}, that $\max({\rm spec}(M(t)))\leq 6cn^2$ for every $t\in I_{t^0}$.  Indeed,
\[
v^TM(t)v\leq \!\sum_{i,j=1}^n|v_i||v_j||M^i_j(t)|\leq 6c \!\sum_{i,j=1}^n\|v\|^2=6cn^2 \|v\|^2,\;\;\; \forall v\in \mathbb{R}^n.
\]
Consequently, $v^TM(t)v<\Lambda v^Tv$ for every non-zero $v\in \mathbb{R}^n$ and $\Lambda> 6cn^2$.

The results obtained above employ a distance on an open coordinate neighbourhood of $[z_e]$ that was induced by a standard norm in $\mathbb{R}^n$. The topology induced by this norm is the same as the one induced by any other Riemannian metric on the open neighbourhood of $[z_e]$. Hence,  our results concerning the stability of $[z_e]$ are independent of the used Riemannian metric.

The lemma below shows that Corollary \ref{Cor::StablePoint} has a geometric meaning: the conditions provided are satisfied independently of the chosen coordinates, which may change the values of the constants $\lambda,c$, but not its existence, which is the relevant fact for our main aims. Its proof can be found in \cite{LZ21}.

\begin{lemma}\label{Lem:Coorl2} If $M(t)$, which is defined in a local coordinate system $\{x_1,\ldots,x_n\}$ on an open neighbourhood of an equilibrium point $[z_e]\in P^\Delta_{\mu_e}$, is such  that $0< \lambda<{\rm inf}_{t^0\leq t}\min{\rm spec}\,M(t) $  for some $\lambda$ (resp. ${\rm sup}_{ t^0\leq t}\max {\rm spec}\,M(t)< \Lambda$ for some $\Lambda$), then $M_{\mathcal{B}'}(t)$,which is determined like $M(t)$ but in  another coordinate system $\mathcal{B}'=\{\tilde{x}_1,\!\ldots\!,\tilde{x}_n\}$ around  $[z_e]\in P_{\mu_e}^\Delta$, holds that $0< \lambda'<{\rm inf}_{t^0\leq t}\min {\rm spec}\,M_{\mathcal{B}'}(t)$  for some $\lambda'$  (resp. ${\rm sup}_{t^0\leq t}\max{\rm spec}\,M_{\mathcal{B}'}(t)<\Lambda'$ for some $\Lambda'$).  
\end{lemma}

The condition for $c$ in Corollary \ref{Cor::StablePoint} is independent of the chosen coordinate system. More specifically, the condition on a new coordinate system also holds by choosing a new $c'$ and restricting to a new open subset of $(t,[z_e])$, where the previous condition and the new coordinate system are defined.

In this section, we assume $\mathbf{J}^\Phi$ to be regular at $\mathbf{J}^\Phi(t,z_e)$ for a  relative equilibrium point $z_e\in P$.
The energy-momentum method determines properties of $h$ on a neighbourhood of a relative equilibrium point $m_e=(t,z_e)\in \mathbb{R}\times P$ that ensure a certain type of stability around an associated equilibrium point of the Hamilton equations of $k_{\mu_e}$ in $\mathbb{R}\times P_{\mu_e}$. In particular, we will give conditions on $h_{\mu_e}:(t,x)\in {\bf J}^{\Phi-1}(\mu_e)\mapsto h(t,x)\in \mathbb{R}$, and $\partial h_{\mu_e}/\partial t$ with $t\in \mathbb{R}$, to ensure that the conditions in Theorem \ref{Th:StabilityCon} and/or Corollary \ref{Cor::StablePoint} are satisfied. Instead of investigating  $M(t)$,  we will set conditions on the functions $h_{\xi(t)}\upharpoonright_{\{t\}\times P}$ for $t\in \mathbb{R}$, which is more practical as they are not defined on the quotient of a submanifold of $P$ and they are straightforwardly known without making additional computations. As in the case of Section \ref{Sec::CharStrongEqP}, the present section heavily relies on  the results of \cite{LZ21}, which are to be very slightly modified to be adapted to a cosymplectic realm. Moreover, some technical conditions concerning the ${\rm Ad}^*$-equivariance assumed in \cite{LZ21} are again removed.

\section{Stability, reduced manifold, and  relative equilibrium points}
\label{Sec::RelationofHessians}

In this section, we assume $\mathbf{J}^\Phi$ to be regular at $\mathbf{J}^\Phi(t,z_e)$ for a  relative equilibrium point $z_e\in P$ and every $t\in \mathbb{R}$. The cosymplectic energy-momentum method aims to determine properties of $h$ on a neighbourhood of a  relative equilibrium point $z_e\in P$ that ensure a certain type of stability around an associated equilibrium point of the Hamilton equations of $k_{\mu_e}\in C^\infty(\mathbb{R}\times P^\Delta_{\mu_e})$. In particular, we will give conditions on $h_{\mu_e}:(t,x)\in {\bf J}^{\Phi-1}(\mu_e)\mapsto h(t,x)\in \mathbb{R}$, and $\partial h_{\mu_e}/\partial t$, to ensure that the conditions in Theorem \ref{Th:StabilityCon} and/or Corollary \ref{Cor::StablePoint} are satisfied. Instead of investigating  $M(t)$,  we will set conditions on the functions $h_{\xi(t)}|_{\{t\}\times P}$ for $t\in \mathbb{R}$, which is more practical as they are not defined on the quotient of a submanifold of $P$ and they are straightforwardly known without making additional computations. As in Section 8, this section also heavily relies on  the results of \cite{LZ21}, which are to be slightly modified to be adapted to a cosymplectic realm. Moreover, some technical conditions concerning the ${\rm Ad}^*$-equivariance assumed in \cite{LZ21} are here  removed.

Let us assume that $G^\Delta_{\mu_e}$ acts in a quotientable manner on ${\bf J}^{\Phi-1}(\mu_e)$. Let us define a coordinate system $\{t,z_1,\ldots,z_q\}$ on an open subset $\mathbb{R}\x\mathcal{A}_{\mu_{e}}\subset {\bf J}^{\Phi-1}(\mu_e)$ containing $m_e=(t_e,z_e)$ for some $t_e\in \mathbb{R}$.  Let $\{t,\pi^*_{\mu_e}x_1,\ldots,\pi^*_{\mu_e}x_n\}$ be the coordinates on $\mathbb{R}\x\mathcal{A}_{\mu_{e}}$ obtained by pull-back to $\mathbb{R}\x\mathcal{A}_{\mu_{e}}$ some coordinates $\{t,x_1\ldots,x_n\}$ on $\mathbb{R}\x\mathcal{O}=\pi_{\mu_e} (\mathbb{R}\x\mathcal{A}_{\mu_{e}})$ \footnote{To simplify the notation, $\{x_1,\ldots,x_n\}$ will stand for a set of coordinates on a neighbourhood of $[z_e]$ and their pull-backs to ${\bf J}^{\Phi-1}(\mu_e)$ via $\pi_{\mu_e}$ simultaneously.}, since the cosymplectic Marsden--Weinstein reduction does not `reduce' the space $\mathbb{R}$ (see Corollary \eqref{Cor::De}), and let $\{y_{1},\ldots,y_{s}\}$ be additional coordinates giving rise to a coordinate system $\{t,z_1,\ldots,z_q\}$ on $\mathbb{R}\x\mathcal{A}_{\mu_{e}}$. Due to the $G^\Delta_{\mu_e}$-invariance of $h_{\mu_e}=h\circ \iota_{\mu_e}:{\bf J}^{\Phi-1}(\mu_e)\rightarrow  \mathbb{R}$, one has that there exists $c$ such that 
\[
c\geq \frac{1}{3!}\max_{3\geq |\vartheta|\geq 1}\sup_{z\in \mathcal{A}_{\mu_{e}}}|D^\vartheta h_{\mu_e}(t,y)|,\qquad \forall t\in I_{t^0},
\]
where $\vartheta$ is a multi-index $\vartheta=(\vartheta_1,\ldots,\vartheta_q)$ if and only if 
\begin{equation}\label{Rel2}
c\geq \frac{1}{3!}\max_{3\geq  |\alpha|\geq 1}\sup_{x\in \mathcal{O}}|D^\alpha H_{z_e}(t,x)|,\qquad \forall t\in I_{t^0},
\end{equation}
where $\mathbb{R}\x\mathcal{O}$ is an open neighbourhood of $[m_e]=(t_e,[z_e])$ because $\pi_{\mu_e}$ is an open mapping. Indeed, since $h_{\mu_e}$ is constant on the submanifolds where $t,x_1,\ldots,x_n$ take constant values,  $h_{\mu_e}(t,x_1,\ldots,x_n,y_1,\ldots,y_s)-h(t,z_e)=H_{z_e}(t,x_1,\ldots,x_n)$ and (\ref{Rel2}) follows. 

Consider again the local coordinate system $\{t,z_1,\ldots,z_q\}$ on ${\bf J}^{\Phi-1}(\mu_e)$. We write $[\widehat{M}(t)]$ for the  $t$-dependent $q\times q$ matrix given by the $t$-dependent coefficients of the form 
\[
[\widehat{M}(t)]_i^j:=\frac{\partial^2h_{\mu_e}}{\partial z_i\partial z_j}(t,z_e),\qquad i,j=1,\ldots,q.
\]
It is worth noting that coordinates are constructed respecting the local natural decomposition ${\bf J}^{\Phi-1}(\mu_e)$ of the form $\mathbb{R}\x\mathcal{A}_{\mu_{e}}$. Note also that $h_{\mu_e}(t,z)=h_{\xi(t)}(t,z)$ for $(t,z)\in {\bf J}^{\Phi-1}(\mu)$. 

Lemma \ref{Lem:Coorl2} tells us that,  geometrically, the existence of $\lambda$ and $\Lambda$ amounts to the fact that the $t$-dependent bilinear symmetric form $K(t):T_{[z_e]}P^\Delta_{\mu_e}\times T_{[z_e]}P^\Delta_{\mu_e}\rightarrow  \mathbb{R}$ given by
\[
K(t):=\frac 12\sum_{i,j=1}^n\frac{\partial^2H_{z_e}}{\partial x_i\partial x_j}(t,[z_e])dx_i|_{[z_e]}\otimes dx_j |_{[z_e]}
\]
satisfies that
\begin{equation}\label{Eq:Comen}
K(t)(w,w)> \lambda (w|w)_{\mathcal{B}},\qquad \forall w\in T_{[z_e]}P^\Delta_{\mu_e}\backslash\{0\},\qquad\forall t\in I_{t^0},
\end{equation}
where $(\cdot |\cdot )_{\mathcal{B}}$ is the Euclidean product in $T_{[z_e]}P^\Delta_{\mu_e}$ satisfying that $\{\partial_{x_1},\ldots,\partial _{x_n}\}$ is an orthonormal basis. Indeed,  if $v$ stands for the column vector of  the coordinates of $w\in T_{[z_e]}P^\Delta_{\mu_e}$ in $\{\partial_{x_1},\ldots,\partial _{x_n}\}$, then 
\[
K(t)(w,w)=v^TM(t)v> \lambda v^Tv=\lambda (w|w)_{\mathcal{B}}, \qquad \forall w\in T_{[z_e]}P^\Delta_{\mu_e}\backslash\{0\},\quad \forall t\in I_{t^0}.
\]
Moreover, given another inner product $(\cdot| \cdot)_{\mathcal{B}'}$ on $T_{[z_e]}P^\Delta_{\mu_e}$, there  exist $m_l,m_s>0$ such that $m_s(w|w)_{\mathcal{B}'}> (w|w)_{\mathcal{B}}> m_l(w|w)_{\mathcal{B}'}$ for all $w\in T_{[z_e]}P^\Delta_{\mu_e}\setminus \{0\}$ (recall that in finite-dimensional spaces all metrics induced by norms are strong equivalent \cite{AM78}). Consequently,  if   (\ref{Eq:Comen}) is satisfied for an inner product in $T_{[z_e]}P^\Delta_{\mu_e}$, then it also holds  for any other one, after an eventual change of the value of $\lambda$. The same applies, mutatis mutandis, to the relation $\Lambda  (w|w)_\mathcal{B}> K(t)(w,w)$ for a $\Lambda>0$, for all $t\in I_{t^0}$, and every $w\in T_{[z_e]}P^\Delta_{\mu_e}\backslash\{0\}$.

It is worth noting that the inner product $(\cdot|\cdot)_\mathcal{B}$ is introduced to effectively determine whether the $t$-dependent matrix $M(t)$ has eigenvalues that can be bounded from below simultaneously for every time $t\in I_{t^0}$.

Let us provide a geometric method to verify (\ref{Eq:Comen})  via the space ${\bf J}^{\Phi-1}(\mu_e)$. Since every $h_{\mu_e}|_{\{t\}\times \mathcal{A}_{\mu_e}}$, with $t\in\mathbb{R}$, has a critical point at each  relative equilibrium point $z_e\in \mathcal{A}_{\mu_e}$ there exists a $t$-dependent bilinear symmetric function  $\widehat{M}(t):T_{z_e}\mathcal{A}_{\mu_e}\times T_{z_e}\mathcal{A}_{\mu_e}\rightarrow  \mathbb{R}$, of the form
\[
\widehat{M}(t):=\frac 12 \sum_{i,j=1}^q\frac{\partial^2h_{\mu_e}}{\partial z_i\partial z_j}(t,z_e)dz_i|_{z_e}\otimes dz_j|_{z_e},\qquad \forall t\in I_{t^0},
\]
where $\mathcal{B}=\{t,z_1,\ldots,z_q\}$ is any coordinate system in an open neighbourhood of $(t_e,z_e)\in {\bf J}^{\Phi-1}(\mu_e)$ adapted to $\mathbb{R}\times \mathcal{A}_{\mu_e}$.

Let $\{t,x_1,\ldots,x_n,y_1,\ldots,y_s\}$ be the coordinate system on the open neighbourhood of $(t_e,z_e)$ in ${\bf J}^{\Phi-1}(\mu_e)$ defined above. Then,
\[
\frac{\partial^2h_{\mu_e}}{\partial x_k\partial y_j}(t,z_e)=\frac{\partial^2h_{\mu_e}}{\partial y_i\partial y_j}(t,z_e)=0, \qquad i,j=1,\ldots,s,\qquad k=1,\ldots,n,\qquad \forall t\in \mathbb{R},
\]
on the open neighbourhood.
Note that $\pi_{\mu_e}^*K(t)=\widehat{M}(t)$ and $T_{z_e}(G_{\mu_e}^\Delta z_e)\subset\ker \widehat{M}(t)$ for every $t\in \mathbb{R}$. Since objects are geometrical,   $K(t)$ can be considered as the induced bilinear form by $\widehat{M}(t)$ on $S_{z_e}\simeq T_{z_e}\mathcal{A}_{\mu_e}/T_{z_e}(G^\Delta_{\mu_e}z_e)\simeq T_{[z_e]}P^\Delta_{\mu_e}$. Thus,  the conditions for $M(t)$  can be verified via $\widehat{M}(t)$.
Corollary \ref{Cor::StablePoint} and the previous remarks allow us to ennunciate the following theorem.

\begin{theorem}\label{Th:StabilityCon2} Let us assume that there exist $\lambda,c>0$ and an open coordinate neighbourhood $\mathcal{A}_{\mu_e}$ of $z_e$ so that $\mathbb{R}\times \mathcal{A}_{\mu_e}\subset {\bf J}^{\Phi-1}(\mu_e)$ and
\begin{equation}\label{eq:Cond}
\lambda< {\rm min}({\rm spec}([\widehat{M}(t)]|_{S_{z_e}}),\quad c\geq \frac{1}{3!}\max_{1\leq |\vartheta|\leq 3}\sup_{y\in \mathcal{A}_{\mu_e}}|D^\vartheta h_{\mu_e}(t,y)|,\quad
\frac{\partial h_{\mu_e}}{\partial t}\bigg|_{\mathcal{A}_{\mu_e}}\leq 0,
\end{equation}
for every $t\in I_{t^0}$ and a subspace $S_{z_e}\subset T_{z_e}\mathcal{A}_{\mu_e}$ suplementary to $T_{z_e}(G^\Delta_{\mu_e}z_e)$, then $[z_e]$ is a uniformly stable point of the Hamiltonian system $k_{\mu_e}$ on ${\bf J}^{\Phi-1}(\mu_e)/G^\Delta_{\mu_e}$ from $t^0$. 
\end{theorem}

Finally, let us relate the properties of $h_{\xi(t)}|_{\{t\}\times P}$ with $H_{\mu_e}$ so as to investigate  relative equilibrium points in $P$ and their associated equilibrium points in $P^\Delta_{\mu_e}$. Since $h_{\xi(t)}\upharpoonright_{\{t\}\times P}$ has a critical point at a  relative equilibrium point $z_e\in P$ for every $t\in \mathbb{R}$, it is possible to define the $t$-dependent bilinear symmetric form on $T_{z_e}P$ given by
\[
T_{z_e}(t):=\frac 12\sum_{i,j=1}^{\chi}\frac{\partial^2 h_{\xi(t)}}{\partial u_i\partial u_j}(t,z_e)du_i|_{z_e}\otimes du_j|_{z_e},\qquad \forall t\in \mathbb{R},
\]
where $\{t,u_1,\ldots,u_\chi\}$, with $\chi=\dim P$, is a coordinate system on an open neighbourhood of $m_e=(t,z_e)$ in $\mathbb{R}\times P$. Let us study the relation of $T_{z_e}(t)$ with $\widehat{M}(t)$ to study the latter via the former. Note that $T_{z_e}(t)$ is a geometric object easy to be constructed as it is defined on $T_{z_e}P$ and it depends, essentially, only on $h$ and ${\bf J}^\Phi$.

Since ${\bf J}^\Phi$ is regular, the coordinates of ${\bf J}^{\Phi}$ around ${\bf J}^{\Phi-1}(\mu_e)$, e.g. $\mu_1,\ldots,\mu_r$, give rise to $\dim \mathfrak{g}$ functionally independent functions on $P$. Consider now the coordinate system on a neighbourhood $\mathcal{A}_{\mu_e}$ of $z_e$ so that $\mathbb{R}\x\mathcal{A}_{\mu_e}\subset{\bf J}^{\Phi-1}(\mu_e)$  given by $\{t,x_1,\ldots,x_n,y_1,\ldots,y_s\}$. Let us  extended smoothly such coordinates to an open neighbourhood in $M$ containing  $\mathbb{R}\times \{z_e\}$. As ${\bf J}^\Phi$ is regular at each $(t,z_e)$ for $t\in \mathbb{R}$, the functions $\mu_1,\ldots,\mu_r$, which are constant on the level sets of ${\bf J}^{\Phi}$, obey $d\mu_1\wedge \ldots\wedge d\mu_r\neq 0$ on each $(t,z_e)$ for every $t\in \mathbb{R}$. This gives rise to a coordinate system $\{t,x_1,\ldots,x_n,y_1,\ldots,y_s,\mu_1,\ldots,\mu_r\}$ on an open neighbourhood in $\mathbb{R}\times P$ containing $\mathbb{R}\times\{z_e\}$. Hence,
\[
\frac{\partial h}{\partial y_i}\bigg|_{{\bf J}^{\Phi-1}(\mu_e)}\!\!\!\!\!\!\!=0,\quad\frac{\partial \langle {\bf J}^\Phi-\mu_{e},\ \xi(t)\rangle}{\partial y_i}=0,\qquad \forall t\in \mathbb{R},\qquad i=1,\ldots,s.
\]
It is worth noting that the above does not need to hold away from ${\bf J}^{\Phi-1}(\mu_e)$ since $y_1,\ldots,y_s$ were defined just as a mere smooth extension from ${\bf J}^{\Phi-1}(\mu_e)$. Moreover, 
\[
\left(\frac{\partial}{\partial y_j} \frac{\partial h}{\partial y_i}\right)\bigg|_{{\bf J}^{\Phi-1}(\mu_e)}\!\!\!\!\!\!\!\!\!\!\!=0,\quad \qquad \left(\frac{\partial}{\partial x_k} \frac{\partial h}{\partial y_i}\right)\bigg|_{{\bf J}^{\Phi-1}(\mu_e)}\!\!\!\!\!\!\!\!\!\!\!=0,
\]
\[
\frac{\partial}{\partial y_j} \frac{\partial \langle \mathbf{J}^\Phi-\mu_{e},\ \xi(t)\rangle}{\partial y_i}=0,\qquad \frac{\partial}{\partial x_k} \frac{\partial \langle \mathbf{J}^\Phi-\mu_{e},\ \xi(t)\rangle}{\partial y_i}=0,
\]
for all $t\in  \mathbb{R}$ with $i,j=1,\ldots,s$ and  $k=1,\ldots,n$. 
The first and second relations above hold because the derivative on the left depends on ${\bf J}^{\Phi-1}(\mu_e)$ only on $\partial h/\partial y_i$ within ${\bf J}^{\Phi-1}(\mu_e)$. Nevertheless, in the chosen coordinate system, the Hessian of $h_{\xi(t)}$ restricted to $\{t\}\times P$ on $T_{(t,z_e)}{\bf J}^{\Phi-1}(\mu_e)\cap \ker\eta|_{(t,z_e)}$ coincides with $\widehat M(t)$. Hence, we can use the functions $h_{\xi(t)}$ to study $\widehat{M}(t)$ and $M(t)$.

\section{A two-state quantum system}
\label{Sec::QuantumExample}

Let us analyse a quantum mechanical system defined by a $t$-dependent Schr\"odinger equation on a finite-dimensional Hilbert space so as to study its  relative equilibrium points with respect to the group of symmetries of $t$-dependent Schr\"odinger equations given by multiplying by complex non-zero numbers. In particular, this section applies some of the techniques of our paper to a two-level quantum system under the effect of a $t$-dependent Hermitian Hamiltonian operator $\widehat{H}(t)$, which can be induced, for instance, by a spin-magnetic interaction with a drift term. States of the two-level system are elements of the Hilbert space $\mathbb{C}^2$, but physical relevance has only non-zero ones. It is known that $\mathbb{C}^n$  admits a real differential structure that makes $\mathbb{C}^n$ globally homeomorphic to $\mathbb{R}^{2n}$. Hence, the Hilbert space describing the two-level system is two-dimensional as a complex manifold, while it is a four-dimensional manifold as a real one. The evolution of such a system is described by the action of the Lie group, $U_2$, of unitary automorphisms on $\mathbb{C}^2$. More specifically, the solution to the $t$-dependent Schr\"odinger equation induced by $\widehat{H}(t)$ from $t=0$ with an arbitrary initial state $\Psi_0\in \mathbb{C}^2$ takes the form $\Psi(t)=U_t\Psi_0$ for a  curve $\mathbb{R}\ni t\mapsto U_t\in U_2$. Recall that the $t$-dependent Schr\"odinger equation associated with $\widehat{H}(t)$ reads
\begin{equation}\label{Eq:Sch}
i\frac{d\Psi(t)}{dt}=\widehat{H}(t)\Psi(t),
\end{equation}
where $\widehat{H}(t)$, for every $t\in \mathbb{R}$, is assumed to be a Hermitian Hamiltonian operator acting on $\mathbb{C}^2$.

Let us consider a basis of the vector space (over the reals) of Hermitian operators on $\mathbb{C}^2$, let us say $\mathfrak{u}_2^*$. This basis can be obtained by giving a basis of the form $\{\widehat{S}_j:=\frac{1}{2}\sigma_j\}_{j=1,2,3}$, where $\sigma_1,\sigma_2,\sigma_3$ are the Pauli matrices, and the $2\times2$ identity matrix $\widehat{I}$. To simplify our further calculations, we introduce the Lie bracket on $\mathfrak{u}^*_2$ given by
\[
\co A,B \cc:=-i [A,B],\qquad \forall A,B\in \mathfrak{u}_2^*,
\]
where $[\cdot,\cdot]$ is the operator commutator for endomorphisms on $\mathbb{C}^2$. Then, the commutation relations in the given basis are
\[
\co \widehat{I}, \widehat{S}_j\cc =0,\quad \co \widehat{S}_j,\widehat{S}_k \cc =\sum^3_{l=1}\epsilon_{jkl} \widehat{S}_l,\quad k,j=1,2,3,
\]
where $\epsilon_{jkl}$, with $j,k,l=1,2,3$, are the {\it Levi-Civita symbols}. 

 In the presence of an external magnetic field $\vec{B}(t):=B(t)(B_1,B_2,B_3)$, where $B(t)$ is any $t$-dependent function,  applied to a spin $1/2$ particle and under the action of a drift term $B(t)B_0\widehat{I}$, the $t$-dependent Hamiltonian operator takes the form
\[
\widehat{H}(t)=B(t)B_0\widehat{I} + \vec{B}(t)\cdot \vec{S},
\]
where $\vec{S}:=(\widehat{S}_1,\widehat{S}_2,\widehat{S}_3)$. Note that $\widehat{H}(t)$ is Hermitian for every $t\in\mathbb{R}$.

Since $\mathbb{C}^2$ is diffeomorphic to $\mathbb{R}^4$ as a manifold, then a point $(z_1,z_2)\in \mathbb{C}^2$ can be represented by $\Psi:=(q_1,p_1,q_2,p_2)\in\mathbb{R}^4$, where $q_i= \mathfrak{Re}(z_i)$ and $p_i=\mathfrak{Im}(z_i)$ for $i=1,2$. Thus, the $t$-dependent Schr\"odinger equation (\ref{Eq:Sch}) of our problem reads
\begin{equation}\label{Eq:SchrodingerArray}
                \frac{\rm d}{{\rm d}t}\!\left[\begin{array}{c}q_1\\p_1\\q_2\\p_2\end{array}\right]\!=\!
                \frac 12 B(t)\!\left[\begin{array}{cccc}0&2B_0\!+\!B_3&-B_2&B_1\\
                -2B_0\!-\!B_3&0&-B_{1}&-B_2\\
                B_{2}&B_{1}&0&2B_0\!-\!B_3\\
                -B_{1}&B_2&-\!2B_0\!+\!B_3&0
                \end{array}\right]\!\left[\begin{array}{c}q_1\\p_1\\q_2\\p_2\end{array}\right].
\end{equation}
The manifold $\mathbb{R}\times \mathbb{C}^2\simeq \mathbb{R}^5$ is related to    a natural cosymplectic manifold $(\mathbb{R}\x\mathbb{C}^2,\omega_S:=dq_1\wedge dp_1+dq_2\wedge dp_2,\eta_S:=dt)$, where $t$ is the natural coordinate on $\mathbb{R}$ understood as a coordinate on $\mathbb{R}\x\mathbb{C}^2$. The solutions of system (\ref{Eq:SchrodingerArray}) can be geometrically described as the integral curves, parametrised by $t$, of the evolution vector field on $\mathbb{R}\x\mathbb{C}^2$ given by 
\begin{equation}\label{eq:ESQVF}
R+B(t)(B_0X_0+B_1X_1+B_2X_2+B_3X_3),
\end{equation}
where $R=\partial/\partial t$ is the Reeb vector field and $X_0,\ldots,X_3$ are the vector fields on $\mathbb{R}\x\mathbb{C}^2$ of the form
{\small
\begin{align}
                \label{Eq:BXalpha} 
                X_0 &:={p_1}\frac{\partial}{\partial
                        {q_1}}-{q_1}\frac{\partial}{\partial
                        {p_1}}+{p_2}\frac{\partial}{\partial
                        {q_2}}-{q_2}\frac{\partial}{\partial {p_2}},\,\,\,
                X_1 :=\frac 12\left({p_2}\frac{\partial}{\partial
                        {q_1}}-{q_2}\frac{\partial}{\partial
                        {p_1}}+{p_1}\frac{\partial}{\partial
                        {q_2}}-{q_1}\frac{\partial}{\partial {p_2}}\right), \\ 
                \!\!X_2 &:= \frac 12 \left( -{q_2}\frac{\partial}{\partial {q_1}}
                -{p_2}\frac{\partial}{\partial {p_1}} +{q_1}\frac{\partial}{\partial
                        {q_2}} + {p_1}\frac{\partial}{\partial {p_2}}\right),\,\,\, 
                X_3:=\frac 12\left({p_1}\frac{\partial}{\partial
                        {q_1}}-{q_1}\frac{\partial}{\partial
                        {p_1}}-{p_2}\frac{\partial}{\partial
                        {q_2}}+{q_2}\frac{\partial}{\partial {p_2}}\right).
\end{align}
}
Their commutation relations read
\[
[X_0,X_j]=0,\quad [X_1,X_2]=-X_3,\quad [X_2,X_3]=-X_1,\quad [X_3,X_1]=-X_2,\qquad j=1,2,3.
\]
The cosymplectic manifold $(\mathbb{R}\x\mathbb{C}^2,\omega_S=dq_1\wedge dp_1+dq_2\wedge dp_2,\eta_S=dt)$ allows us to write that $X_0,\ldots,X_3$ are Hamiltonian vector fields with Hamiltonian functions $h_0,\ldots,h_3$ of the form
\begin{equation*}
\begin{gathered}
h_0(\Psi)=\frac{1}{2}\langle\Psi,\widehat{I}\Psi\rangle=\dfrac{1}{2}\left(q_1^2+q_2^2+p_1^2+p_2^2 \right), \qquad h_1(\Psi)=\frac{1}{2}\langle\Psi,\widehat{S}_1\Psi\rangle=\frac 12(p_1p_2+q_1q_2), \\
h_2(\Psi)=\frac{1}{2}\langle\Psi,\widehat{S}_2\Psi\rangle=\frac 12(q_1p_2- q_2p_1), \qquad h_3(\Psi)=\frac{1}{2}\langle\Psi,\widehat{S}_3\Psi\rangle=\dfrac{1}{4}\left( p_1^2+q_1^2-p_2^2-q_2^2 \right).
\end{gathered}
\end{equation*}
The functions $h_1,h_2,h_3$ are functionally independent and $h_0^2=4(h_1^2+h_2^2+h_3^2)$.

Then, the $t$-dependent Schr\"odinger equation in coordinates \eqref{Eq:SchrodingerArray} can be associated with an evolution vector field, $E_h$, on $\mathbb{R}\x\mathbb{C}^2$ induced by the Hamiltonian function $h\in C^\infty(\mathbb{R}\x\mathbb{C}^2)$ given by
\begin{equation}\label{eq:HamSc}
h(t,\Psi):=B(t)\sum_{\alpha=0}^3B_\alpha h_\alpha(\Psi),\qquad t\in \mathbb{R},\qquad\Psi\in \mathbb{C}^2.
\end{equation}
In other words, the solutions $(q_1(t),p_1(t),q_2(t),p_2(t))$ to (\ref{Eq:SchrodingerArray}) designate integral curves $t\mapsto (t,q_1(t),p_1(t),q_2(t),p_2(t))$ of $E_{h}$.  
Let us write  $\Psi=(z_1,z_2)$, with $z_1,z_2
\in \mathbb{C}$. Then, we have a Lie group action $\Phi:U_1\times \mathbb{C}^2\ni (e^{i\theta},z_1,z_2)\mapsto (e^{-i\theta}z_1,e^{-i\theta}z_2)\in \mathbb{C}^2$. Indeed, this Lie group action gives rise to a Lie group of symmetries of (\ref{Eq:Sch}). Note that the fundamental vector fields of $\Phi$ are spanned by $X_0$ (considered as a vector field on $\mathbb{C}^2$).  
Let us consider the Lie group action on $\mathbb{R}\x\mathbb{C}^2$ with  fundamental vector fields $\langle X_0\rangle$ of the form
\[
\Phi:SO_2\times \mathbb{R}\x\mathbb{C}^2\ni(\theta;t,q_1,p_1,q_2,p_2)\mapsto(t,(R_\theta\otimes R_\theta)(q_1,p_1,q_2,p_2))\in\mathbb{R}\x\mathbb{C}^2,
\]
where $SO_2$ is the special orthogonal $2\x2$ matrix group and $R_\theta$ satisfies
\[
R_\theta=\left(\begin{array}{cc}\cos\theta&\sin\theta\\-\sin\theta&\cos\theta\end{array}\right)\in SO_2,\qquad R_\theta\left(\begin{array}{c}
     q_j  \\
     p_j
\end{array}\right)=\left(\begin{array}{cc}\cos\theta&\sin\theta\\-\sin\theta&\cos\theta\end{array}\right)\left(\begin{array}{c}q_j\\p_j\end{array}\right),\qquad j=1,2.
\]
Equivalently, one can understand $\Phi$ as a Lie group action
\[
\Phi^C:U_1\times \mathbb{R}\times \mathbb{C}^2\ni(e^{i\theta};t,z_1,z_2)\mapsto (t,e^{-i\theta}z_1,e^{-i\theta}z_2)\in\mathbb{R}\times\mathbb{C}^2.
\]
Note that the Lie group action $\Phi$ leaves invariant the Hamiltonian function (\ref{eq:HamSc}). Moreover, $\Phi$ is cosymplectic, i.e. $\Phi_g^*\omega_S=\omega_S$ and $\Phi_g^*\eta_S=\eta_S$ for every $g\in SO_2$.
Additionally, it has associated a momentum map $\mathbf{J}^\Phi: \mathbb{R}\x\mathbb{C}^2\rightarrow  \mathfrak{so}_2^*$ given by
\[
\mathbf{J}^\Phi(t,q_1,p_1,q_2,p_2):=h_0(q_1,p_1,q_2,p_2),
\]
where $\mathfrak{so}^*_2\simeq \mathbb{R}^*$. Note that $0\neq \mu\in\mathfrak{so}^*_2 $ is a regular value of ${\bf J}^\Phi$ and $\mu=0$ is not a weak regular value of ${\bf J}^{\Phi}$ because $T_{(t,0,0,0,0)}{\bf J}^{\Phi}=0$ but ${\bf J}^{\Phi-1}(0)=\{(t,0,0,0,0):t\in \mathbb{R}\}$.  Hence, $T_{(t,0,0,0,0)}{\bf J}^{\Phi-1}(0)\neq \ker T_{(t,0,0,0,0)}{\bf J}^\Phi$. Then, for $\mu\neq 0$, the $\mathbf{J}^{\Phi-1}(\mu)$ is a submanifold given by
\[
\mathbf{J}^{\Phi-1}(\mu)=\{(t,q_1,p_1,q_2,p_2)\,:\,q_1^2+p_1^2+q_2^2+p_2^2=2\mu,t\in \mathbb{R}\}=\mathbb{R}\times A_{\mu},
\]
where
\[
A_\mu=\{(q_1,p_1,q_2,p_2)\in \mathbb{R}^4\,:\,q_1^2+p_1^2+q_2^2+p_2^2=2\mu\},
\]
is a three-dimensional sphere in $\mathbb{R}^4\simeq \mathbb{C}^2$ centred at $0$ and $\sqrt{2\mu}$ can be understood as its radius. We can therefore write $A_\mu\simeq \mathbb{S}^3$, where $\mathbb{S}^3$ is the three-dimensional sphere in $\mathbb{R}^4$ with centre at $0$ and radius equal to one. Since $\mathfrak{so}_2^*$ is isomorphic to $\mathbb{R}^*$ and $SO_2$ is abelian, the coadjoint action of $SO_2$ on $\mathfrak{so}_2^*$ is trivial (every element of $SO_2$ acts as the identity in $\mathfrak{so}_2^*$) and, since $h_0$ is invariant relative to $SO_2$, then ${\bf J}^{\Phi}$  is ${\rm Ad}^*$-equivariant. Moreover, the isotropy group of every $\mu\in \mathbb{R}^*\backslash\{0\}$ is $SO_2$, i.e. $G_\mu=SO_2$ for every $\mu\neq 0$. Since $SO_2$ is diffeomorphic to the one-dimensional sphere in $\mathbb{R}^2$, i.e. the circle with radius one and centre at $0$ in $\mathbb{R}^2$, one has that
\[
(\mathbb{R}\times A_\mu)/G_\mu\simeq \mathbb{R}\times (S^3/S^1).
\]
It is known that $\mathbb{S}^1$ acting on $\mathbb{S}^3$ gives rise to a space of orbits diffeomorphic to $\mathbb{S}^2$. Hence,
\[
{\bf J}^{\Phi-1}(\mu)/G_\mu\simeq \mathbb{R}\times S^2.
\]
In particular, one has the coordinates on ${\bf J}^{\Phi-1}(\mu)$ given by $\{t,\varphi,\theta_1,\theta_2\}$  so that the points in ${\bf J}^{\Phi-1}(\mu)$ can be parametrised by
\begin{align*}
q_1=&\sqrt{2\mu}\sin \varphi \cos \theta_1,\quad p_1= \sqrt{2\mu}\sin\varphi\sin\theta_1,\\ q_2=&\sqrt{2\mu}\cos\varphi\cos\theta_2,\quad p_2=\sqrt{2\mu}\cos\varphi\sin\theta_2,
\end{align*}
with $t\in \mathbb{R}$, $\varphi\in ]0,\pi/2[$, while $\theta_1\in[0,2\pi[$ and $\theta_2\in[0,2\pi[$. Then,
\[
\iota_{\mu}^*\omega=\mu\sin (2\varphi)d\varphi\wedge d(\theta_1-\theta_2).
\]
Note that $\iota^*_\mu\omega$ is degenerate at $\varphi\in \{0,\pi/2\}$, but this is due to the fact that the chosen coordinates are not well defined at such values. A proper result with a non-degenerate $\iota^*_\mu\omega$ can be obtained by means of other, properly defined, coordinates. 
Since the Lie group action of $SO_2$ on ${\bf J}^{\Phi-1}(\mu)$ reads
\[
e^{i\theta}(t,\varphi,\theta_1,\theta_2)=(t,\varphi,\theta_1-\theta,\theta_2-\theta),
\]
then one obtains the submersion
$\pi_\mu:(t,\varphi,\theta_1,\theta_2)\in {\bf J}^{\Phi-1}(\mu)\mapsto (t,\varphi,\theta_1-\theta_2)\in (\mathbb{R}\times A_\mu)/G_\mu$.
Obviously, $\{t,\varphi,\theta_1-\theta_2\}$ are local coordinates on ${\bf J}^{\Phi-1}(\mu)/G_\mu$ and one can define
\[
\eta_\mu:=dt,\qquad \omega_\mu:=\mu\sin (2\varphi)d\varphi\wedge d(\theta_1-\theta_2).
\]
It is worth stressing that $\iota_\mu^*\omega=\pi^*_\mu\omega_\mu$.

Let us write $h_{\xi(t)}$ in coordinates
\[
h_{\xi(t)}=B(t)\left[\sum_{\alpha=0}^3B_\alpha h_\alpha\right]-(h_0-\mu)\xi(t),
\] 
for certain $B_0,B_1,B_2,B_3\in \mathbb{R}$.
Then, the critical points of $h_{\xi(t)}$ for a fixed $t$ are given by the solutions to the system of equations
\begin{equation}\label{eq::1}
\begin{gathered}
q_1 (2 B_0+B_3-2 \xi(t)/B(t) )+B_1 q_2+B_2 p_2=0,\\ q_2 (2 B_0-B_3-2 \xi(t)/B(t) )+B_1 q_1-B_2 p_1=0,\\ p_1 (2 B_0+B_3-2 \xi(t)/B(t) )+B_1 p_2-B_2 q_2=0,\\ B_1 p_1 + 2 B_0 p_2 + B_2 q_1 - p_2 (B_3 + 2 \xi(t)/B(t))=0,
\end{gathered}
\end{equation}
which amount to
\begin{equation}\label{eq::2}
B(t)\left(B_0\widehat{I}+\sum_{\alpha=1}^3B_\alpha \widehat{S}_\alpha\right)\left[\begin{array}{c}q_1+ip_1\\q_2+ip_2\end{array}\right]=\xi(t)\left[\begin{array}{c}q_1+ip_1\\q_2+ip_2\end{array}\right],
\end{equation}
which determines the  relative equilibrium points for every $t\in\mathbb{R}$ with respect to $\Phi$. Indeed, these are the common eigenvectors of the operators $\widehat{H}(t)=B(t)(B_0\widehat{I}+B_1\widehat{S}_1+B_2\widehat{S}_2+B_3\widehat{S}_3)$ for every $t\in \mathbb{R}$. Note that each $\widehat{H}(t)$ is Hermitian and only has real eigenvalues. This is natural:  relative equilibrium points relative to the Lie group action of symmetries given by multiplying by a complex non-zero number are given by the eigenvalues of the $\widehat{H}(t)$ that are common for every value of $t$.

Note that the reduced Hamiltonian system has a Hamiltonian function that can be written by means of 
\[
k_0=\mu,\quad k_1=\frac 12\mu\sin(2\varphi)\cos\theta,\quad k_2=-\frac 12\mu\sin(2\varphi)\sin\theta,\quad k_3=-\frac 12 \mu\cos(2\varphi),
\]
where $\theta:=\theta_1-\theta_2$  and the function on ${\bf J}^{\Phi-1}(\mu)/G_\mu$ of the form
\[
k_\mu(t,[\Psi]):=B(t)\sum_{\alpha=0}^3B_\alpha k_\alpha([\Psi]).
\]
Let us determine the  relative equilibrium points for  $B_0=B_1=B_2=0$ and $B_3=1$.  In view of (\ref{eq::1}) and (\ref{eq::2}), the  relative equilibrium points are given by points in $\mathbb{C}^2$ given by
\[\langle
(1,0)\rangle_\mathbb{C} \cup \langle (0,1)\rangle_{\mathbb{C}}.
\]
The stability close to the equilibrium points of the projected space are given, in our methods, by the Hessian of $k_3$ and they are undetermined by the standard criteria \cite{AM78, MS88} since the Hessian of $k_3$ is degenerated  as it depends only on $\varphi$. Geometrically, it can be proved that the evolution in $\mathbb{S}^2$ leaves invariant a Riemannian metric \cite{CCJL19} on $\mathbb{S}^2$, which involves that the reduced system is such that the evolution leaves the distance of a solution to the equilibrium point invariant over the time and the reduced equilibrium points are stable.

\section{Relative equilibrium points of \texorpdfstring{$n$-}sstate quantum system}\label{Sec:SREP}

Let us consider a more general quantum mechanical system than in the previous section, namely a system given by the $t$-dependent Schr\"odinger equation on a finite-dimensional Hilbert space $\mathbb{C}^n$ related to a $t$-dependent Hermitian Hamiltonian operator $\widehat{H}(t)$ of the form
\begin{equation}
\label{Eq::SchrodingerEq}
    i\frac{d\psi}{dt}=\widehat{
    H}(t)\psi,\qquad\forall\psi\in\mathbb{C}^n,\qquad\forall t\in\mathbb{R}.
\end{equation}
In particular, we will focus on determining  relative equilibrium points of this system. The reduction of (\ref{Eq::SchrodingerEq}) to the projective space $\mathbb{P}\mathbb{C}^n$, namely the space of linear one-dimensional spaces in $\mathbb{C}^n$, will be proven to be stable at its equilibrium points as a consequence of the same remarks of the previous section and the results given in \cite{CCJL19}.

First, let us introduce basic notions essential to describe the systems (\ref{Eq::SchrodingerEq}). The states of an $n$-level quantum system are the elements of the Hilbert space $\mathbb{C}^n$, and any Hilbert basis in $\mathbb{C}^n$ defines a real global chart on $\mathbb{C}^n$. Indeed, let $\{e_j\}_{1,\ldots,n}$ be an orthonormal basis of $\mathbb{C}^n$ relative to its canonical inner product $\langle\cdot,\cdot\rangle:\mathbb{C}^n\times \mathbb{C}^n\rightarrow  \mathbb{C}$. Then, the functions $q_j,p_j:\mathbb{C}^{n}\rightarrow  \mathbb{R}$ , with $j=1,\ldots,n$, defined by
\[
\langle e_j,\psi\rangle=:q_j(\psi)+ip_j(\psi),\qquad j=1,\ldots,n,\quad \forall\psi\in\mathbb{C}^{n},
\]
define a real global chart on $\mathbb{C}^{n}$.
Recall that $\widehat{H}(t)$ is a Hermitian Hamiltonian operator with respect to the above inner product for every $t\in\mathbb{R}$. Since $\mathbb{C}^n\simeq\mathbb{R}^{2n}$, then, at each $\tilde{\psi}\in \mathbb{C}$,  there exists an $\mathbb{R}$-linear isomorphism $\psi\in\mathbb{R}^{2n}\simeq \mathbb{C}^n\mapsto \psi_{\tilde{\psi}}\in T_{\tilde{\psi}}\mathbb{R}^{2n}\simeq T_{\tilde{\psi}}\mathbb{C}^n$, where
\[
\psi_{\tilde{\psi}}f:=\frac{d}{dt}
\bigg|_{t=0}f(\tilde{\psi}+t\psi),\qquad \forall f\in C^\infty(\mathbb{C}^{n}).
\]
Therefore, one can introduce an anti-symmetric, non-degenerated two-form $\omega$ on $\mathbb{C}^n$ of the form
\begin{equation}
\label{Eq::QuantuSForm}
    \omega_\psi(\psi_{1\tilde{\psi}},\psi_{2\tilde{\psi}}):=\mathfrak{Im}{\langle\psi_1,\psi_2\rangle},\qquad \forall\psi,\psi_1,\psi_2\in\mathbb{C}^{n}.
\end{equation}
In coordinates $\{q_j,p_j\}
_{j=1,\ldots,n}$, one has $\omega_n=\sum^n_{j=1}dq_j\wedge dp_j$. Since $\omega$ is closed, one has that $\omega_n$ is a symplectic form on $\mathbb{C}^{n}$ with Darboux coordinates $\{q_1,\ldots,q_n,p_1,\ldots,p_n\}$.
Let $\mathfrak{u}^*_n$ denote the real vector space of Hermitian operators on $\mathbb{C}^n$. Then, every observable on $\mathbb{C}^n$, namely $\widehat{A}\in\mathfrak{u}^*_n$, leads to a real function on $\mathbb{C}^{n}$ of the form
\[
f_{\widehat{A}}(\psi):=\frac{1}{2}\langle\psi,\widehat{A}\psi\rangle,\qquad \forall \psi\in\mathbb{C}^{n},
\]
giving rise to the Hamiltonian vector field
\[
X_{\widehat{A}}:=\{\cdot,f_{\widehat{A}}\},
\]
where $\{f,g\}:=\omega(X_f,X_g)$ for every real valued functions $f,g\in C^\infty(\mathbb{C}^{n})$ is a Poisson bracket.

The integral curves of the $t$-dependent Hamiltonian vector field $X_{\widehat{H}(t)}$ associated with $f_{\widehat{H}(t)}$ correspond to the solutions of the $t$-dependent Schr\"odinger equation \eqref{Eq::SchrodingerEq} (see  \cite{CCJL19} and references therein for details).

Now, let us proceed to the cosymplectic setting. The manifold $\mathbb{C}^n$ can be related to the manifold $(t,z_1,\ldots,z_n)\in\mathbb{R}\x\mathbb{C}^{n}\simeq\mathbb{R}^{2n+1}\ni(t,q_1,p_1,\ldots,q_n,p_n)$ and is endowed with the natural cosymplectic structure $(\mathbb{R}\x\mathbb{C}^{n},{\rm pr}_{\mathbb{C}^{n}}^*\omega_n,dt)$ where ${\rm pr}_{\mathbb{C}^{n}}:\mathbb{R}\x\mathbb{C}^{n}\rightarrow \mathbb{C}^{n}$ is the canonical projection onto the second factor and $t$ is the pull-back to $\mathbb{R}\times \mathbb{C}^n$ of the natural variable in $\mathbb{R}$. The solutions of \eqref{Eq::SchrodingerEq} are the curves $z(t)$ such that $(t,z(t))$ is an integral curve of the evolution vector field $E_{f_{\widehat{H}(t)}}=R+X_{f_{\widehat{H}(t)}}$, where $R=\frac{\partial}{\partial t}$ is the Reeb vector field.

The Lie group action of the form
\begin{align*}
\Phi_n:SO_2\x\mathbb{R}\x\mathbb{C}^{n}&\rightarrow  \mathbb{R}\x\mathbb{C}^{n},\\
(R_\theta,t,q_1,p_1,\ldots,q_n,p_n)&\mapsto(t, (R_\theta\otimes\ldots\otimes R_\theta)(q_1,p_1,\ldots,q_n,p_n)),
\end{align*}
gives rise to a Lie group of symmetries of $E_{\widehat{H}(t)}$. Moreover, the action of each element of $SO_2$ leaves invariant the canonical inner product on $\mathbb{C}^{n}$. Therefore, $\Phi_n$ leaves $\omega$ and $\eta$ invariant and hence $\Phi_n$ is a cosymplectic Lie group action. Note that $\Phi_n$ also leaves invariant the Hamiltonian function $f_{\widehat{H}(t)}$.

A cosymplectic momentum map $\mathbf{J}^{\Phi_n}$ is given by
\[
\mathbf{J}^{\Phi_n}:\mathbb{R}\x\mathbb{C}^{n}\ni(t,q_1,p_1,\ldots,q_n,p_n)\mapsto\frac{1}{2}\sum^n_{i=1}(q^2_i+p^2_i)\in\mathfrak{so}^*_2,
\]
where $\mathfrak{so}^*_2\simeq\mathbb{R}^*$. Similarly, $\mu\in\mathfrak{so}_2^*$ is a regular value of $\mathbf{J}^{\Phi_n}$ if $\mu\neq 0$. Otherwise, $T\mathbf{J}^{\Phi_n-1}(0)\neq \ker T\mathbf{J}^{\Phi_n}|_{\mathbf{J}^{\Phi-1}(0)}$ and hence $\mu=0$ is not a regular value of $\mathbf{J}^{\Phi_n}$. Thus, for $\mu\neq 0$, one has
\[
\mathbf{J}^{\Phi_n-1}(\mu)=\left\{(t,q_1,p_1,\ldots,q_n,p_n)\,:\,\sum^n_{i=1}(q_i^2+p^2_i)=2\mu,\,t\in\mathbb{R}\right\}=\mathbb{R}\times \mathbb{S}^{2n-1}.
\]
Since the coadjoint action of $SO_2$ on $\mathfrak{so}^*_2$ is trivial, a cosymplectic momentum map $\mathbf{J}^{\Phi_n}$ is ${\rm Ad}^*$-equivariant. Therefore, the isotropy group of every non-zero $\mu\in\mathfrak{so}^*_2$ is $G_\mu=SO_2$. Theorem \ref{Th:CoSymRed} and Corollary \ref{Cor::De} yield that the manifold
\[
(\mathbb{R}\x\mathbb{S}^{2n-1})/SO_2\simeq\mathbb{R}\x\left(\mathbb{S}^{2n-1}/\mathbb{S}^1\right),
\]
is a cosymplectic manifold. It is known that $\mathbb{S}^1$ acting on $\mathbb{S}^{2n-1}$ gives rise to a space of orbits diffeomorphic to the projective space  $\mathbb{PC}^{n}\simeq \mathbb{C}^n_\times/ \mathbb{C}_\times$, where $\mathbb{C}_\times:=\mathbb{C}\backslash \{0\}$ and $\mathbb{C}^n_\times:=\mathbb{C}^n\backslash \{0\}$, namely the space of one-dimensional subspaces in  $\mathbb{C}^n$ \cite{AM78}. Hence,
\[
M^\Delta_{\mu}=\mathbf{J}^{\Phi_n-1}(\mu)/G_\mu\simeq\mathbb{R}\x\mathbb{PC}^{n},
\]
for every non-zero $\mu\in\mathfrak{so}^*_2$.

From Proposition \ref{Prop:SingPPCos}, it follows that the  relative equilibrium points have the form
\[
(t,\psi(t))=\Phi_{n\,g(t)}(t,\psi_e),\qquad g(t)\in G_{\mu_e}=U_1,
\]
where $g(t)\in SO_2\simeq U_1$ is the evolution operator of the Schr\"odinger equation \eqref{Eq::SchrodingerEq} and $\Psi_e$ is an eigenvector for each $\widehat{H}(t)$. Thus, the  relative equilibrium points relative to the Lie group action of symmetries are given by the eigenvalues of a Hamiltonian operator $\widehat{H}(t)$ that are common for every $t\in \mathbb{R}$.

\section{Cosymplectic-to-symplectic reduction and gradient relative equilibrium points}
\label{Sec::GradRelEqPoints}
This section presents a new cosymplectic-to-symplectic reduction and a new associated type of relative equilibrium points: the gradient relative equilibrium points. Our cosymplectic-to-symplectic reduction does not rely on using Lie symmetries taking values in the kernel of the one-form $\eta$ of a cosymplectic manifold as standard cosymplectic Marsden--Weinstein reductions do, allowing for a broader application of this method in physics. Moreover, our reduction is more general then the cosymplectic-to-symplectic reduction devised by Albert in \cite[pg. 640]{Al89}, which is here retrieved as a particular case. Finally, our reduction is a modification of a Poisson reduction that cannot be fully described by means of the standard Poisson theory for a number of reasons to be described afterwards. 

Let us first recall the cosymplectic-to-symplectic reduction devised by Albert \cite[pg. 640]{Al89}.

\begin{theorem} Let $(M,\omega,\eta)$ be a cosymplectic manifold  let $Y$ be a vector field on $M$ such that

\begin{equation}\label{eq:ReAl}
\iota_Y\eta=1,\qquad \iota_Y\omega=-df
\end{equation}
for a certain $f\in C^\infty(M)$. Then, if the space $M/Y$ of orbits of $Y$ in $M$ is a manifold and $\pi_Y:M\rightarrow M/Y$ is a submersion, there exists a symplectic form $\omega_Y$ on $M/Y$ and a unique function $f_Y\in C^\infty(M/Y)$ such that $R$ projects onto the Hamiltonian vector field $X_{f_Y}$ on $M/Y$ relative to $\omega_Y$ and $\pi_Y^*f_Y=f$.
\end{theorem}

It is worth noting that the conditions (\ref{eq:ReAl}) imply that $Y=R-X_f$, i.e. $Y$  is an evolution vector field and $Rf=0$. Hence, this reduction is quite restrictive, although it allows for a reduction relative to a vector field that does not take values in $\ker \eta$.

Recall that a cosymplectic manifold $(M,\omega,\eta)$ is such that the restriction of $\omega$ to each leaf of the integral distribution $\ker \eta$ is symplectic. Hence, on each leaf, one can define a Poisson bivector associated with the restriction of $\omega$ to it, giving rise to a Poisson bivector on $M$. In fact, this is the Poisson bivector associated with the Poisson bracket (\ref{Eq::PoissonStructure}) naturally defined on cosymplectic manifolds. Consequently, one has the following result.

\begin{proposition} Every cosymplectic manifold $(M,\omega,\eta)$ gives rise to a unique Poisson bivector $\Lambda_{\omega,\eta}$ on $M$ that is tangent to the leaves of $\ker \eta$ and becomes the Poisson bivector associated with $\omega$ on each such a leaf.
\end{proposition}

In Darboux coordinates $\{t,x^i,p_i\}$ for $(M,\omega,\eta)$, one obtains
\[
\Lambda_{\omega,\eta}=\sum_{i=1}^n\frac{\partial}{\partial x^i}\wedge\frac{\partial}{\partial p_i}.
\]
Using Darboux coordinates, we obtain that $R$ and every $X_h$ are Lie symmetries of $\Lambda_{\omega,\eta}$. In spite of that, $R$ and every $\nabla h$, for $h\in C^\infty(M)$ with $Rh\neq 0$, are not Hamiltonian vector fields relative to the Poisson bivector $\Lambda_{\omega,\eta}$. In fact, they do not belong to the image of the induced vector bundle morphism $\Lambda_{\omega,\eta}^\sharp: \vartheta_p\in T^*M\mapsto (\Lambda_{\omega,\eta})_p(\vartheta_p,\cdot)\in TM$. This implies that the standard Marsden--Weinstein reduction cannot be directly applied to reduce the dynamics of such vector fields. Moreover, the special form of $\Lambda_{\omega,\eta}$ has special properties that are not shared by general Poisson bivectors and need to be specifically studied. Before we prove the main result of this section, let us first show the following lemma.

\begin{lemma}
\label{Lemm::LieBivector}
    Let $\Lambda_{\omega,\eta}$ be the Poisson bivector on $M$ associated with $(M,\omega,\eta)$. Then,
    \[
    \mathcal{L}_{\nabla \Upsilon}\Lambda_{\omega,\eta}=0,
    \]
    for $\Upsilon\in C^\infty(M)$ such that $\iota_{d(R\Upsilon)}\Lambda_{\omega,\eta}=0$.
\end{lemma}
\begin{proof}

Recall that $\nabla \Upsilon=(R\Upsilon)R+X_{\Upsilon}$. Since $\iota_{X_\Upsilon}\eta=0$, one has that $X_\Upsilon$ is tangent to the leaves of the integrable distribution $\ker\eta$. Moreover, the restriction of $X_\Upsilon$ to one of the leaves of the distribution $\ker\eta $ is a Hamiltonian vector field relative to the restriction of $\Lambda_{\omega,\eta}$ to such a leaf, which is symplectic. Indeed, for a  vector field $X$ taking values in $\ker \eta$, one has that $\iota_X\iota_{X_\Upsilon}\omega=X\Upsilon$ and then, on each integral leaf of $\ker \eta$, one has that $\iota_{X_\Upsilon}\omega=d\Upsilon_t$, where $\Upsilon_t$ is the restriction of $\Upsilon$ to the particular integral leaf of $\ker \eta$. Hence, $\mathcal{L}_{X_\Upsilon}\Lambda_{\omega,\eta}=0$. The assumption $\iota_{d(R\Upsilon)}\Lambda_{\omega,\eta}=0$ yields that $$\mathcal{L}_{(R\Upsilon)R}\Lambda_{\omega,\eta}=(R\Upsilon)\mathcal{L}_R\Lambda_{\omega,\eta}+(\iota_{d(R\Upsilon)}\Lambda_{\omega,\eta})\wedge R=0,$$ and the statement follows.
\end{proof}

Lemma \ref{Lemm::LieBivector} could be proved in a coordinate-dependent manner by using Darboux coordinates, but the above proof is intrinsic and illustrates more clearly the geometric properties of cosymplectic manifolds. 
Now, let us explain one of the main results of this section.

\begin{proposition}{\bf (The cosymplectic-to-symplectic reduction theorem)}
\label{Th::GradReduction}Let $(M,\omega,\eta)$ be a cosymplectic manifold with a Reeb vector field $R$. Let $\Upsilon\in C^\infty(M)$ be such that $\iota_{d(R\Upsilon)}\Lambda_{\omega,\eta}=0$ and $R\Upsilon\neq 0$ at any point of $M$. Assume that $M/\nabla \Upsilon$ is a manifold and $\pi_\Upsilon:M\rightarrow M/\nabla\Upsilon$ is a submersion. Then, $\Lambda_{\omega,\eta}$ projects onto a bivector field $\Lambda_\Upsilon$ on  $M/\nabla \Upsilon$ giving rise to a symplectic manifold. Moreover, if $h\in C^\infty(M)$ is such that $[\nabla \Upsilon, E_h]=0$, then $E_h$ projects onto a vector field $Y_k$ on $M/\nabla \Upsilon$  via $
\pi_\Upsilon$ and becomes a Hamiltonian vector field relative to the symplectic form induced by $\Lambda_\Upsilon$ on $M/\nabla \Upsilon$. In this latter case, $R\Upsilon$ is a constant and $Y_k$ admits a Hamiltonian function $k
\in C^\infty(M/\nabla \Upsilon)$ defined uniquely  by 
$$
\pi_\Upsilon^*k=h-\Upsilon/c-\int^t [(\nabla\Upsilon)(h-\Upsilon/c)]dt.
$$
\end{proposition}

\begin{proof} Recall that $\nabla \Upsilon=(R\Upsilon)R+X_\Upsilon$ for $R\Upsilon\neq 0$. Lemma \ref{Lemm::LieBivector} yields that $\mathcal{L}_{\nabla \Upsilon}\Lambda_{\omega,\eta}=0$ and then $\Lambda_{\omega,\eta}$ projects onto $M/\nabla \Upsilon$ via the projection $
\pi:M\rightarrow M/\nabla \Upsilon$ giving rise to a Poisson bivector $\Lambda_{\Upsilon}$. 

Let us prove that $\Lambda_{\Upsilon}$ gives rise to a symplectic form by reduction to absurd. If $\Lambda_{\Upsilon}$ is degenerate, it admits a non-trivial kernel at some point  $y\in M/\nabla \Upsilon$. In other words, there exists a nonzero  covector $\vartheta_y\in T^*_y(M/\nabla \Upsilon)$ such that $(\Lambda_{\Upsilon})_y(\vartheta_y,\cdot)=0$. Hence, $(\Lambda_{\omega,\eta})_x(\vartheta_y\circ \pi_{*x},\vartheta'_y\circ \pi_{*x})=0$ for every $x\in \pi^{-1}(y)$ and every $
\vartheta'_y\in T^*_y(M/\nabla \Upsilon)$. This implies that $\vartheta_y\circ \pi_{x*}$ is orthogonal relative to $\Lambda_{\omega,\eta}$ at $x$ to the annihilator of $\langle \nabla\Upsilon\rangle_x$. Let us denote by $\langle \nabla\Upsilon\rangle ^{\Lambda_{\omega,\eta}}_x$ the orthogonal to $\langle\nabla\Upsilon\rangle_x$ relative to $\Lambda_{\omega,\eta}$. Since $\Lambda_\Upsilon$ is a bivector field on an even-dimensional manifold, its kernel is even dimensional. Since $\vartheta_y\in \ker \Lambda_{\Upsilon}$ is not zero, there exists another linearly independent covector $\vartheta'_y\in T^*_y(M/\nabla \Upsilon)$ in $\ker \Lambda_\Upsilon$. Hence, their pullbacks via $\pi^*_x$ give two linearly independent elements in $\langle \nabla\Upsilon\rangle ^{\Lambda_{\omega,\eta}}_x$. Moreover, $dt$ belongs to $\langle \nabla\Upsilon\rangle ^{\Lambda_{\omega,\eta}}_x$ at $x$ because belongs to $\ker (\Lambda_{\omega,\eta})_x$. But $dt_x$ is not the pull-back via $\pi_{x*}$ of any element of $T^*_y(M/\nabla \Upsilon)$, because $\iota _{\nabla \Upsilon}dt=R\Upsilon \neq 0$. Hence, $\langle \nabla\Upsilon\rangle ^{\Lambda_{\omega,\eta}}_x$ has, at least, dimension three. But this is impossible, because $\Lambda_{\omega,\eta}$ has rank $2n$ and the orthogonal relative to $\Lambda_{\omega,\eta}$ to a subspace of codimension $k$ has, at most, dimension $k+1$. This is a contradiction, and $\Lambda_\Upsilon$ is nondegenerate and it gives rise to a symplectic form.

Consider now a vector field $E_h$. Since $[E_h,\nabla \Upsilon]=0$ by assumption, then $E_h$ projects onto a vector field $Z$ on $M/\nabla \Upsilon$. Hence, $\mathcal{L}_Z\Lambda_{ \Upsilon}=0$ and $Z$ is locally Hamiltonian relative to $\Lambda_{\Upsilon}$ and its associated symplectic form. Note that $\iota_{d( R\Upsilon)}\Lambda_{\omega,\eta}=0$ implies that $R\Upsilon$ is, in Darboux coordinates, a function depending only on time. Then,  $X_h(R\Upsilon)=0$ and
$$
0=[E_h,\nabla\Upsilon]=[R+X_h,(R\Upsilon )R+X_\Upsilon]=(R^2\Upsilon)R+[R,X_\Upsilon]+R\Upsilon[X_h,R]+[X_h,X_\Upsilon].
$$
Therefore, $R^2\Upsilon=0$ and since $R\Upsilon$ depends only on $t$ in Darboux coordinates, $R\Upsilon$ is a nonzero constant $c$. Thus,
\begin{equation}\label{eq:RedRedTime}
0=[E_h,\nabla \Upsilon]=X_{R\Upsilon}-cX_{Rh}-X_{\{h,\Upsilon\}}=X_c-cX_{Rh}-X_{\{h,\Upsilon\}}=-cX_{Rh}-X_{\{h,\Upsilon\}}=X_{-cRh-\{h,\Upsilon\}}.
\end{equation}
Hence, $-\{h,\Upsilon\}-cRh$ depends only on time in Darboux coordinates.

Since $\nabla \Upsilon$ projects onto zero on $M/\nabla\Upsilon$, the projection of $E_h$ onto $M/\nabla \Upsilon$ is the same as the projection of $E_h-\nabla\Upsilon/c=X_{h-\Upsilon/c}$. But $X_{h-\Upsilon/c}$ is a Hamiltonian vector field relative to $\Lambda_{\omega,\eta}$. In view of this and (\ref{eq:RedRedTime}), one has 
$$
\nabla \Upsilon(h-\Upsilon/c)=\{h,\Upsilon\}+c Rh-c=g(t)
$$ 
for a certain function $g(t)$ in Darboux coordinates.
Hence,  
$$
\nabla \Upsilon\left(h-\Upsilon/c-\frac 1c\int^tg(t')dt'\right)=0
$$
and $h-\Upsilon/c-\int^t[\nabla \Upsilon(h-\Upsilon/c)](t')dt'/c$ is the pull-back of a function on $M/\nabla \Upsilon$. Moreover, $\Lambda_{\omega,\eta}(d(h-\Upsilon/c-\int
^tg(t')dt'/c),\cdot)$ is  projectable onto $M/\nabla \Upsilon$ giving a vector field $\Lambda_\Upsilon(d(h-\Upsilon/c-\int^t g(t')dt'/c),\cdot)=Y_k=Z$, and the final result follows.
\end{proof}

Note that our reduction allows for studying general evolution vector fields, which is more general than in the case of the Albert's cosymplectic-to-symplectic reduction dealing with cases whose Hamiltonian is a first integral of the Reeb vector field. Moreover, our cosymplectic-to-symplectic reduction allows for the reduction of an evolution vector field relative to another vector field, which is a more general scheme than in Albert's reduction. In particular, Albert's reduction is a particular case of our reduction for $E_h=R$, $Y=R-X_f=\nabla (t-f)$, and $\Upsilon=t-f$ for a function $f$ such that $Rf=0$ and $t$ is a Darboux time coordinate. Since $R$ is a Hamiltonian vector field with zero Hamiltonian function $h=0$ and $c=R\Upsilon=1$, one has
$$
\nabla (t-f)(0-(t-f))=-1\Rightarrow \pi_Y^*k=-t+f+\int^tdt'=f.
$$
Hence, our reduction gives that $R$ projects onto a vector field on $M/Y$ with a Hamiltonian function $k$ whose pull-back to $M$ is $f$.

Then, we can define the so-called {\it gradient relative equilibrium point}.

\begin{definition}
\label{Def::GradRelEqPoint}Let $(M,\omega,\eta)$ be a cosymplectic manifold, let $h\in C^\infty(M)$ be a Hamiltonian function, and let $\Phi: G\times M\rightarrow  M$ be a cosymplectic Lie group action on $M$ whose fundamental vector fields are of the form $\xi_M=\nabla \Upsilon$ for some $\Upsilon\in C^\infty(M)$ and for some $\xi\in\mathfrak{g}$ such that $\iota_{d(R\Upsilon)}\Lambda_{\omega,\eta}=0$. Then, a {\it gradient relative equilibrium point} of $h$ is a point $z_e\in M$ such that
\[
\nabla h_{z_e}=(\xi_M)_{z_e},
\]
for a certain fundamental vector field $\nabla \Upsilon$ of $\Phi$.
\end{definition}

Note that $E_h$ is such that $\iota_{E_h}\eta$ is not zero. Therefore, at gradient relative equilibrium points, and in an open neighbourhood around them, $\nabla\Upsilon$ satisfies the condition given in our cosymplectic-to-symplectic reduction. After reducing, standard symplectic techniques to evaluate the stability of the projected system can be used.

It is worth commenting that the immediate part of our cosymplectic-to-symplectic reduction, the one about the existence of the projection of $\Lambda_{\omega,\eta}$, can be described as a particular type of Marsden--Weinstein Poisson reduction, but this theory does not take into account the special nature of our Poisson bivector induced by the cosymplectic structure and it does not cover our scheme of reduction by a gradient vector field that is not Hamiltonian relative to $\Lambda_{\omega,\eta}$.

\section{A reduced circular three-body problem}\label{Sec::RDTP}

This section proves that the cosymplectic Marsden--Weinstein reduction given in previous sections can be insufficient to study certain relevant physical problems and our cosymplectic-to-symplectic reduction is needed. 

Let us consider the physical problem given by three masses $\mu,1-\mu$, and $m$ moving on a plane and interacting gravitationally. We assume that the gravitational constant is equal to one for simplicity. Moreover, we also set $\mu$ to be much larger than $1-\mu$. Mathematically, one may assume $m=1$ for simplicity, while the model for a general value of $m$ follows from it straightforwardly.  Additionally, let us assume that the mass $1-\mu$ is spinning around $\mu$ in a circular stable motion with constant angular frequency $\varpi$. We assume that $m$ does not affect the motion of $\mu$ and $1-\mu$. Physically, this happens when $m$ is much smaller than $\mu$ and $1-\mu$. 
Moreover, we will skip the study of collisions. This just depicted model is sometimes called the {\it circular restricted three-body problem} \cite{Al89, GZ14}. The above series of assumptions is a standard approach in the literature (see for instance \cite[pg. 663]{AM78} and references therein).

Note that the centre of mass of our circular restricted three-body problem is located, under our assumptions, in the line between the masses $\mu$ and $1-\mu$, being the distances of the masses $\mu$ and $1-\mu$ to it equal to $r_1=1-\mu$ and $r_2=\mu$, respectively. This model is a quite good approximation for the motion of a three-body system Sun-Earth-satellite, where it is assumed that the Earth moves around the Sun in a circular motion with radius equal to one, i.e. $r_1+r_2=1$, and a fixed frequency $\varpi$; and the satellite moves being affected by the gravitational forces induced by the Sun and the Earth but without having any effect in the motion of the Sun and the Earth. There are many other astronomical systems, in particular types of asteroids of the Solar System, like the Jupiter Trojan asteroids, that can be described via our model. 

Mathematically, our model is described by a $t$-dependent Hamiltonian $h$ on the phase space of a plane, which amounts to a function on $\mathbb{R}\times T^*\mathbb{R}^2$, whose form, in adapted coordinates $\{t,r,\varphi,p_r,p_{\varphi}\}$ induced by polar coordinates in $\mathbb{R}^2$ and a time coordinate $t$, namely $\{t,r,\varphi\}$,  reads as follows
\[
    h(t,r,\varphi,p_r,p_\varphi)=\frac{p_r^2}{2}+\frac{p_\varphi^2}{2 r^2}-\frac{\mu}{[r^2+r^2_1+2rr_1\cos(\varphi-\varpi t)]^{1/2}}-\frac{1-\mu}{[r^2+r^2_2-2rr_2\cos(\varphi-\varpi t)]^{1/2}}.
\]
We will ignore technical problems related to the lack of differentiability of $h$, which has no relevance for our further discussion, and we will study the problem via a cosymplectic manifold $(\mathbb{R}\times T^*\mathbb{R}^2,\omega_{TB},\eta_{TB}=dt)$, where $\omega_{TB}$ is the pull-back to $\mathbb{R}\times T^*\mathbb{R}^2$ of the canonical symplectic form on $T^*\mathbb{R}^2$, namely  $\omega_{TB}=dr\wedge dp_r+d\varphi\wedge dp_\varphi$ in the chosen coordinates (see \cite{Al89} for a different approach using old techniques in cosymplectic geometry). The evolution vector field describing the dynamics of the system determined by $h$ is given by $R_{TB}+X_h$, namely
\begin{multline}\label{Eq::ExpreEvo}
\frac{\partial}{\partial t}-\bigg(\frac{\mu(r+r_1 \cos (\varphi - \varpi t))}{\left(r^2+2rr_1\cos(\varphi-\varpi t)+r_1^2\right)^{3/2}}+\frac{(1-\mu) (r-  r_2 \cos (\varphi -\varpi t))}{\left(r^2-2 r r_2 \cos (\varphi-\varpi t)+r_2^2\right)^{3/2}}-\frac{p_\varphi^2}{r^3}\bigg)\frac{\partial}{\partial p_r}+p_r\frac{\partial}{\partial r}\\+\frac{ p_\varphi}{r^2}\frac{\partial}{\partial \varphi}+\bigg(\frac{\mu r r_1 \sin (\varphi -\varpi  t)}{\left(r^2+2 r r_1 \cos (\varphi-\varpi t )+r_1^2\right)^{3/2}}-\frac{(1-\mu) r r_2 \sin (\varphi -\varpi t)}{\left(r^2-2 r r_2 \cos (\varphi -\varpi t)+r_2^2\right)^{3/2}}\bigg)\frac{\partial}{\partial p_\varphi}.
\end{multline}
The Hamilton equations corresponding to $h$ read
\begin{equation}\label{eq::ExprEqu}
\begin{gathered}
\frac{dr}{dt}=p_r,\qquad \frac{d\varphi}{dt}=\frac{p_\varphi}{r^2},
\\ 
\frac{dp_r}{dt}=\frac{p_\varphi^2}{r^3}-\frac{\mu(r+r_1 \cos (\varphi - \varpi t))}{\left(r^2+2rr_1\cos(\varphi-\varpi t)+r_1^2\right)^{3/2}}-\frac{(1-\mu) (r-  r_2 \cos (\varphi -\varpi t))}{\left(r^2-2 r r_2 \cos (\varphi-\varpi t)+r_2^2\right)^{3/2}},
\\
\frac{dp_\varphi}{dt}=\frac{\mu r r_1 \sin (\varphi -\varpi  t)}{\left(r^2+2 r r_1 \cos (\varphi-\varpi t )+r_1^2\right)^{3/2}}-\frac{(1-\mu) r r_2 \sin (\varphi -\varpi t)}{\left(r^2-2 r r_2 \cos (\varphi -\varpi t)+r_2^2\right)^{3/2}}.
\end{gathered}
\end{equation}

Consider the vector field on $\mathbb{R}^3$ given by
\[
Y=\frac{\partial}{\partial t}+{\varpi}\frac{\partial}{\partial \varphi}
\]
and let $\widehat{Y}$ be the  fundamental vector field of the lift to $\mathbb{R}\times T^*\mathbb{R}^2$ of the Lie group action of $\mathbb{R}$ on $T^*\mathbb{R}^2$ related to the flow of $Y$ associated with the same element of the Lie algebra of $\mathbb{R}$ (see Section \ref{Sec::AMomentumMap}).

The vector field $\widehat{Y}$ is a cosymplectic vector field, i.e. $\mathcal{L}_{\widehat{Y}}\omega_{TB}=0$ and $\mathcal{L}_{\widehat{Y}}\eta_{TB}=0$. In fact, it is the gradient vector field relative to the function $\Upsilon= t+p_\varphi\varpi$, which satisfies that $R_{TB}\Upsilon$ is a constant. Note that $\widehat{Y}$ is not a Hamiltonian vector field relative to $(\mathbb{R}\times T^*\mathbb{R}^2,\omega_{TB},\eta_{TB}=dt)$ since $\iota_{\widehat{Y}}\eta_{TB}\neq 0$. Moreover,   $\widehat{Y}$ is a Lie symmetry of the Hamiltonian function $h$, which follows from the fact that $\widehat{Y}$ takes the same form as $Y$ but in the coordinates $\{t,r,\varphi,p_r,p_\varphi\}$. At this level, it is evident that Theorem \ref{Th:CoSymRed} can not be applied.

It is relevant to us to find gradient relative equilibrium points of $h$, namely when $E_h$ is proportional to $\widehat{Y}$. Physically, this happens when the mass $m$ moves around the centre of mass at a fixed distance $r$ and frequency $\varpi$. Note that the notion of a  relative equilibrium point for the cosymplectic manifold $(\mathbb{R}\times T^*\mathbb{R}^2,\omega_{TB},\eta_{TB}=dt)$ does not apply to this case because, for instance, $\widehat Y$ is not Hamiltonian. Similarly, the techniques devised in \cite{LZ21} can not be used either as $\widehat Y$ is not tangent to $T^*\mathbb{R}^2$. Hence, we will employ Theorem \ref{Th::GradReduction}.

By Definition \ref{Def::GradRelEqPoint} a point $z_e\in R\times T^*\mathbb{R}^2$ is a gradient relative equilibrium point if $R_{TB}+X_h$ is proportional to $\widehat{Y}$. If this happens at a point $(t,r,\varphi,p_r,p_\varphi)$, then the last expression in the Hamilton equations (\ref{eq::ExprEqu}) is equal to zero and $\varphi=\varpi t+k\pi$, with $k\in \mathbb{Z}$, or $\varphi-\varpi t$ is such that the distance between the mass from $m$ to $\mu$ and from $m$ to $1-\mu$ are the same. In this latter case, one can prove that $\varphi-\varpi t=\Delta$ and $r\cos \Delta =\mu-1/2$. Note that we can restrict ourselves to $k\in \{0,1\}$. Moreover, the remaining equations for the gradient relative equilibrium points read
\begin{equation}\label{eq:Remaining}
\begin{gathered}
\frac{p_\varphi^2}{r^3}- \frac{\mu(r+r_1 \cos (\varphi - \varpi t))}{\left(r^2+2rr_1\cos(\varphi-\varpi t)+r_1^2\right)^{3/2}}-\frac{(1-\mu) (r-  r_2 \cos (\varphi -\varpi t))}{\left(r^2-2 r r_2 \cos (\varphi-\varpi t)+r_2^2\right)^{3/2}}=0,\\ p_r=0,\qquad \frac{p_\varphi}{r^2}=\varpi.
\end{gathered}
\end{equation}
Since the masses $1-\mu$ and $\mu$ spin around their centre of mass, which is located at $r=0$, with constant angular velocity $\varpi$ due to their gravitational attraction, one has that
\[
\frac{\mu}{(r_1+r_2)^{2}}=\varpi^2r_2\,\Rightarrow   \,\varpi=\pm 1.
\]
Note that the force is given by the relative distance between the masses, while the centripetal force is considered relative to the inertial reference system at the centre of mass of the system of $\mu$ and $1-\mu$. 

Let us consider then three options for the relations between $\varphi$ and $t$ for our gradient relative equilibrium points, namely $\varphi=\varpi t$, $\varphi=\varpi t+\pi$ and $\varphi=\varpi t+\Delta$. In the first case, one has the equations
\begin{equation}\label{eq:L1L2}
r=\frac{\mu}{(r+1-\mu)^2}\mp\frac{1-\mu}{(\mu-r)^2},
\end{equation}
which are nothing but the equations for the centripetal force of a circular motion induced by the gravitational force of the masses $\mu$ and $1-\mu$ when the three objects move in circles with a frequency $\varpi$ while keeping their positions along a line which turns around the origin with such a frequency. Then, (\ref{eq:L1L2}) leads to two quintic equations
\begin{multline}\label{eq:Quintic}
P_{\pm}(r,\mu):=r^5+(2-4 \mu ) r^4+\left(6 \mu ^2-6 \mu +1\right) r^3+\left(-4 \mu ^3+6 \mu ^2-(3\pm 1)\mu \pm 1\right) r^2\\+\left(\mu ^4-2 \mu ^3+(3\pm 2)\mu ^2\mp(4 \mu -2)\right) r-\mu^3\pm(1-\mu)^3 =0
\end{multline}
for the gradient relative equilibrium position of $r$, which has always a root in $]0,\infty[$ since the polynomial has negative value at $r=0$; the value of $\mu$ is approximately equal to 1 with $\mu<1$; and the value of the polynomial (\ref{eq:Quintic}) tends to infinity when $r$ does so. Each of the above two equations in (\ref{eq:L1L2}) has just one real solution. This can be seen by analysing the right- and left-hand side functions in (\ref{eq:L1L2}). Let us work out the gradient relative equilibrium points in an approximate manner. The quintic polynomial has a triple root $r=1$ for $\mu=1$. Let us write $r=1+\sum_{n\in \mathbb{N}}\delta^{n/3}x_n$ for certain constants $\{x_n\}_{n\in \mathbb{N}}$ and a parameter $\delta\geq 0$, and consider the quintic polynomical as an expression $P_{\pm}(r,\mu)=\sum_{n=0}^\infty P_{\pm n}(r)\delta^{1+n/3}$ for $\delta=1-\mu$, and we look for solutions of $P_\pm(r(\delta),\delta)=0$ for every $\delta$ in some $[0,\delta_{\max}[$. Note that for $\delta=0$, one has that $P_\pm(r,1)$ has a triple root $r=1$ and one obtains 
\[
0=P_\pm(r(\delta),\delta)=(\pm 1+ 3x^3)\delta+(\pm 2x+3x^4+9x^2y)\delta^{4/3}+\ldots,
\]
The convergence of solutions of $P_\pm(r(\delta),\delta)=0$ can be obtained by the implicit function theorem and writing $P_\pm(r(\delta),\delta)$ in an appropriate manner. 
Then, the equilibrium points to order $\delta^{1/3}$ are given by 
\[
r=1\mp   \sqrt[3]{\frac{1-\mu}{3 }},
\]
which are, up to the chosen  level of  approximation, the known values for the Hill spheres. Note that $r$ is positive. Then,  the model recovers two gradient relative equilibrium points for $k=0$. Let us call them $L_2$ and $L_1$ for the signs plus and minus in $r$, respectively.

Meanwhile, for $k=1$, the equations for the gradient relative equilibrium points read
\begin{equation}\label{eq:L3}
r=\pm\frac{\mu}{(r-1+\mu)^2}+\frac{1-\mu}{(\mu+r)^2}.
\end{equation}
Note that the case of (\ref{eq:L3}) with the minus sign in $\pm$ amounts to one of the equations in (\ref{eq:L1L2}) with $-r$. Since (\ref{eq:L1L2}) has only a real positive solution for each possibility of the signs, it turns out that (\ref{eq:L3}) has no positive solution with the minus sign in $\pm$ and only the case
\[
r=\frac{\mu}{(r-1+\mu)^2}+\frac{1-\mu}{(\mu+r)^2}
\]
has a physical interest. The above can be written as a polynomial in terms of $r$ and $\delta=1-\mu$ of the form
\[
0=(1-r)S(r)+\delta Q(r,\delta), \qquad S(r)=-(1+r)^2(1+r+r^2),
\]
\[
Q(r,\delta)=3+4r-2r^2-6r^3-4r^4+(-3+r+6r^2+6r^3)\delta+(-2r-4r^2)\delta^2+r\delta^3.
\]
For $\mu=1$, this gives a solution $r=1$. Let us assume $r=1+\lambda \delta$. 
This case has an approximate positive solution, obtained by skipping terms in second or higher powers in $\delta$, given by
\[
0=\lambda \delta S(1)+\delta Q(1,0)=\lambda 12\delta-\delta 5 \Rightarrow   r=1+(1-\mu) \frac{5}{12}.
\]
This is a known value of the Lagrange point $L_3$. 

There exist still another two Lagrange points, which can be recovered for $\varphi-\varpi t=\Delta$. We leave it as an exercise for the reader to verify that they recover the well-known Lagrange points $L_4$ and $L_5$. The main point is that, at all Lagrange points, we have
\[
R_{TB}+X_h=\nabla \Upsilon,
\]
where $\Upsilon=t+\varpi p_\varphi$, for every $t\in \mathbb{R}$. 
Let us recall that $\nabla \Upsilon$ is a symmetry of $\eta_{TB }$ and $\omega_{TB}$, but it does not take values in the kernel of $\eta_{TB}$, which makes the standard cosymplectic reduction impossible to be applied and Theorem \ref{Th::GradReduction} has to be used.

The projection from $\mathbb{R}\times T^*\mathbb{R}^2$ onto the quotient space $(\mathbb{R}\times T^* \mathbb{R}^2) /\nabla \Upsilon\simeq \mathbb{R}^2\times \mathbb{R}^2$ corresponding to the orbit space of the integral curves of $\nabla \Upsilon$ is given by
\[
\pi:(t,r,\varphi,p_r,p_\varphi)\in \mathbb{R}\times T^*\mathbb{R}^2\mapsto (r,\varphi-t\varpi ,p_r,p_\varphi)\in \mathbb{R}^2\times \mathbb{R}^2,
\]
where $\{r,\varphi',p_r,p_\varphi\}$ is the chosen global coordinate system on $\mathbb{R}^2\times \mathbb{R}^2$.

Recall $\mathbb{R}\times T^*\mathbb{R}^2$ admits a Poisson bracket on $\mathbb{R}\times T^*\mathbb{R}^2$ associated with its cosymplectic structure. In our chosen coordinates, the Poisson bracket related to $(\mathbb{R}\times T^*\mathbb{R}^2,\omega_{TB},\eta_{TB})$ has the form
\[
\Lambda_{TB}=\frac{\partial}{\partial \varphi}\wedge\frac{\partial }{\partial p_\varphi}+\frac{\partial}{\partial r}\wedge\frac{\partial }{\partial p_r}. 
\]

By Theorem \ref{Th::GradReduction}, one can project $\Lambda_{TB}$ onto the quotient space of the orbits of a gradient vector field satisfying given conditions. In fact,
\[
\pi_*\Lambda_{TB}=\frac{\partial}{\partial \varphi'}\wedge\frac{\partial }{\partial p_\varphi}+\frac{\partial}{\partial r}\wedge\frac{\partial }{\partial p_r},\qquad \pi_*R_{TB}=-\varpi\frac{\partial}{\partial \varphi'},\qquad \pi_*\nabla \Upsilon=0
\]
and
\begin{multline*}
\pi_*(R_{TB}+X_h)= p_r\frac{\partial}{\partial r}+\left(-\varpi+\frac{ p_\varphi}{r^2}\right)\frac{\partial}{\partial \varphi'}\\-\left(\frac{\mu(  r+  r_1 \cos \varphi')}{(r^2+2 r r_1 \cos \varphi'+r_1^2)^{3/2}}+\frac{(1-\mu)(  r-  r_2 \cos \varphi')}{  \left(r^2-2 r r_2 \cos \varphi'+r_2^2\right)^{3/2}}-\frac{p_\varphi^2}{  r^3}\right)\frac{\partial}{\partial p_r}\\-r_1rr_2 \sin \varphi'\left(\frac{ 1  }{\left(r^2+2 r r_1 \cos \varphi'+r_1^2\right)^{3/2}}+\frac{  1  }{\left(r^2-2 r r_2 \cos \varphi'+r_2^2\right)^{3/2}}\right)\frac{\partial}{\partial p_\varphi}.
\end{multline*}
Note that $\pi_*(R_{TB}+X_h)$ vanishes only at the image under $\pi$ of gradient relative equilibrium points. Moreover, the quotient manifold becomes a symplectic manifold and $\pi_*(R_{TB}+X_h)$ is Hamiltonian with  a Hamiltonian function
\[
k(r,\varphi',p_r,p_\varphi)=-\varpi p_\varphi +\frac{p_r^2}{2}+\frac{p_\varphi^2}{2r^2}-\frac{\mu}{[r^2+r^2_1+2rr_1\cos\varphi']^{1/2}}-\frac{1-\mu}{[r^2+r^2_2-2rr_2\cos\varphi']^{1/2}}.
\]

Physically, this is an autonomous Hamiltonian system obtained by fixing a coordinate system spinning around the centre of mass with angular frequency $\varpi$. Mathematically, $k$ is the Hamiltonian function predicted by Theorem \ref{Th::GradReduction}.

This example shows that the cosymplectic approach opens new possibilities, which must further be developed.

Finally, it is immediate that the cosymplectic-to-symplectic reduction of the Poisson bivector field $\Lambda_{TB}$ can be understood as a Poisson reduction by the distribution spanned by $\nabla \Upsilon$ (cf. \cite{MR86}). Nevertheless, the dynamics studied at the beginning is related to a vector field, $R_{TB}+X_h$, which is not Hamiltonian relative to the Poisson bivector and, therefore, its reduction is not the typical Marsden--Weinstein one. Moreover, our methods give the dynamics of the reduced system and other features that Poisson reduction does not provide.

\section{Conclusions and Outlook}

This work has presented a survey on cosymplectic geometry that was missing in the literature. Several technical requirements in classical results have been eliminated. In particular, cosymplectic momentum maps have been described in full generality to apply an extended cosymplectic Marsden--Weinstein reduction theorem. Previous results on cosymplectic geometry and Marsden--Weinstein reductions have been slightly generalised.  Moreover, our work develops possible generalisations of the energy-momentum method to the cosymplectic realm and it justifies the need for new approaches to the study of non-autonomous Hamiltonian systems and their reductions. To illustrate our theory, examples from several $n$-level quantum systems described by $t$-dependent Schr\"odinger equations have been considered. We have also devised a new type of cosymplectic-to-symplectic reduction that significantly differs from previous techniques \cite{Al89,LS15}. This has led to defining gradient equilibrium points and then applying our theory to describe a circular restricted three-body problem, its Lagrange points and Hill spheres.

In the future, we plan to extend the characterisation of relative equilibrium points to different types of stability, e.g. exponential stability. Additionally, we want to study the stability of gradient relative equilibrium points. We are also interested in studying types of $k$-symplectic and almost cosymplectic manifolds and associated Hamiltonian systems with our techniques.  We also plan to study singular cosymplectic reduction by means of orbifolds \cite{Sc83}. We hope that the latter will give a more general approach than the one used in \cite{LS15}. Moreover, this will allow for the use of the cosymplectic energy-momentum method for problems for which the level sets of the momentum map or its quotients cannot be considered as manifolds. Finally, we plan to apply all our novel techniques to new physical examples. For instance, we are considering the motion of a spinning diver. Biomechanical problems of this type have been recently analysed using different approaches \cite{DT15}.

Finally, we would like to analyse the extensions/modifications of energy-momentum methods to study field theories with physical applications. In particular, we are interested in analysing equations of magneto-hydrodynamic type and their relative stability \cite{HMRW85} through appropriate modifications our of techniques. In particular, we are analysing the use of Poisson \cite{HMRW85,MS88}, poly-Poisson \cite{Ma15} and $k$-poly(co)symplectic techniques \cite{LRVZ23} in energy-momentum and energy-Casimir methods.

\label{Sec::Conclusions}
\section*{Acknowledgements}

We acknowledge partial financial support from the 
project MINIATURA 5 of the Polish National Science Centre (NCN) under grant number Nr 2021/05/X/ST1/01797. J. de Lucas acknowledges a CRM-Simons professorship in the Centre de Recherchers Math\'ematiques (CRM) of the Universit\'e de Montr\'eal, during which, our new type of cosymplectic-to-symplectic reduction was found and this work was finished. J. de Lucas and B.M. Zawora would like to thank the members and staff of the CRM for their hospitality. We would like to thank prof. Grundland for drawing our attention to the work \cite{HMRW85} and useful discussions, which will be of relevance in our further research.

\pdfbookmark[1]{References}{ref}
\footnotesize\itemsep=0pt
{\providecommand{\eprint}[2][]{\href{http:farxiv.org/abs/#2}{arXiv:#2}}


\begin{thebibliography}{10}
\footnotesize\itemsep=0pt
\providecommand{\url}[1]{#1}
\providecommand{\urlprefix}{}
\providecommand{\eprint}[2][]{\href{https://arxiv.org/abs/#2}{arXiv:#2}}

\bibitem{AM78}
Abraham R., Marsden J.E., Foundations of mechanics, \textit{AMS Chelsea
  publishing}, Vol. 364, Benjamin/Cummings Pub. Co., 1978.
  \url{https://doi.org/10.1090/chel/364}

\bibitem{Al89}
Albert C., Le th\'{e}or\`eme de r\'{e}duction de {M}arsden-{W}einstein en
  g\'{e}om\'{e}trie cosymplectique et de contact, \textit{J. Geom. Phys.}
  \textbf{6} (1989), 627--649.
  \url{https://doi.org/10.1016/0393-0440(89)90029-6}

\bibitem{Ar89}
Arnold V.I., Mathematical Methods of Classical Mechanics, \textit{Graduate
  Texts in Mathematics}, Vol.~60, Springer, New York, 1989.
  \url{https://doi.org/10.1007/978-1-4757-2063-1}

\bibitem{BS97}
Baguis P., Stavracou T., Marsden-{W}einstein reduction on graded symplectic
  manifolds, \textit{J. Math. Phys.} \textbf{38} (1997), 1670--1684.
  \url{https://doi.org/10.1063/1.531902}

\bibitem{BBHLS15}
Ballesteros A., Blasco A., Herranz F.J., de~Lucas J., Sard\'{o}n C.,
  Lie-{H}amilton systems on the plane: properties, classification and
  applications, \textit{J. Differential Equations} \textbf{258} (2015),
  2873--2907.
  \url{https://doi.org/10.1016/j.jde.2014.12.031}

\bibitem{BG17}
Bruce A.J., Grabowska K., Grabowski J., Remarks on contact and {J}acobi
  geometry, \textit{SIGMA Symmetry Integrability Geom. Methods Appl.}
  \textbf{13} (2017), 059, 22.
  \url{https://doi.org/10.3842/SIGMA.2017.059}

\bibitem{Ca06}
Cannas~da Silva A., Lectures on symplectic geometry, \textit{Lecture Notes in
  Mathematics}, Vol. 1764, Springer-Verlag, Berlin, 2001.
  \url{https://doi.org/10.1007/978-3-540-45330-7}

\bibitem{CNY13}
Cappelletti-Montano B., de~Nicola A., Yudin I., A survey on cosymplectic
  geometry, \textit{Rev. Math. Phys.} \textbf{25} (2013), 1343002, 55.
  \url{https://doi.org/10.1142/S0129055X13430022}

\bibitem{CCJL19}
Cari\~{n}ena J.F., Clemente-Gallardo J., Jover-Galtier J.A., de~Lucas J.,
  Application of {L}ie systems to quantum mechanics: superposition rules, in
  Classical and quantum physics---60 years {A}lberto {I}bort {F}est---geometry,
  dynamics, and control, \textit{Springer Proc. Phys.}, Vol. 229, Springer,
  Cham, 2019, pp. 85--119.
  \url{https://doi.org/10.1007/978-3-030-24748-5_6}

\bibitem{CZ13}
Cattaneo A.S., Zambon M., A supergeometric approach to {P}oisson reduction,
  \textit{Comm. Math. Phys.} \textbf{318} (2013), 675--716.
  \url{https://doi.org/10.1007/s00220-013-1664-7}

\bibitem{LR89}
de~Le\'{o}n M., Rodrigues P.R., Methods of differential geometry in analytical
  mechanics, \textit{North-Holland Mathematics Studies}, Vol. 158,
  North-Holland Publishing Co., Amsterdam, 1989.
  \url{https://doi.org/10.1112/blms/23.1.105}

\bibitem{LS93}
de~Le\'{o}n M., Saralegi M., Cosymplectic reduction for singular momentum maps,
  \textit{J. Phys. A} \textbf{26} (1993), 5033--5043.
  \url{https://doi.org/10.1088/0305-4470/26/19/032}

\bibitem{LX22}
de~Lucas J., Rivas X., Contact {L}ie systems: theory and applications,
  \textit{arXiv:2207.04038v2\!\!} .
  \url{https://doi.org/10.1088/1751-8121/ace0e7}

\bibitem{LRVZ23}
de~Lucas J., Rivas X., Vilari\~no S., Zawora B.M., {On $k$-polycosymplectic Marsden--Weinstein reductions}, {\it J. Geom. Phys.} {\bf 191} (2023), 104899. \url{https://doi.org/10.1016/j.geomphys.2023.104899}

\bibitem{LS15}
de~Lucas J., Sard\'on C., Jacobi-{L}ie systems: theory, classification,
  \textit{Dyn. Sys. Diff. Eq. App. AIMS Proceedings}  (2015), 605--614.
  \url{https://doi.org/10.48550/arXiv.1412.0300} 

\bibitem{LS20}
de~Lucas J., Sard\'{o}n C., A guide to {L}ie systems with compatible geometric
  structures, World Scientific Publishing Co. Pte. Ltd., Singapore, 2020.
  \url{https://doi.org/10.1142/q0208}

\bibitem{LZ21}
de~Lucas J., Zawora B.M., A time-dependent energy-momentum method, \textit{J.
  Geom. Phys.} \textbf{170} (2021), 104364.
  \url{https://doi.org/10.1016/j.geomphys.2021.104364}


\bibitem{DT15}
Dullin H. R., Tong W., Twisting Somersault, \textit{SIAM J. Appl. Dyn. Syst.} \textbf{15} (2015).
\url{https://doi.org/10.1137/15M1055097}


\bibitem{GZ14}
Gao F.B., Zhang W., A study on periodic solutions for the circular restricted
  three-body problem, \textit{Astronomical J.} \textbf{148} (2014), 116.
  \url{https://doi.org/10.1088/0004-6256/148/6/116}

\bibitem{Go80}
Goldstein H., Classical mechanics, 2nd ed., Addison-Wesley Publishing Co.,
  Reading, Mass., 1980.

\bibitem{GO92}
Gonz\'{a}lez-L\'{o}pez A., Kamran N., Olver P.J., Lie algebras of vector fields
  in the real plane, \textit{Proc. London Math. Soc.} \textbf{64} (1992),.
  \url{https://doi.org/10.2307/2374757}

\bibitem{GG22}
Grabowski J., Grabowska K., A novel approach to contact {H}amiltonians and
  contact {H}amilton-{J}acobi theory, \textit{J. Phys A} \textbf{55} (2022),
  435204.
  \url{https://doi.org/10.48550/arXiv.2207.04484}
\bibitem{GS90}
Guillemin V., Sternberg S., Symplectic techniques in physics, Cambridge
  University Press, Cambridge, 1990.



\bibitem{HMRW85}
Holm D., Marsden J.E., Ratiu T. and Weinstein, { A
nonlinear stability of fluid and plasma equilibria}, 
{\it Phys. Rep.} \textbf{123} (1985) 1--116.
\url{https://doi.org/10.1016/0370-1573(85)90028-6}

\bibitem{Kh87}
Khatib O., A unified approach for motion and force control of robot
  manipulators: The operational space formulation, \textit{IEEE J. Robotics
  Autom.} \textbf{1} (1987), 43--53.
  \url{https://doi.org/10.1109/JRA.1987.1087068}

\bibitem{La09}
Lagrange J.L., M\'{e}canique analytique. {V}olume 2, Cambridge Library
  Collection, Cambridge University Press, Cambridge, 2009.\url{https://doi.org/10.1017/CBO9780511701795}

\bibitem{Le09}
Lee J.M., Manifolds and differential geometry, \textit{Graduate Studies in
  Mathematics}, Vol. 107, American Mathematical Society, Providence, 2009.
  \url{https://doi.org/10.1090/gsm/107}

\bibitem{Le13}
Lee J.M., Introduction to smooth manifolds, \textit{Graduate Texts in
  Mathematics}, Vol. 218, Springer, New York, 2013.
  \url{https://doi.org/10.1007/978-1-4419-9982-5}

\bibitem{Li59}
Libermann P., Sur les automorphismes infinit\'{e}simaux des structures
  symplectiques et des structures de contact, in Colloque {G}\'{e}om. {D}iff.
  {G}lobale ({B}ruxelles, 1958), Centre Belge Rech. Math., Louvain, 1959,
  37--59.

\bibitem{Li62}
Libermann P., Sur quelques exemples de structures pfaffiennes et presque
  cosymplectiques, \textit{Ann. Mat. Pura Appl. (4)} \textbf{60} (1962),
  153--172.
  \url{https://doi.org/10.1007/BF02412771}

\bibitem{LM87}
Libermann P., Marle C.M., Symplectic geometry and analytical mechanics,
  \textit{Mathematics and its Applications}, Vol.~35, D. Reidel Publishing Co.,
  Dordrecht, 1987, translated from the French by Bertram Eugene Schwarzbach.
  \url{https://doi.org/10.1007/978-94-009-3807-6}

\bibitem{Li75}
Lichenerowicz A., Vari\'et\'es canoniques i transformations canoniques,
  \textit{C.R.A.S.}  (1975), 280.

\bibitem{Ma09}
Marle C.M., The inception of symplectic geometry: the works of {L}agrange and
  {P}oisson during the years 1808--1810, \textit{Lett. Math. Phys.} \textbf{90}
  (2009), 3--21.
  \url{https://doi.org/10.1007/s11005-009-0347-y}

\bibitem{MRSV15}
Marrero J.C., Rom\'{a}n-Roy N., Salgado M., Vilari\~{n}o S., Reduction of
  polysymplectic manifolds, \textit{J. Phys. A} \textbf{48} (2015), 055206, 43.
  \url{https://doi.org/10.1088/1751-8113/48/5/055206}

\bibitem{MR86}
Marsden J.E., Ratiu T., Reduction of {P}oisson manifolds, \textit{Lett. Math.
  Phys.} \textbf{11} (1986), 161--169.
  \url{https://doi.org/10.1007/BF00398428}

\bibitem{RM99}
Marsden J.E., Ratiu T.S., Introduction to mechanics and symmetry, \textit{Texts
  in Applied Mathematics}, Vol.~17, Springer-Verlag, New York, 1999.
  \url{https://doi.org/10.1007/978-0-387-21792-5}

\bibitem{MS88}
Marsden J.E., Simo J.C., The energy momentum method, \textit{Act. Acad. Sci.
  Tau.} \textbf{1} (1988), 245--268.
  \url{https://doi.org/10.1017/CBO9780511624001.006}

\bibitem{MW74}
Marsden J.E., Weinstein A., Reduction of symplectic manifolds with symmetry,
  \textit{Rep. Math. Phys.} \textbf{5} (1974), 121--130.
  \url{https://doi.org/10.1016/0034-4877(74)90021-4}

\bibitem{Ma15}
Martinez N., Poly-symplectic groupoids and poly-Poisson structures, \textit{Lett. Math. Phys.} \textbf{105} (2015), 693--721. \url{https://doi.org/10.1007/s11005-015-0746-1}

\bibitem{Sc83}
Scott P., The geometries of {$3$}-manifolds, \textit{Bull. London Math. Soc.}
  \textbf{15} (1983), 401--487.
  \url{https://doi.org/10.1112/BLMS%2F15.5.401}

\bibitem{Va94}
Vaisman I., Lectures on the geometry of Poisson manifolds, \textit{Progress in
  Mathematics}, Vol. 118, Birkh\"{a}user Verlag, Basel, 1994.
  \url{https://doi.org/10.1007/978-3-0348-8495-2}

\bibitem{Vi02}
Vidyasagar M., Nonlinear systems analysis, \textit{Classics in Applied
  Mathematics}, Vol.~42, Society for Industrial and Applied Mathematics (SIAM),
  Philadelphia, 2002.
  \url{https://doi.org/10.1137/1.9780898719185.fm}

\bibitem{Za21}
Zawora B., A time-dependent energy-momentum method, \textit{Master's thesis},
  University of Warsaw, Faculty of Physics, 2021\!\!.

\end{thebibliography}
\end{document}